\documentclass[11pt]{article}
\pdfoutput=1                    % Need this for accurate compiling of tex files when submitting in ArXiv. %

\usepackage[utf8]{inputenc}
\usepackage{amsfonts}
\usepackage[symbol]{footmisc}
\usepackage{psfrag,epsf}
\usepackage{authblk}
\usepackage{caption}
\usepackage{fullpage}
\usepackage{graphicx, color, subcaption, float, hyperref, enumerate}
\hypersetup{
    colorlinks=true,
    linkcolor=blue,
    citecolor=blue,
    urlcolor=blue}
\usepackage{amsmath,amssymb,mathtools,amsthm, bm, bbm}
\usepackage[ruled,vlined]{algorithm2e}

\usepackage[round]{natbib}
\allowdisplaybreaks

\usepackage{amsmath, amsthm, amssymb, amsthm}
\usepackage[flushleft]{threeparttable}
\usepackage{tabu,multirow}
\usepackage{colortbl,latexsym,mathrsfs}
\usepackage{color}
\usepackage{etoolbox}
\usepackage{float}
\usepackage{booktabs}
\usepackage[dvipsnames,svgnames,x11names]{xcolor}
\usepackage{setspace}
\usepackage{nameref}
\usepackage{listings}
\usepackage{bbm}
\usepackage{babel}
\newtheorem{theorem}{Theorem}

\newtheorem{lemma}{Lemma}
\newtheorem{corollary}{Corollary}
\newtheorem{proposition}{Proposition}
\theoremstyle{definition}
\newtheorem{assumption}{Assumption}
\newtheorem{remark}{Remark}

\newcommand{\ind}{\perp \!\!\! \perp} %independence
\newcommand{\cvP}{\xrightarrow[]{\bbP_{\calD}}}
\newcommand{\dd}{\mathrm{d}} %integral symbol
\newcommand{\iid}{\overset{\text{i.i.d.}}\sim} % iid symbol
\newcommand{\Op}{O_{\mathbb{P}}}
\newcommand{\op}{o_{\mathbb{P}}}
\newcommand{\TV}{\mathrm{TV}}

\newcommand{\CF}{\mathrm{CF}}

\def\bbE{\mathbb{E}}
\def\bbL{\mathbb{L}}
\def\bbK{\mathbb{K}}
\def\bbP{\mathbb{P}}
\def\bbR{\mathbb{R}}
\def\bbZ{\mathbb{Z}}
\def\bX{\mathbf{X}}
\def\bZ{\mathbf{Z}}
\def\bD{{\bf D}}
\def\bW{{\bf W}}
\def\bx{\mathbf{x}}
\def\phat{\widehat{p}}
\def\ptilde{\widetilde{p}}
\def\muhat{\widehat{\mu}}
\newcommand{\Var}{\mathrm{Var}}
\newcommand{\Cov}{\mathrm{Cov}}
\def\bbS{\mathbb{S}}
\def\calD{\mathcal{D}}
\def\calI{\mathcal{I}}

\def\calN{\mathcal{N}}
\def\calP{\mathcal{P}}

\def\calT{\mathcal{T}}

\def\nk{n_{\bbK}}
\def\Pm{\Pi_{m_1}}
\def\reg{\text{reg}}
\def\-{\text{--}}
\def\m{\underbar{m}_1}
\def\bm{\underbar{m}}
\def\r{\underbar{r}_1}
\def\e{\underbar{e}}
\def\ubm{\underbar{m}_0}
\def\ubr{\underbar{r}_0}
\def\br{\underbar{r}}

\def\bb{\underbar{b}}
\def\b{\underbar{b}_1}
\def\ubb{\underbar{b}_0}

\def\rt{\underbar{r}_t}
\def\mt{\underbar{m}_t}
\def\PD{\bbP_\calD}
\def\tATE{\Delta^\dagger} % true ATE
\def\rATE{\Delta} % random ATE
\def\pATE{\Pi_\Delta} % posterior of ATE
\def\d{\dagger}

\usepackage{microtype}
\graphicspath{ {Pictures/} }

\begin{document}
 \date{}
\title{Bayesian semiparametric causal inference: Targeted doubly robust estimation of treatment effects}

\author[1]{Gözde Sert}
\author[1]{Abhishek Chakrabortty}
\author[1,*]{Anirban Bhattacharya}
\affil[1]{Department of Statistics, Texas A\&M University}
\makeatletter
\renewcommand\AB@affilsepx{\\ \protect\Affilfont}
\renewcommand\Authsep{, }
\renewcommand\Authand{, }
\renewcommand\Authands{, }
\makeatother

\maketitle
\footnotetext[1]{Corresponding author.}
\footnotetext{{\it Email addresses:} \hyperlink{gozdesert@stat.tamu.edu}{gozdesert@stat.tamu.edu} (Gözde Sert), \hyperlink{abhishek@stat.tamu.edu}{abhishek@stat.tamu.edu} (Abhishek Chakrabortty), \hyperlink{anirbanb@stat.tamu.edu}{anirbanb@stat.tamu.edu} (Anirban Bhattacharya).}

\begin{abstract}We propose a semiparametric Bayesian methodology for estimating the average treatment effect (ATE) within the potential outcomes framework using observational data with high-dimensional nuisance parameters. Our method introduces a Bayesian debiasing procedure that corrects for bias arising from nuisance estimation and employs a targeted modeling strategy based on summary statistics rather than the full data. These summary statistics are identified in a debiased manner, enabling the estimation of nuisance bias via weighted observables and facilitating hierarchical learning of the ATE. By combining debiasing with sample splitting, our approach separates nuisance estimation from inference on the target parameter, reducing sensitivity to nuisance model specification. We establish that, under mild conditions, the marginal posterior for the ATE satisfies a Bernstein-von Mises theorem when both nuisance models are correctly specified and remains consistent and robust when only one is correct, achieving Bayesian double robustness. This ensures asymptotic efficiency and frequentist validity. Extensive simulations confirm the theoretical results, demonstrating accurate point estimation and credible intervals with nominal coverage, even in high-dimensional settings. The proposed framework can also be extended to other causal estimands, and its key principles offer a general foundation for advancing Bayesian semiparametric inference more broadly.
\end{abstract}


\section{Introduction}\label{sec_intro}
Inferring the causal effect of a treatment or exposure is central to many scientific disciplines. While randomized controlled trials are the gold standard for causal estimation, they are often infeasible due to ethical, logistical, or financial constraints. A common alternative is to use {\it observational} data, which is typically easier to obtain, but also requires careful methodology to handle potential confounding and high dimensionality issues, while ensuring robust (unbiased) estimation of causal estimands. Among these, the {\it average treatment effect} (ATE) is a key popular estimand, measuring the treatment's overall
causal impact, and is widely adopted in various scientific disciplines.

Estimation of the ATE is naturally linked to semiparametric inference, as its identification involves infinite-dimensional nuisance parameters \citep{bang2005doubly}. Most existing approaches are frequentist, such as propensity score adjustment or matching \citep{rosenbaum1983central, rosenbaum1984reducing}, and doubly robust (DR) estimators \citep{robins1994estimation, robins1995semiparametric}. Recently, {\it Bayesian} semiparametric methods for ATE estimation have gained attention \citep{ray2019debiased, ray2020semiparametric, hahn2020bayesian, antonelli2022causal, linero2022, luo2023semiparametric, breunig2025double}. Traditional Bayesian methods marginalize out nuisance parameters to obtain a posterior of the target parameter and can achieve desirable contraction rates \citep{ghosal2017fundamentals}. However, strong regularization often induces nuisance estimation bias \citep{bickel2012semiparametric, rivoirard2012bernstein, castillo2015bernstein} that jeopardizes the validity of Bayesian inference for low-dimensional targets, such as the ATE.

Several strategies have been proposed to mitigate this bias. One line of research modifies or tailors priors to incorporate the propensity score and better align the prior with the semiparametric model structure \citep{ray2019debiased, ray2020semiparametric}. Another applies posterior corrections or influence function \citep{hahn98} (IF)-based updates \citep{breunig2025double, yiu2025}. A related method by \citet{antonelli2022causal} constructs a posterior for the ATE by plugging nuisance posterior samples into the IF, followed by an additional variance correction for valid inference.

Building on recent advances, we propose the {\it doubly robust debiased Bayesian (DRDB) procedure}, which provides a principled and scalable solution to nuisance bias in high-dimensional or complex settings. DRDB departs from existing Bayesian methods in two key ways. First, it adopts a {\it targeted} modeling strategy that focuses on {\it summary statistics} informative about the ATE, rather than the full data distribution. Second, it introduces a {\it Bayesian debiasing} mechanism that {\it learns nuisance bias} directly from data, eliminating the need for prior modification or post hoc correction \citep{ray2020semiparametric, breunig2025double, yiu2025}. By decoupling inference for the ATE from nuisance estimation, DRDB ensures robustness of the marginal posterior and serves as a {\it Bayesian analogue} of the frequentist double machine learning framework \citep{chernozhukov2018double}, maintaining validity even under high-dimensional and/or misspecified models.

A prominent usage of summary statistics in Bayesian inference appears in the approximate Bayesian computation (ABC) literature to mitigate issues with a low acceptance rate \citep{drovandi2015}. DRDB instead leverages them in a targeted manner to separate the ATE from nuisance bias. Its key component is a {\it retargeting step} that models the nuisance bias using weighted observables (an idea akin to importance sampling) which naturally incorporates the {\it propensity score} (PS) into the Bayesian framework. Although the PS plays a central role in the frequentist literature on DR estimation \citep{robins1994estimation, robins1995semiparametric, bang2005doubly}, it has lacked a principled Bayesian counterpart \citep[Section 5]{li2023bayesian}. DRDB fills this gap by {\it integrating the PS  seamlessly} through its debiasing mechanism. Related work by \cite{sert2025} develops Bayesian inference via summary statistics for semi-supervised learning, which motivates the construction of DRDB. However, DRDB differs in two key respects: (i) DRDB uses a hierarchical model to learn the ATE directly, without relying on independence between data subsets, and (ii) it introduces a retargeting mechanism to identify and estimate nuisance bias appropriately.

Building on this bias estimation, DRDB integrates the bias into a {\it hierarchical} Bayesian framework: The posterior for the bias informs a conditional likelihood for the ATE, whose integration yields a valid marginal posterior for the ATE (see Equation~\eqref{eqn_posterior_for_mu_for_S}). Another salient feature of DRDB is its use of {\it sample-splitting} and cross-fitting (CF) \citep{chernozhukov2018double}. {\it Beyond} their traditional role in technical aspects, DRDB uses them as critical {\it methodological} tools to validate the debiasing step and {\it decouple} nuisance estimation from target inference. DRDB employs randomized splitting to obtain multiple subposteriors and \emph{aggregates} them using a consensus Monte Carlo–type scheme \citep{scott2022bayes}, producing a posterior that efficiently utilizes the {\it entire} data (see Section~\ref{sec_DRDB_extended_ATE}).

DRDB establishes a semiparametric {\it Bernstein-von Mises (BvM) result} for the marginal posterior of the ATE (Theorems~\ref{thm_BvM_mu1} and \ref{thm_BvM_ATE}): When both nuisance models are well-specified and their posteriors contract at rates whose {\it product} is $o(n^{-1/2})$, the posterior concentrates around the true ATE at the parametric rate and is asymptotically Gaussian. In this case, the posterior mean is an asymptotically {\it efficient} estimator of the true ATE, converging at a $\sqrt{n}$-rate with asymptotic variance that achieves the semiparametric efficiency bound \citep{hahn98}. Notably, the DRDB posterior depends on the nuisance posteriors only through their asymptotic limits, underscoring its robustness to nuisance modeling. Moreover, DRDB satisfies {\it Bayesian double robustness}: when only one nuisance model is well-specified (consistently estimated, while the other may be misspecified or slowly estimated), the {\it posterior remains consistent} for the ATE, contracting at the rate of the well-specified nuisance, extending the frequentist DR principle \citep{bang2005doubly} to Bayesian inference (posteriors).

Finally, the key principles of DRDB (its debiasing mechanism, targeted modeling, and hierarchical learning strategy) extend naturally beyond the ATE, providing valid Bayesian inference for a broad class of causal estimands, including the average treatment effect on the treated (ATT), that on the control (ATC), and subgroup-specific effects. For clarity and brevity, the detailed extension of DRDB to these general estimands is presented in Section~\ref{sec_DRDB_conditional} of the \hyperref[sec_supplementary]{Supplementary Material}.

The rest of this paper is organized as follows. Section~\ref{sec_setup} introduces the basic setup and preliminaries. Section~\ref{sec_methodology} develops our proposed DRDB methodology, first for one counterfactual mean (Section \ref{sec_DRDB_for_mu1}),
then for the ATE (Section \ref{sec_DRDB_extended_ATE}). Section~\ref{sec_theory} presents the technical details of the DRDB posterior, and the main theoretical results, including BvM results and Bayesian double robustness. Section~\ref{sec_numerics} reports finite-sample performance via simulations and data analysis. Section~\ref{sec_conclusion} provides a concluding discussion. Extensions of our methodology, additional simulation results, and proofs and technical details are deferred to the \hyperref[sec_supplementary]{Supplementary Material} (Sections~\ref{sec_DRDB_conditional}-\ref{supp_sec_proof_of_preliminary}).

\section{Setup and preliminary causal assumptions}\label{sec_setup}
\subsection{Data and notation}
Let $T \in \{0, 1\}$ denote a binary treatment indicator; $Y \in \bbR$ denote the {\it observed} outcome, defined as: $Y = TY(1) + (1-T)Y(0)$, where $\{ Y(1), Y(0)\}$ are the {\it potential outcomes} \citep{rubin1974estimating, imbens2015causal} under treatment ($T = 1$) and control ($T = 0$), respectively (i.e., $Y(t)$ is the outcome that would have been observed if $T = t$, possibly contrary to fact); and $\bX \in \bbR^p$ denote the vector of covariates (or potential {\it confounders}). The {\it observed data} $\calD$ consists of independent and identically distributed (i.i.d) observations $\bZ_1, \dots, \bZ_n$ of the random variable $\bZ := (Y,\bX, T)$ with support $\mathcal{Y} \times \mathcal{X} \times \{0, 1\}$ and underlying joint probability distribution (p.d.) $\bbP_{\bZ}$. Also, the setting is throughout allowed to be (possibly) {\it high dimensional} (i.e., $p$ is allowed to grow with $n$).

Let $U$ be a random object with an underlying p.d. $\bbP_{U}$, and $f$ be a measurable $\bbR$-valued function of $U$. The expectation of $f(U)$ is defined as $\bbE_{U}\{f(U)\}:= \int f(u) d\bbP_{U}(u)$, whenever it exists. For any $d\geq 1$, $L_d(\bbP_{U})$ denotes the space of all $\bbR$-valued measurable functions of $U$ equipped with the norm $\|f \|_{L_d(\bbP_{U})}:= [\bbE_{U}\{f(U)^d\}]^{1/d}$. {\it We adopt the following Bayesian notation throughout: for a generic random object $\theta$, $\Pi_{\theta}$ denotes its posterior, $\underline{\theta}$ a posterior sample, and $\theta^\dagger$ its true value}.

\subsection{Identification}\label{subsec_identification}
The parameter of interest is the {\it average treatment effect (ATE), defined as:} $\Delta^\dagger:= \mu^\dagger(1) - \mu^\dagger(0)$, with $\mu^\dagger(t):= \bbE_{\bbZ}\{Y(t)\}$ for $t \in \{0, 1\}$, where the expectation is taken under the true p.d. $\bbP_{\bbZ}$ of $\bbZ := \{Y(1), Y(0), \bX, T\}$. Since $\{Y(1), Y(0)\}$ cannot be jointly observed in the data, we impose standard causal assumptions to identify $\tATE$ from the available data $\calD$ \citep{rosenbaum1984reducing}.

\begin{assumption}[Causal assumptions]
\label{assumptions_standard_causal_assump} (i) (Ignorability) $T \ind \{Y(1), Y(0) \} | \bX$. (ii) (Positivity) Let $e^\dagger(\bX) := \bbP(T = 1  | \bX)$ be the propensity score. Then, $\ell \leq e(\bX) \leq 1 - \ell$, for some constant $\ell >0$.
\end{assumption}

\noindent Assumption~\ref{assumptions_standard_causal_assump}(a), known as {\it no unmeasured confounding (NUC)}, posits that the set of observed covariates captures all confounding factors affecting both treatment assignment and the potential outcomes. Assumption~\ref{assumptions_standard_causal_assump}(b) imposes an {\it overlap} condition, ensuring that $\bX$ in the treatment groups (i.e., $\bX \,| \, T = t$) share sufficient common support for valid comparisons \citep{imbens2015causal}.

\paragraph{Regression-based identification.} Define $m^\d_t(\bX) \equiv m^\d(\bX, t) := \bbE(Y(t) \mid \bX)$ as the {\it regression function} for treatment $t \in \{0,1\}$. Under Assumption~\ref{assumptions_standard_causal_assump} (ignorability), it can be equivalently written as: $m_t^\d(\bX) = \bbE(Y \mid \bX, T = t)$. The ATE is then {\it identified} using the law of iterated expectations:

\begin{equation}
\tATE ~\equiv~ \tATE(\overrightarrow{m}^\d,\bbP_{\bX}) ~=~  \tATE(m_1^\d, m_0^\d,\bbP_{\bX}) ~:=~ \bbE_{\bX}\{m_1^\d(\bX) - m_0^\d(\bX)\}, \label{eqn_regression_identification_ATE}
\end{equation}
where we note that each $m_t^\d(\cdot)$ is {\it estimable} via a regression in the {\it observable} data on: $(Y, \bX) \hspace{0.01in}| T = t$. Hence, the ATE $\tATE$ is a {\it functional} of both $\mathbb{P}_{\bX}$ and the {\it nuisance functions} $\overrightarrow{m}^\d \equiv \overrightarrow{m}^\d(\cdot)$, with $\overrightarrow{m}^\d:= (m_1^\d, m_0^\d)$, and this identification serves as a {\it foundation for our approach} to estimating $\tATE$.

\section{Methodology}\label{sec_methodology}
Motivated by \eqref{eqn_regression_identification_ATE}, we propose a doubly robust debiased Bayesian (DRDB) procedure for estimating the ATE. To clarify the main steps of DRDB, we first present the methodology for a {\it single-arm:} $\mu_1^\d\equiv\mu^\d(1) = \mathbb{E}[Y(1)]$, and thereafter, extend it to the ATE in Section \ref{sec_DRDB_extended_ATE}. Importantly, estimating $\mu_1^\d$ is an interesting and non-trivial
problem in its own right, as it corresponds to the mean of an outcome that is missing at random (MAR) within the {\it missing data} framework \citep{tsiatis2007semiparametric}.

Let $\bbK \ge 2$ be a {\it fixed} integer. We randomly {\it split} $\calD$ into $\bbK$ disjoint subsets $\{\calD_k\}_{k=1}^\bbK$, each of equal size $\nk := n/\bbK$, assuming without loss of generality that $n$ is divisible by $\bbK$. The corresponding index sets are denoted by $\{\calI_k\}_{k=1}^\bbK$. For each $k \in \{1, \dots, \bbK\}$, define $\calD_k^\- := \calD \setminus \calD_k$, which has size $\nk^\- := n - \nk$ and index set $\calI_k^\-$. Let $(S, S^\-) := (\calD_k, \calD_k^\-)$ denote a  generic pair of \textit{test} and \textit{training} datasets with corresponding index sets $(\calI, \calI^\-)$ for some $k \in \{1, \dots, \bbK\}$. For $t \in {0,1}$, let $(S_t, S_t^\-)$ denote the {\it subgroups} of $(S, S^\-)$ corresponding to treatment level $T = t$, so $S_1$ and $S_1^\-$ represent the respective {\it treated} subgroups, and $S_0$ and $S_0^\-$ represent the {\it control} subgroups. By construction, $S$ and $S^\-$ are {\it independent} ($S \!\ind\! S^\-$), which is both {\it crucial} and {\it necessary} for the DRDB approach.

\paragraph{Motivating the DRDB procedure.} An intuitive approach to estimating $\mu_1^\d$, motivated by \eqref{eqn_regression_identification_ATE}, is to use a regression-based Bayesian ({\tt BREG}) procedure: Suppose the unknown nuisance function $m_1^\d(\cdot)$ is {\it learned from $S^\-$} via {\it any} suitable Bayesian regression method--parametric (like Bayesian ridge regression via Gaussian priors, or high dimensional sparse Bayesian linear regression
using spike-and-slab type priors \citep{johnson2012bayesian}) or nonparametric (such as Gaussian process regression \citep{williams1998prediction} or Bayesian
additive regression trees (BART) \citep{bart2010})--yielding a {\it posterior $\Pm$ for $m_1$}. For a sample $\m \sim \Pm$, one can treat $\{\m(\bX_i)\}_{i \in \calI}$ as {\it derived i.i.d. samples in $S$}, targeting $\mu_1^\d$ through their mean. A standard Bayesian analysis, specifying a likelihood for this data and a prior on model parameters, then yields a posterior $\Pi_{\mathrm{reg}}$ for $\mu_1$.

Despite its intuitive appeal, BREG is highly {\it sensitive} to the quality of nuisance estimation: A misspecified nuisance model leads to an inconsistent posterior $\Pi_{\reg}$ for $\mu_1$. Even with a correctly specified nuisance model, the posterior's {\it first-order} properties, such as its rate and shape, are {\it strongly}
determined by the nuisance estimation bias: $\bbE_\bX\{\m(\bX) - m_1^\d(\bX) | \m\}$. This makes the posterior overly dependent on the behavior of $\Pm$ and the choice of regression method, which in turn requires restrictive conditions to control the bias (e.g., in high dimensions) for achieving BvM-type results. These limitations motivate the key principles of our DRDB approach, which systematically eliminates this nuisance estimation bias within a Bayesian likelihood framework.

\subsection{Doubly robust debiased Bayesian (DRDB) procedure for \texorpdfstring{$\mu^\d(1)$}
{mu(1)}}\label{sec_DRDB_for_mu1}

DRDB is fundamentally a two-step approach. First, a Bayesian debiasing step learns and corrects for nuisance estimation bias within a Bayesian framework via a retargeting method. Second, a hierarchical learning framework learns the parameter of interest, $\mu_1^\d \equiv \mu^\d(1)$, after this bias has been addressed. We detail the steps in subsequent sections.

Adopting the notation from Section~\ref{sec_methodology}, let $(S, S^\-)$ be a pair of test and training datasets. Assume the nuisance posterior $\Pm \equiv \Pm(\cdot; S^\-)$ for $m_1$ is obtained from $S^\-$ as before.

\paragraph{Debiasing step.}  Let $\m \sim \Pm \equiv \Pm(\cdot; S^\-)$ be {\it one} sample independent (by design) of $S$. Using the regression-based representation of $\mu_1^\d$
in given \eqref{eqn_regression_identification_ATE}, we obtain the {\it debiased identification} of $\mu_1^\d$:
\begin{equation}
    \begin{aligned}
\mu_1^\d & ~=~ \bbE_{\bX \in S}\{m_1^\d(\bX) - \m(\bX) | \m\} +  \bbE_{\bX \in S}\{\m(\bX) | \m\} \\
& ~=:~ b^\d(\m) + \bbE_{\bX \in S}\{\m(\bX) |\m\}.\label{eqn_debiased_represent_first_step}
\end{aligned}
\end{equation}

The term $b^\d(\m)$ captures the nuisance estimation {\it bias} from replacing the true $m_1^\d$ with a random sample $\m$. This bias is the primary source of the limitations of {\tt BREG} and serves as the central target of our {\it Bayesian debiasing} strategy. Its analysis and the validity of the debiased decomposition in \eqref{eqn_debiased_represent_first_step} crucially rely on the {\it independence condition} that ensures the distribution of $\bX \in S$ in $\m(\bX)$ is {\it unaffected} by that of $\m \sim {\Pm(\cdot;S^\-)}$ since $S^\- \!\ind\! S$. To further analyze $b^\d(\m)$, we write it as:
\begin{align}
    b^\d(\m)& ~=~ \bbE_{\bX \in S}[\bbE\{Y(1) - \m(\bX) \mid \bX, T = 1\} \mid  \m].\label{eqn_bm1_with_Y1_and_m1}
\end{align}
This formulation implies that if $Y(1)$ and $\bX$ were observed for all units in $S$, one could directly estimate $b^\d(\m)$ from $S$. However, both $\{Y(1), \bX\}$ are {\it only} observed in the {\it treated subgroup:} $S_1$. Moreover, given $\m \sim \Pm$, the observables $\{Y - \m(\bX)\} \in S_1$ target $\bbE_{(Y, \bX) | T = 1}\{Y - \m(\bX)\}$, rather than the desired bias $b^\d(\m)$. To correct this discrepancy, a {\it retargeting} step is required in which the (derived) observations $\{Y - \m(\bX)\} \in S_1$ are {\it reweighted} using a {\it density ratio} function. This adjustment ensures that the distribution of the weighted observations {\it aligns} with that of $(Y, \bX)$ in the whole population, rather than the conditional distribution given $T = 1$.

\paragraph{Retargeting bias via weighting.} Let $r_1^\d(\bX) := f(\bX)/f_1(\bX | T = 1)$ be the {\it density ratio} function, where $f(\cdot)$ is the density function (pdf) of $\bX$ and $f_1(\cdot)$ is the conditional pdf given $T = 1$. Given $\m \sim \Pm$, we define the {\it weighted} observations $r_1^\d(\bX)\{Y - \m(\bX)\}$ in $S_1$ and observe that:
\begin{equation}
\begin{aligned}
\bbE_{(Y, \bX) | T = 1}[r_1^\d(\bX)\{Y-\m(\bX)\}  \mid  T = 1] & ~=~ \bbE_{\bX}[\bbE\{Y-\m(\bX)\mid \bX, T = 1\}]\\ & ~=~ b^\d(\m) ~\; \equiv ~\; b^\d(\m, r_1^\d). \label{eqn_weighted_expression_of_bias1}
\end{aligned}
\end{equation}
This derivation shows that {\it unbiasedly} estimating the bias requires using the {\it weighted} observables $r_1^\d(\bX)\{Y - \m(\bX)\}$ in $S_1$. It also clarifies that the bias $b^\d(\m) \equiv b^\d(\m, r_1^\d)$ should be viewed as a {\it functional} of {\it two} nuisances:
$\m$ and $r_1^\d \equiv r_1^\d(\cdot)$. To learn $b^\d(\m, r_1^\d)$ from $S_1$, we must first estimate $r_1^\d$ by deriving a posterior from $S^\-$. This leads to the final analysis of the bias $b^\d(\m, r_1^\d)$.

Let $\r$ be {\it one} sample from the posterior $\Pi_{r_1} \equiv \Pi_{r_1}(\cdot; S^\-)$ of $r_1 \equiv r_1(\cdot)$, which we derive from a Bayesian {\it binary regression} method (e.g., Bayesian logistic regression or BART \citep{bart2010}), as detailed in Remark~\ref{remark_post_for_r_and_PS}. Since $\r$ is independent of $S$, substituting it into \eqref{eqn_weighted_expression_of_bias1} yields:
\begin{align}
b^\d(\m, r_1^\d)  ~=~  b^\d(\m,  \r) + \bbE_{\bX}[\{r_1^\d(\bX) - \r(\bX)\}\{m_1^\d(\bX) - \m(\bX)\} | \m,  \r]. \label{eqn_debiased_rep_bm}
\end{align}
\eqref{eqn_debiased_rep_bm} provides a full characterization of the bias.
The second term in \eqref{eqn_debiased_rep_bm}: $\Gamma^\d(\m, \r) := \bbE[\{r_1^\d(\bX) - \r(\bX)\} \{m_1^\d(\bX) - \m(\bX)\} |\m,  \r]$ is a {\it second-order} `bias of bias' (or `drift') term, arising from the {\it product} of the estimation errors for $r_1^\d$ and $m_1^\d$. We subsequently focus on {\it modeling and correcting the more tractable, first-order, bias} $b^\d(\m,  \r)$. The second-order term $\Gamma^\d(\m, \r)$, while accounted for in the
theoretical analysis of our eventual posterior, is not the primary debiasing target.

\paragraph{Targeted modeling strategy for bias.}
Given $\m \sim \Pi_{m_1}$ and $\r \sim \Pi_{r_1}$, the bias $b^\d(\m, \r) = \bbE_{S_1}[\r(\bX)\{Y - \m(\bX)\} | \m, \r]$ can be viewed as a {\it functional} of the underlying distribution of $S_1$, specifically, relying on the {\it summary statistic} of the weighted observables $\r(\bX)\{Y - \m(\bX)\} \in S_1$. We can then construct a {\it working} likelihood based on these i.i.d. observables in $S_1$ and place a prior on the model parameters, yielding {\it a posterior $\Pi_{b_1}$ for $b_1 \equiv b(\m, \r)$}, as detailed in Proposition~\ref{prop_posterior_b1}.

A defining feature of the DRDB procedure, beyond its debiasing mechanism, is the {\it targeted} use of data. While traditional methods model the entire data \citep{ray2019debiased, ray2020semiparametric, breunig2025double}, DRDB exclusively targets the parameters directly informative for $\mu_1^\d$. This targeted modeling strategy, combined with debiasing, forms the core of our DRDB approach, distinguishing it from existing Bayesian methodologies. Building on these, we next introduce the {\it hierarchical learning framework}, which facilitates a {\it construction} of a valid marginal posterior for $\mu_1^\d$, by leveraging in a novel way the conventional integral representation of the marginal posterior.

\begin{remark}[Key \emph{methodological} role of the data splitting]\label{remark_data_splitting} A core characteristic of DRDB is its strict separation of nuisance estimation from target inference using independent data sources ($S \!\ind\! S^\-$). This independence is methodologically crucial, validating our debiased representation, and equally importantly, validating the construction of the respective likelihoods, enabling our targeted modeling for the bias. This is distinct from earlier approaches that used independent auxiliary data mainly to address theoretical challenges, and incorporate the propensity score \citep{ray2019debiased, ray2020semiparametric, breunig2025double}. In our framework, independence is not a mere technical tactic, but the central mechanism driving the DRDB procedure.
\end{remark}

\paragraph{Hierarchical learning strategy.} If one had access to a joint posterior for $\{\mu_1, b(\m)\}$, which would be the case in a traditional Bayesian framework, the marginal posterior for $\mu_1$ would be obtained by integrating out $b(\m)$: $[\mu_1| S] = \int [\mu_1 | b(\m), S] \hspace{0.05cm}[b(\m) | S_1] \hspace{0.05cm}\dd b(\m)$. We instead take a fundamentally different perspective, using the right-hand side of this representation itself as the {\it defining principle} for constructing a posterior for $\mu_1$. In this formulation, $[\mu_1 | b(\m), S]$ and $[b(\m)| S]$ are individual building blocks which we aim to learn separately under our targeted modeling strategy, and mixing (integrating) over the uncertainty in $b(\m)$ defines an idealized marginal posterior for $\mu_1$. Since the actual bias $b(\m)$ involves the intractable second-order term in \eqref{eqn_debiased_rep_bm}, we achieve a key simplification by replacing the mixing variable $b(\m)$ with its first-order {\it proxy} $b_1 \equiv b(\m, \r)$. Thus, we replace $[b(\m)| S]$ with the posterior $\Pi_{b_1}$ for $b_1$ alluded to earlier, and construct a conditional posterior $\Pi_{\mu_1|b_1}$ for $\mu_1|b_1$; see discussion after Remark \ref{rem:hierarchical}. These switches allow us to mix over the proxy variable $b_1$, and define a valid probability measure
\begin{align}
   \hspace{-0.2cm} \left[\mu_1 \mid S\right] & ~:=~ \int [\mu_1 \mid b_1, S] \hspace{0.05cm}[b_1 \mid S] \hspace{0.05cm}\dd \hspace{0.02cm}b_1 ~=~ \int [\mu_1 \mid b_1, S] \hspace{0.05cm}[b_1 \mid S_1] \hspace{0.05cm}\dd \hspace{0.02cm}b_1,\label{eqn_integral_representation_of_marginal_post_for_mu1_with_b1m}
\end{align}
which serves as our {\it constructive definition} of a posterior for $\mu_1$. Here, $\Pi_{b_1}(\cdot; S)$ coincides with $\Pi_{b_1}(\cdot; S_1)$, as it is derived solely from $S_1$ under the targeted modeling strategy (see Proposition~\ref{prop_posterior_b1}). This formulation is not only a practical construction of a marginal posterior but also represents a conceptually distinct perspective, motivated directly by the structure of the Bayesian integral, offering a new way to define posteriors when a full joint model is unavailable.

\begin{remark}[Hierarchical novelties]\label{rem:hierarchical}
A key methodological contribution of DRDB is its novel use of the Bayesian hierarchical framework. Unlike standard methods that derive a marginal posterior from a joint posterior, DRDB builds it through a hierarchically specified conditional likelihood, enabling a targeted modeling strategy that efficiently uses data while maintaining valid Bayesian inference. Although this hierarchical specification bears resemblance to semi-implicit variational inference (SIVI) \citep{yin2018semi}, the goal is fundamentally different. SIVI employs a hierarchy to improve posterior approximation, DRDB, in contrast, leverages it for exact Bayesian inference. This redefines the role of the Bayesian hierarchy, transforming it from a conventional modeling tool into a principled mechanism for achieving targeted and debiased inference.
\end{remark}

To implement this hierarchical construction, it {\it suffices to specify $\Pi_{\mu_1|b_1}$}. We formulate a {\it conditional likelihood} on $S$, which allows us to exploit the debiased representation in \eqref{eqn_debiased_represent_first_step} and adhere to our target-specific strategy. Specifically, we take a sample $\underbar{\it b}_1 \sim \Pi_{b_1}$ and model $\mu_1^\d - \b$ conditional on $\b$, using $S$. Given $\b \sim \Pi_{b_1}$, we have i.i.d. observables $\{\m(\bX_i)\} \in S$, which {\it target} $\mu_1^\d - \b$ via their mean.
Using these observables, we construct a working conditional likelihood with a corresponding prior, yielding a conditional posterior $\Pi_{\mu_1 |b_1}$ for $\mu_1|b_1$ (see Proposition~\ref{prop_posterior_mu1_given_b1}). Finally, combining $\Pi_{\mu_1| b_1}$ and $\Pi_{b_1}$ via the integral
\eqref{eqn_integral_representation_of_marginal_post_for_mu1_with_b1m}, we get a {\it marginal posterior $\Pi_{\mu_1} \equiv \Pi_{\mu_1}(\cdot; S)$ for $\mu_1$}.

\begin{remark}\label{rem:targeted-learning}
A key feature of DRDB is its efficient use of the data. It first models the nuisance bias $b_1^\d$ using $S_1$, and then models the debiased quantity $\mu_1^\d - b_1$ using $S$. This two-step, hierarchical approach allows inference to focus directly on target-specific quantities while correcting for the nuisance bias. Notably, under the targeted learning framework (and the Bayesian debiasing mechanism), the nuisance estimation method is not restricted to any particular class, enabling the use of a wide range of \textbf{flexible} models for the nuisance posterior $\Pi_{m_1}$. On the other hand, the target posteriors for the summary statistics are simple and analytically tractable (typically $t$-distributions; see Propositions \ref{prop_posterior_b1} and \ref{prop_posterior_mu1_given_b1}). An additional noteworthy advantage of DRDB (a consequence of the debiasing) is that it requires only a \textbf{single} nuisance posterior draw, yielding substantial computational efficiency without compromising theoretical validity (see Theorem \ref{thm_BvM_ATE}).
\end{remark}

\begin{remark}[Role of PS]\label{remark_post_for_r_and_PS} To simplify the computation of $\Pi_{r_1}$, Bayes' theorem yields:
\begin{align}
  r_1^\d(\bX) ~=~ \frac{\bbP(T = 1)}{\bbP(T = 1  | \bX)} ~:=~ \frac{p_1}{e^\d(\bX)}, ~~\text{ where } e^\d(\cdot) \text{ is the propensity score (PS).} \label{eqn_r1_as_PS}
\end{align}
This representation offers a flexible regression-based approach to estimate $r_1^\d(\cdot)$, avoiding direct density estimation. Specifically, we first learn $e^\d(\cdot)$ using a Bayesian binary regression on $S^\-$, e.g., Bayesian logistic regression, sparse Bayesian binary regression based on spike-and-slab type priors \citep{george1993variable}, and BART \citep{bart2010}, which yields a posterior $\Pi_{e} \equiv \Pi_{e}(\cdot; S^\-)$. Further, we construct a point estimator $\phat_1$ for $p_1$ from $S^\-$. By \eqref{eqn_r1_as_PS}, for a sample $\e(\cdot) \sim \Pi_{e}$, we define: $ \r \equiv \r(\cdot) := \phat_1/\e(\cdot)$ as a sample from its posterior $\Pi_{r_1}$. This formulation naturally incorporates the PS into our framework. While frequentist methods have long recognized the critical role of the PS in causal inference \citep{rosenbaum1983central, rosenbaum1984reducing,
bang2005doubly}, its integration in Bayesian approaches varies across methodologies \citep{ray2019debiased, ray2020semiparametric, luo2023semiparametric, breunig2025double}, and there is no consensus on how to incorporate it systematically \citep[Section 5]{li2023bayesian}. In contrast, DRDB brings the PS in organically as a core component of the debiasing mechanism via the targeted reweighting step, implemented via the density ratio, which explicitly links the procedure to the PS as shown in \eqref{eqn_r1_as_PS}.
\end{remark}
The initial DRDB formulation relies on a single data split, $(S^\-, S)$, to ensure the independence required for our debiasing and targeted modeling strategy. The drawback, however, is a significant loss of efficiency from using only a fraction of the data for the final inference. To address this, we now detail a strategy to construct a final posterior for $\mu_1$ based on usage of the {\it full data} $\mathcal{D}$.

\paragraph{The final DRDB posterior with cross-fitting.} Leveraging the randomized sample-splitting presented in Section~\ref{sec_methodology}, we can apply the DRDB procedure, as detailed in Steps \eqref{eqn_debiased_represent_first_step}--\eqref{eqn_integral_representation_of_marginal_post_for_mu1_with_b1m}, to {\it each} of the $\bbK$ test and training folds $\{(\calD_k, \calD_k^\-)\}_{k=1}^\bbK$. This yields {\it corresponding posteriors:} $\Pi_{\mu_1}^{(1)}, \ldots, \Pi_{\mu_1}^{(\bbK)}$.

To efficiently use all available data, we {\it aggregate} these fold-specific posteriors into a {\it final posterior} for $\mu_1$ that incorporates information from all splits. Following the {\it consensus Monte Carlo (CMC)-type aggregation} strategy employed in \cite{sert2025}, we define a new random variable, $\mu_1^{\CF}$, as the average of independent samples $\{\mu_1^{(k)}\}_{k = 1}^\bbK$ drawn from the respective posteriors $\{\Pi_{\mu_1}^{(k)}\}_{k=1}^\bbK$:
\begin{align}
     \mu_1^{\CF} ~:=~ \frac{1}{\bbK} \sum_{k = 1}^{\bbK} \mu_1^{(k)}, ~~\text{ and let \, $\Pi_{\mu_1}^\CF$ \, denote the corresponding distribution.}  \label{eqn_CF_version_mu1}
\end{align}

The resulting distribution, $\Pi_{\mu_1}^\CF$, serves as the {\it final DRDB posterior for $\mu_1$} and is a scaled convolution of the fold-specific posteriors $\{\Pi_{\mu_1}^{(k)}\}_{k=1}^\bbK$. This construction provides a principled and computationally efficient way to unify inference for $\mu_1^\d$ across all splits.

Although the combination step draws inspiration from CMC \citep{scott2022bayes}, its goal here is quite different. We use sample splitting not for computational efficiency, but as a methodological necessity to create independent training and test sets that validate the debiased representation and enable targeted modeling. The aggregation strategy then provides a {\it Bayesian analogue of cross-fitting} (CF) \citep{chernozhukov2018double} of {\it posteriors}, for semiparametric inference problems.

\subsection{Generalized DRDB procedure for the ATE}\label{sec_DRDB_extended_ATE}
Building on the DRDB procedure for $\mu_1^\d$, we now extend the method to our primary target, the ATE, $\tATE$.
Unlike frequentist ATE {\it point estimators}, which are simply the difference between the mean estimates for the two arms, Bayesian inference has {\it no} such direct analogue of `subtracting' {\it posterior distributions}. Hence, constructing a {\it valid} posterior for the ATE requires a more careful `first-principles' approach, accounting for the {\it full} posterior structure of the underlying components.

For clarity, we first detail the generalized DRDB procedure for one data split $(S, S^\-) = (\calD_k, \calD_k^\-)$. Second, we use the combination step from Section~\ref{sec_DRDB_for_mu1} to aggregate the posteriors from all $\bbK$ folds.

\paragraph{Debiasing step for the ATE.} The regression-based identification in \eqref{eqn_regression_identification_ATE} shows the ATE $\tATE$ can be expressed as a {\it functional} of $\overrightarrow{m}^\d$ and $\bbP_\bX$, denoted $\tATE = \tATE(\overrightarrow{m}^\d, \bbP_{\bX})$. The nuisance function $\overrightarrow{m}^\d = (m_1^\d, m_0^\d)$ includes the unknown regression functions $m_1$ and $m_0$, both requiring estimation.

Define $m^\d(\cdot) := m_1^\d(\cdot) - m_0^\d(\cdot)$. Let $\bm \equiv \bm(\cdot)$ be {\it one} draw from the posterior $\Pi_{m} \equiv \Pi_{m}(\cdot; S^\-)$ obtained from $S^\-$, as in Remark~\ref{remark_post_for_m}. Using the debiasing framework in Section~\ref{sec_DRDB_for_mu1}, we express $\tATE$ in its {\it debiased} form as: $\tATE = b^\d(\bm) + \bbE_{\bX \in S}\{\bm(\bX)  | \bm\}$, where $b^\d(\bm):= \bbE_{\bX \in S}\{m(\bX) -\bm(\bX) | \bm\}$, with the equality being valid due to the independence condition ($\bm \ind S$). The term $b^\d(\bm)$ is the expected bias from learning $m^\d(\cdot)$ and can be viewed as a function of $(b_1^\d, b_0^\d) \equiv (b^\d(\m), b^\d(\ubm))$, where $b^\d(\underbar{\it m}_t)$ is the bias for each arm $\mu^\d(t)$ for $t = 0,1$. To accurately model $b^\d(\bm)$ and construct its posterior $\Pi_b \equiv \Pi_b(\cdot; S)$, we require a {\it joint learning} strategy: both biases, $b_1^\d$ and $b_0^\d$, must be learned together to produce a joint posterior for $(b_1, b_0)$, which defines a valid posterior for $b(\bm)$.

\begin{remark}
A key distinction of the generalized DRDB procedure, compared to the one-arm case, is that modeling or learning the bias $b^\d(\bm)$ must be approached as a function of \textbf{both} $(b_1^\d, b_0^\d)$. Thus, valid and accurate inference for $b^\d(\bm)$ is only feasible if the \textbf{joint} posterior for $(b_1, b_0)$ is obtained, rather than constructing separate posteriors for each bias and combining them afterward.
\end{remark}

\paragraph{Bias modeling for the ATE.}
Following the bias analysis in Section~\ref{sec_DRDB_for_mu1} and adopting the notational conventions introduced there for the {\it first-order bias} and the {\it drift term} (see Equations~\eqref{eqn_bm1_with_Y1_and_m1}--\eqref{eqn_debiased_rep_bm}), the bias $b^\d(\bm) \equiv b^\d(\bm, r^\d)$ can be decomposed into a first-order bias $b^\d(\bm, \br)$ and a drift term $\Gamma^\d(\bm, \br)$:
\begin{equation}
    \begin{aligned}
 b^\d(\bm, r^\d) &~=~  \{ b_1^\d(\m,\r) - b_0^\d(\ubm, \ubr)\} + \{\Gamma^\d(\m,\r)
 - \Gamma^\d(\ubm, \ubr)\} \\
 & \, =: \, b^\d(\bm, \br) + \Gamma^\d(\bm, \br), \label{eqn_bias_decomposition_for_ATE}
\end{aligned}
\end{equation}
 where $r^\d \equiv r^\d(\cdot) := (r_1^\d(\cdot), r_0^\d(\cdot))$, with $r_1^\d(\bX) = p_1/e^\d(\bX)$ and $r_0^\d(\bX) := (1-p_1)/\{1-e^\d(\bX)\}$, and with corresponding posterior draws $\r:= \widehat{p}_1/\underbar{\it e} \sim \Pi_{r_1}$ and $\ubr:= (1-\widehat{p}_1)/(1-\underbar{\it e}) \sim \Pi_{r_0}$ (see Remark~\ref{remark_post_for_r_and_PS} for details), yielding a posterior sample $\br:= (\ubr, \r)$ from $(\Pi_{r_0}, \Pi_{r_1})$. \eqref{eqn_bias_decomposition_for_ATE} emphasizes that retargeting with the density ratio is {\it crucial} for accurate bias estimation. Similar to the one-arm case (see the discussions around \eqref{eqn_debiased_rep_bm} and \eqref{eqn_integral_representation_of_marginal_post_for_mu1_with_b1m}), we {\it focus} on modeling the {\it first-order} bias $b^\d(\bm, \br)$ in \eqref{eqn_bias_decomposition_for_ATE} above, consistent with our main goal of debiasing. The second-order term $\Gamma^\d(\bm, \br)$, though not our primary debiasing target, is included in the theoretical analysis of our eventual posterior of $\tATE$.

\paragraph{Posterior calculation for bias.}
Recall that $b^\d(\bm, \br)$ can be expressed as function of $(b_1^\d, b_0^\d) \equiv (b^\d(\m, \r), b^\d(\ubm, \ubr))$, where $b^\d(\underbar{\it m}_t, \underbar{\it r}_t)$ denotes the {\it first-order} bias for each arm $t = 0, 1$, as defined in \eqref{eqn_bias_decomposition_for_ATE}. Thus, calculating the posterior $\Pi_b$ for $b$ reduces to obtaining the {\it joint} posterior $\Pi_{(b_1, b_0)}$ for $(b_1, b_0)$ from $S$. Notably, since the treated and control subsets, $S_1$ and $S_0$, are {\it independent} (by design), the {\it joint} posterior $\Pi_{(b_1, b_0)}$ can be factorized as the product of the {\it marginal} posteriors:
\begin{align}
    \Pi_{(b_1, b_0)} ~\equiv~ \Pi_{(b_1, b_0)}(\cdot; S) ~=~ \Pi_{b_1}(\cdot; S_1) \times \Pi_{b_0}(\cdot; S_0) ~\equiv~ \Pi_{b_1} \times \Pi_{b_0}. \label{eqn_posterior_for_bias}
\end{align}
where $\Pi_{b_t}$ denotes the posterior of $b_t$ based on $S_t$ for $t = 0,1$. For explicit derivations, see Proposition~\ref{prop_posterior_b1} in Section~\ref{sec_likelihood_and_posterior_calc}. This factorization not only simplifies the analysis of $\Pi_b$, but also provides a straightforward sampling procedure: first, draw a sample $\b$ from $\Pi_{b_1}$ and $\ubb$ from $\Pi_{b_0}$, then define $\bb := \b - \ubb$, yielding a {\it posterior sample from $\Pi_b$} for constructing the ATE posterior.

\paragraph{Hierarchical learning and posterior construction for the ATE.} For completeness, we briefly restate the hierarchical learning strategy from Section~\ref{sec_DRDB_for_mu1}, now adopted to construct a valid posterior $\pATE$ for $\rATE$. Building on the motivation and derivations in Section~\ref{sec_DRDB_for_mu1}, construction proceeds by first drawing $b \sim \Pi_b$ using the joint posterior factorization in \eqref{eqn_posterior_for_bias}. Conditional on $b$, the posterior $\Pi_{\Delta \mid b}$ is obtained through the conditional likelihood formulation with a suitably chosen prior; see Section~\ref{sec_likelihood_and_posterior_calc} for details. The {\it marginal posterior for $\rATE$} is then defined as:
\begin{align}
        [\rATE \mid S] ~:=~ \int [\Delta \mid b, S] \hspace{0.05cm} [b \mid S] \dd b ~=~ \int [\Delta \mid b, S] \hspace{0.05cm} [b_1 \mid S_1] \hspace{0.05cm} [b_0 \mid S_0] \hspace{0.03cm}\dd b_1 \hspace{0.03cm}\dd b_0, \label{eqn_posterior_for_mu_for_S}
\end{align}
where the second equality follows from \eqref{eqn_posterior_for_bias}. This procedure then yields a {\it valid} posterior for $\rATE$, integrating the joint bias information and the hierarchical learning framework, both of which are central to the generalized DRDB procedure.

\paragraph{Posterior aggregation.} Finally, we construct the final posterior $\Pi_\rATE^\CF$ for $\rATE$ using the entire data $\calD$, to recover the efficiency lost. Following the DRDB with CF procedure in Section~\ref{sec_DRDB_for_mu1}, we combine the posteriors $\pATE^{(1)}, \cdots \pATE^{(\bbK)}$ obtained from the corresponding splits $(\calD_k, \calD_k^\-)$ via a CMC approach. We draw independent samples $\{\rATE^{(k)}\}_{k = 1}^\bbK$ from these posteriors, and define a new random variable:
\begin{equation}
    \rATE^\CF ~:=~ \frac{1}{\bbK}\sum_{i = 1}^\bbK \rATE^{(k)} \quad \text{and denote the distribution of $\rATE^\CF$ by: \,~$\pATE^\CF$.} \label{eqn_CF_version_mu}
\end{equation}
The {\it aggregated posterior} $\pATE^\CF$, {\it our final output}, integrates information on $\rATE$ across all splits, providing a principled and computationally efficient basis for final inference on the ATE (see Theorem~\ref{thm_BvM_ATE}). The main steps of the DRDB procedure with CF are summarized in Algorithm~\ref{algo}.

\begin{algorithm}
\caption{Generalized Doubly Robust Debiased Bayesian Procedure for the ATE}
\label{algo}
\vspace{0.05in}
\KwIn{Observed data $\calD$, number of folds $\bbK$, number of posterior draws $M$.}
\vspace{0.05in}
Randomly partition $\calD$ into $\bbK$ disjoint subsets $\{\calD_k\}_{k=1}^{\bbK}$, as described in Section~\ref{sec_methodology}.

\vspace{0.05in}
\For{$k = 1$ \KwTo $\bbK$:}
{
Construct training dataset $\calD_k^-:= \calD \setminus \calD_k$ and test dataset $\calD_k$. \\
Compute the nuisance posteriors $\Pi_m^{(k)} \equiv \Pi_m^{(k)}(\cdot; \calD_k^-)$ and $\Pi_{r_t}^{(k)} \equiv \Pi_{r_t}^{(k)}(\cdot; \calD_k^-)$ for $t = 0,1$, using any suitable Bayesian methods (see Remarks~\ref{remark_post_for_m} and \ref{remark_post_for_r_and_PS}) based on $\calD_k^-$. \\
Draw {\it one} sample $\bm \sim \Pi_m^{(k)}$ and {\it one} sample $\br:= (\ubr, \r)$ where $\underbar{\it r}_t \sim \Pi_{r_t}$ for $t = 0,1$. \\
Retarget the bias via density-ratio reweighting, as described in Section~\ref{sec_DRDB_extended_ATE}, to obtain the {\it first-order} bias $b^\d \equiv b^\d(\bm,\br)$. \\
Obtain the bias posterior $\Pi_b^{(k)} \equiv \Pi_b^{(k)}(\cdot; \calD_k)$ as formulated in~\eqref{eqn_posterior_for_bias} and computed in Proposition~\ref{prop_posterior_b1}.\\
Given $b \sim \Pi_b^{(k)}$, compute conditional posterior $\Pi_{\Delta \mid b}^{(k)}$ as in Proposition~\ref{prop_posterior_mu1_given_b1}. \\
Construct marginal posterior $\pATE^{(k)}$ for $\rATE$ using $\calD_k$ as in \eqref{eqn_posterior_for_mu_for_S}.
}
\vspace{0.05in}
Combine posteriors $\{\pATE^{(k)}\}_{k=1}^{\bbK}$ via a consensus Monte Carlo-type approach in \eqref{eqn_CF_version_mu} to obtain aggregated DRDB posterior $\pATE^{\CF}$ using the full dataset $\calD$.\\
\vspace{0.05in}
\KwOut{Final aggregated DRDB posterior $\pATE^{\CF}$ for the target parameter $\rATE$ from $\calD$.}
\end{algorithm}

\begin{remark}[Scalability aspects] Algorithm~\ref{algo} summarizes the generalized DRDB procedure for obtaining the posterior of $\rATE$, which can be easily adopted for the one-arm $\mu_1^\d = \bbE[Y(1)]$ detailed in Section~\ref{sec_DRDB_for_mu1}. A key feature of DRDB is its computational efficiency: it requires only a single draw from each nuisance posterior, enabling fast estimation of high-dimensional nuisance functions under both parametric and nonparametric models. In contrast, conventional Bayesian methods rely on multiple nuisance posterior samples, which can be computationally costly \citep{antonelli2022causal}. Moreover, DRDB facilitates direct posterior sampling for $\rATE$, as both the bias posterior $\Pi_b$ and the conditional posterior $\Pi_{\Delta \mid b}$ have simple, tractable forms (see Propositions~\ref{prop_posterior_b1} and \ref{prop_posterior_mu1_given_b1}). \textbf{The choice of $K$}: Theoretically, the number of folds $\bbK$ does not affect asymptotic properties as long as it remains fixed. In finite samples, however, $\bbK$ should be chosen carefully: larger $\bbK$ improves nuisance estimation through larger training sets but may increase posterior variance due to smaller test sets. Simulations suggest that $\bbK = 5$ or $10$ generally achieve a favorable balance; in Section~\ref{sec_numerics}, we report results with $\bbK = 5$ for simplicity, noting similar conclusions for $\bbK = 10$.
\end{remark}

\section{Theory}\label{sec_theory}

This section develops the theoretical foundations of the DRDB procedure and establishes posterior consistency and BvM–type results (Theorems~\ref{thm_BvM_mu1}--\ref{thm_BvM_ATE}) for the final DRDB posteriors, under mild regularity conditions on the nuisance parameters.
We first present the posterior construction details for DRDB in Section \ref{sec_likelihood_and_posterior_calc}, followed by the result for $\mu_1^\d$ in Theorem~\ref{thm_BvM_mu1}, a problem of independent interest in missing data theory, and thereafter the main result for the ATE in Theorem~\ref{thm_BvM_ATE}.

\subsection{Likelihood constructions and posterior calculations} \label{sec_likelihood_and_posterior_calc}
We provide a general characterization of the likelihood construction and prior specification used to obtain the posterior of the bias $b(\mt, \rt)$ for $t = 0,1$ and the conditional posterior calculation used in deriving the marginal posterior for $\rATE$ (and $\mu_1$) as discussed in Sections~\ref{sec_DRDB_for_mu1} and~\ref{sec_DRDB_extended_ATE}.

To avoid repetition, we present a unified procedure applicable to any bias term defined in Sections~\ref{sec_DRDB_for_mu1} and~\ref{sec_DRDB_extended_ATE}. Likewise, the conditional posterior derivation is also framed generally, covering the computation of conditional posteriors for $\rATE$ and $\mu_1$, given the corresponding bias(es). This construction includes the specific forms used in Sections~\ref{sec_DRDB_for_mu1} and~\ref{sec_DRDB_extended_ATE} as special cases.

\paragraph{Bias modeling via targeted modeling strategy.} For notational convenience, we use $\calN(\mu, \sigma^2)$ for a Normal distribution with mean $\mu$ and variance $\sigma^2$, and $t_\nu(\eta, c^2)$ for a $t$-distribution with degrees of freedom $\nu >0$, center $\eta$ and scale $c$. Define $b_t :=  b(\mt, \rt)$, $W\!(\bZ,\rt,\mt):=  \rt(\bX)\{Y -\mt(\bX)\}$ and $\sigma^2_t := \Var_{(Y,\bX) \in S_t}\{W\!(\bZ,\rt,\mt)|\rt, \mt\}$ for $t = 0,1$ and $\bZ = (Y, \bX) \in S (\ind (\underbar{\it m}_t, \underbar{\it r}_t))$. Then, given $\mt \sim \Pi_{m_t}$ and $\rt \sim \Pi_{r_t}$, for $\bZ_i \in S_t \subset S$, the weighted observables $W\!(\bZ_i,\rt, \mt)$ are i.i.d. with mean $b_t$ and variance $\sigma^2_t$. This motivates a natural {\it working} model based on a Normal distribution with unknown variance. For simplicity, we recommend using an improper prior on the model parameters, though more general priors yield the same asymptotic properties. Let $\calI_t$ denote the index set of $S_t$. The model and prior formulation are then given as: for each $t \in \{0,1\}$,
\begin{equation}
     W(\bZ_i, \rt, \mt) \mid \mt,  \rt, b_t, \sigma_t^2  ~ \iid ~ \calN(b_t, \sigma_t^2) ~ \ \text{for} \ i \in \calI_t; \quad \pi(b_t, \sigma_t^2) ~\propto~ (\sigma_t^2)^{-1}. \label{eqn_model_constr_for_b1}
\end{equation}

\begin{proposition}\label{prop_posterior_b1}
Under the model construction and the prior given in \eqref{eqn_model_constr_for_b1}, the marginal posterior $\Pi_{b_t} \equiv \Pi_{b_t}(\cdot; S_t)$ for $b_t$ follows a $t$-distribution for $t = 0,1$. Specifically, for $n_t= |S_t|, \ \nu_t = n_t -1$,
\begin{equation}
   \hspace{-0.2cm} \Pi_{b_t} = t_{\nu_t}(\eta_t, c^2_t), ~\text{with}~
   \eta_t = \frac{1}{n_t}\sum_{i \in \calI_t } W(\bZ_i, \rt, \mt)  ~ \text{and} ~ c^2_t = \frac{1}{n_t(n_t - 1)} \sum_{i \in \calI_t} \!\{W(\bZ_i,\rt, \mt) - \eta_t\}^2. \label{eqn_marginal_posterior_b1}
\end{equation}
\end{proposition}

\paragraph{Conditional posterior construction.}
To generalize the conditional posterior derivation, we introduce the generic random variables $\theta$ and $\lambda$, where $\theta$ represents either $\rATE$ or $\mu_1$, and $\lambda$ denotes the corresponding bias, i.e., $b_1$ in Section~\ref{sec_DRDB_for_mu1} or $b$ in Section~\ref{sec_DRDB_extended_ATE}. Since the likelihood is constructed using the entire data $S$ and does not depend on specific properties of the bias term or the target parameter, we are justified in adopting this unified notation.

Let $\varphi(\cdot)$ denote a generic regression function, corresponding to $m_1(\cdot)$ in Section~\ref{sec_DRDB_for_mu1} or $m(\cdot)$ in Section~\ref{sec_DRDB_extended_ATE}. Given posterior samples $\underline{\varphi} \sim \Pi_\varphi \equiv \Pi_\varphi(\cdot; S^\-)$ and $\underline{\lambda} \sim \Pi_{\lambda} \equiv \Pi_{\lambda}(\cdot; S)$, we have the i.i.d. replicates $\{\underline{\varphi}(\bX_i)\}_{i \in \calI} \in S$ {\it targeting} $\theta - \underline{\lambda}$ through their mean. Within the target-specific modeling strategy, $\theta - \underline{\lambda}$ can be viewed as a functional of the distribution of $S$, characterized by the summary statistic (mean) of $\underline{\varphi}(\bX)$. Given the independence of these observables, it is natural to adopt a Normal {\it working} model with unknown variance and place an improper prior (for simplicity again, though more general priors are allowed)
on its parameters, yielding an analytically tractable posterior. Specifically, let $\sigma^2_2 := \Var_{\bX \in S}\{\underline{\varphi}(\bX) | \underline{\varphi}\}$. Then, the resulting model is specified as:
\begin{equation}
 \underline{\varphi}(X_i)  \mid  \underline{\varphi}, \underline{\lambda},  \theta, \sigma^2_2 ~ \iid ~ \calN(\theta - \underline{\lambda}, \sigma^2_2) ~\ \text{ for } \ i \in \calI ; \quad \pi(\theta, \sigma^2_2) ~\propto~ (\sigma^2_2)^{-1}. \label{eqn_model_mu1}
\end{equation}

\begin{proposition}\label{prop_posterior_mu1_given_b1}
    Under the model-prior specification in \eqref{eqn_model_mu1}, the conditional posterior $\Pi_{\theta|\underline{\lambda}} \equiv \Pi_{\theta | \underline{\lambda}}(\cdot; \underline{\lambda}, S)$ for $\theta | \underline{\lambda}$ is a $t$-distribution:  For $\nu_S = n_S - 1$ and $\eta_S := \eta_{\underline{\varphi}} + \underline{\lambda}$,
\begin{align}
    & \Pi_{\theta |\underline{\lambda}} = t_{\nu_S}(\eta_S, c^2_S), \text{ with } \eta_S = \frac{\sum_{i \in \calI} \underline{\varphi}(\bX_i) + \underline{\lambda}}{n_S}, \ c^2_S = \frac{\sum_{i \in \calI}\{ \underline{\varphi}(\bX_i) - \eta_{\underline{\varphi}}\}^2}{n_S(n_S - 1)} . \label{eqn_conditional_posterior_mu1_given_b1}
\end{align}
\end{proposition}
In particular, setting $\theta = \mu_1, \underline{\lambda} = \b$ and $\underline{\varphi} = \m \sim \Pi_{m_1}$ gives the conditional posterior $\Pi_{\mu_1|b_1}$ for $\mu_1$ given $\b \sim \Pi_{b_1}$, as in Section~\ref{sec_DRDB_for_mu1}. Similarly, setting $\theta = \rATE, \underline{\lambda} = \underbar{\it b}$ and $\underline{\varphi} = \underbar{\it m} \sim \Pi_{m}$ yields the conditional posterior $\Pi_{\rATE|b}$ for $\rATE$ given $\underbar{\it b} \sim \Pi_{b}$, as in Section~\ref{sec_DRDB_extended_ATE}.

\begin{remark}[Some implementation details]\label{remark_post_for_m} By construction, obtaining a posterior $\Pi_m$ for $m$ based on $S^\-$ reduces to obtaining the \textbf{joint} posterior $\Pi_{(m_1, m_0)}$ for $(m_1, m_0)$ using $S^\-$. Now, the marginal regression functions $m_t(\cdot)$ for $t \in \{0,1\}$ need only the corresponding treated/control subsets $S_t^\- \subset S^\-$ to obtain $\Pi_{m_t}$, via any proper Bayesian regression method. Then, we can \textbf{directly} get the joint posterior as: $\Pi_{(m_1, m_0)} \equiv \Pi_{(m_1, m_0)}(\cdot; S^\-) = \Pi_{m_1}(\cdot; S_1^\-) \times \Pi_{m_0}(\cdot; S_0^\-)$, where the factorization follows from the independence: $S_1^\- \ind S_0^\-$ (notably \textbf{not} due to CF, but a {natural} consequence of the two-arm setup). This crucially ensures: the {joint} $\Pi_{\overrightarrow{m}}$ is obtainable from the marginals only. (Same type of independence, $S_1 \ind S_0$\,, was also used for the {\it joint} bias modeling step in \eqref{eqn_posterior_for_bias}.) To sample $\bm \sim \Pi_m$, we independently draw
$\m \sim \Pi_{m_1}$ and $\ubm \sim \Pi_{m_0}$, and then set $\bm(\cdot) := \m(\cdot) - \ubm(\cdot)$.
\end{remark}

\subsection{Main results}\label{sec_main_results}
This section presents the theoretical properties of the proposed DRDB procedure, providing results for both the ATE $\tATE$ and the one-arm parameter $\mu_1^\d\equiv\mu^\d(1) = \bbE[Y(1)]$.

\begin{assumption}\label{assumption_nuisance}
Let $\bbK \geq 2$ be a fixed integer. For $k = 1, \dots ,\bbK$, we impose the following high-level conditions on the nuisance posteriors $\Pi_{m_t}^{(k)} \equiv \Pi_{m_t}(\cdot; \calD_k^\-)$ and $\Pi_{r_t}^{(k)} \equiv \Pi_{r_t}(\cdot; \calD_k^\-)$ for $t = 0,1$:\begin{itemize}
\item[(a)] Let $\varepsilon_{m,n} \geq 0$ and $\varepsilon_{r,n} \geq 0$ be two sequences satisfying $\max\{\varepsilon_{m,n}, \varepsilon_{r,n}\} \to 0$ and $\sqrt{n_\bbK}\varepsilon_{m,n} \varepsilon_{r,n} \to 0$, and let $M_n \geq 0$ be any sequence such that $M_n \to \infty$ as $n \to \infty$. Then, we assume that:
\begin{align}
       & \Pi_{m_t}^{(k)}\{\|\underbar{\it m}_t(\bX) - m_t^*(\bX) \|_{\bbL_2(\bbP_{\bX})} > M_n\varepsilon_{m,n} \mid \calD_k^\-\} \ \xrightarrow[] {\bbP_{\calD_k^\-}} \ 0, ~~~\mbox{and}  \label{eqn_nuisance_m_contraction} \\
      & \Pi_{r_t}^{(k)}\{\|\underbar{\it r}_t(\bX) - r_t^*(\bX) \|_{\bbL_2(\bbP_{\bX})} > M_n \varepsilon_{r,n} \mid \calD_k^\-\} \ \xrightarrow[]{\bbP_{\calD_k^\-}} \ 0, \label{eqn_nuisance_r_contraction}
    \end{align}
where $m_t^*(\cdot) \in \bbL_2(\bbP_{\bX})$ and $r_t^*(\cdot) \in \bbL_2(\bbP_{\bX})$ is the respective limiting nuisance functions.
\item[(b)] We assume $\sup_{\bx \in \mathcal{X}}\bbE\{Y - m_t^*(\bX) | \bX = \bx\} <\infty$, $\|\underbar{\it r}_t(\bX)\{Y - \underbar{\it m}_t(\bX)\}\|_{\bbL_4(\bbP_\bZ)} = O_{\bbP_{(\underbar{\it m}_t, \underbar{\it r}_t)}}(1)$, $\|\underbar{\it r}_t(\bX)\|_{\bbL_\infty(\bbP_\bX)} = O_{\bbP_{r_t}}(1)$ and $\| \underbar{\it m}_t(\bX) \|_{\bbL_4(\bbP_\bX)} = O_{\bbP_{m_t}}(1)$, for any $\underbar{\it m}_t \sim \Pi_{m_t}^{(k)}$ and $\underbar{\it r}_t \sim \Pi_{r_t}^{(k)}$.
\end{itemize}
\end{assumption}
\begin{remark} Assumption~\ref{assumption_nuisance} (b) is standard, mild moment conditions. Condition (a) specifies the posterior contraction requirement for the nuisance parameters: $\Pi_{m_t}$ and $\Pi_{r_t}$ contract around some fixed functions $m_t^*(\cdot)$ and $r_t^*(\cdot)$ at rates $\varepsilon_{m,n}$ and $\varepsilon_{r,n}$, respectively; these limiting functions need not match with the true $m_t^\d(\cdot)$ and $r_t^\d(\cdot)$. Notably, this is the only assumption required on the nuisance posteriors for Theorems~\ref{thm_BvM_mu1}--\ref{thm_BvM_ATE} to establish posterior consistency and BvM-type results. In contrast to traditional approaches \citep{ray2020semiparametric, breunig2025double, yiu2025}, which often require restrictive conditions (Donsker class) on the nuisance model or explicit posterior correction, DRDB is flexible: any Bayesian method may be used to obtain the nuisance posteriors, provided Assumption~\ref{assumption_nuisance} (a) holds. Moreover, Assumption~\ref{assumption_nuisance} (a) serves as the Bayesian analogue of the $\bbL_2$-consistency requirements for nuisance parameters commonly imposed in frequentist debiased semiparametric inference; see, e.g., \citet{chernozhukov2018double}.
\end{remark}

Let $P$ and $Q$ be two probability measures on a measurable space $(\Omega, \mathcal{B})$. Then, the total variation distance between $P$ and $Q$ is defined as: $d_\TV(P, Q):= \sup_{B \in \mathcal{B}}|P(B) - Q(B)|$.

\begin{theorem}[Main result for the one-arm case: $\mu^\d(1)$]\label{thm_BvM_mu1}
Suppose Assumptions~\ref{assumptions_standard_causal_assump} and~\ref{assumption_nuisance} hold.
\begin{itemize}
\item[(a)] If both nuisance models are well-specified (Case \textbf{C1}), the posterior $\Pi_{\mu_1}^\CF$ satisfies the BvM theorem: $d_\TV(\Pi_{\mu_1}^\CF, \ \calN(\mu_1(m_1^\d, r_1^\d), c^2(m_1^\d, r_1^\d))) \cvP 0$, as $n \to \infty$, where:
    \begin{align}
   \mu_1(m_1^\d, r_1^\d):= \frac{1}{n}\sum_{i = 1}\! m_1^\d(\bX_i) + \frac{1}{n}\!\sum_{i = 1}^n \psi(\bZ_i, m_1^\d, r_1^\d), ~\text{and}~
    c^2(m_1^\d, r_1^\d) := \Var\{\mu_1(m_1^\d, r_1^\d)\},\label{eqn_pmean_pvar_for_mu1}
    \end{align}
with $\psi(\bZ, m_1^\d, r_1^\d) ~: =~ r_1^\d(\bX)T \{Y - m_1^\d(\bX)\}/p_1$.
\item[(b)] If only one nuisance model is well-specified (Case \textbf{C2} or \textbf{C3}), then $\Pi_{\mu_1}^\CF$ contracts around $\mu^\d(1)$ at a rate $\epsilon_n$: For any sequence $M_n \to \infty$, as $n \to \infty$, $\Pi_{\mu_1}^\CF\{|\mu_1 - \mu^\d(1)| \geq M_n \epsilon_n \mid \calD\} \cvP 0$, where $\epsilon_n$ is the contraction rate for the well-specified nuisance model.
\end{itemize}
\end{theorem}

\begin{theorem}[Main result for the ATE: $\tATE$]\label{thm_BvM_ATE}
    Suppose Assumptions~\ref{assumptions_standard_causal_assump} and~\ref{assumption_nuisance} hold.
\begin{itemize}
\item[(a)]Under Case \textbf{C1}, the posterior $\Pi^\CF_\rATE$ satisfies the BvM theorem:\\
$d_\TV(\Pi^\CF_\rATE, \, \calN(\rATE(m^\d, r^\d), c^2(m^\d, r^\d))) \cvP 0$, as $n \to \infty$, where
\begin{align}
    \hspace{-1.35ex} \rATE(m^\d, r^\d) := \frac{1}{n}\! \sum_{i = 1} \!m^\d(\bX_i) + \frac{1}{n}\! \sum_{i = 1}^n \! \gamma(\bZ_i, m^\d, r^\d), ~~ c^2(m^\d, r^\d) := \Var\{\rATE(m^\d, r^\d)\},\label{eqn_pmean_pvar_for_ATE}
\end{align}
and $\gamma(\bZ, m^\d, r^\d) = r_1^\d(\bX)T\{Y - m_1^\d(\bX)\}/p_1 - r_0^\d(\bX)(1-T)\{Y - m_0^\d(\bX)\}/(1-p_1)$.

\item[(b)]Under Case \textbf{C2} or \textbf{C3}, $\pATE^\CF$ contracts around the true $\tATE$ at a rate $\epsilon_n$, where $\epsilon_n$ is the contraction rate for the well-specified nuisance model.
\end{itemize}
\end{theorem}

A direct consequence of Theorems~\ref{thm_BvM_mu1} and \ref{thm_BvM_ATE} is that DRDB provides natural Bayesian point estimators for $\mu^\d(1)$ and the ATE $\tATE$ through the {\it posterior means:} $\widehat{\mu}_1(\m, \r)$ and $\widehat{\rATE}(\bm, \br)$, respectively. In Corollary~\ref{cor_pmean_convergence}, we rigorously characterize the theoretical properties of these DRDB point estimators.

\begin{corollary}[Properties of the posterior means]
\label{cor_pmean_convergence}
Suppose the assumptions of Theorems~\ref{thm_BvM_mu1} and \ref{thm_BvM_ATE} hold.
\begin{itemize}
\item[(a)] Under Case \textbf{C1}, the posterior means are asymptotically equivalent to the means of the limiting distributions: (i) $\sqrt{n}\{\widehat{\mu}_1(\m, \r) - \mu_1(m_1^\d, r_1^\d)\} = \op(1)$ and (ii) $\sqrt{n}\{\widehat{\rATE}(\bm, \br) - \rATE(m^\d, r^\d)\} = \op(1)$.

\item[(b)] Under Case \textbf{C2} or \textbf{C3}, $\widehat{\mu}_1(\m, \r)$ and $\widehat{\rATE}(\bm, \br)$ are $\epsilon_n^{-1}$-consistent estimators for $\mu^\d(1)$ and $\tATE$, respectively, where $\epsilon_n$ denotes the posterior contraction rate of the well-specified nuisance model:
$(i) \big\{\widehat{\mu}_1(\m, \r) - \mu^\d(1)\big\} = \Op(\epsilon_n) \ \text{and} \ (ii) \big\{\widehat{\rATE}(\bm, \br) - \tATE \big\} = \Op(\epsilon_n)$.
\end{itemize}
\end{corollary}

\begin{remark}[Matching frequentist properties]\label{remark_pmean_asym_properties}
Corollary~\ref{cor_pmean_convergence} establishes the asymptotic behavior of DRDB point estimators under different nuisance model specifications. In Case \textbf{C1}, $\widehat{\rATE}(\bm, \br)$ admits asymptotically linear representations at the $\sqrt{n}$-rate, achieving semiparametric efficiency as the mean $\rATE(m^\d, r^\d)$ of the limiting distribution in Theorem~\ref{thm_BvM_ATE} coincides with the `efficient' influence function (EIF) for the ATE \citep{robins1995semiparametric, hahn98}. Under Cases \textbf{C2} and \textbf{C3}, $\widehat{\rATE}(\bm, \br)$ remains consistent for $\tATE$ with convergence rates determined by the posterior contraction rate of the well-specified nuisance model, reflecting the double robustness of the DRDB procedure.
\end{remark}

Similar asymptotic properties to those in Remark~\ref{remark_pmean_asym_properties} hold for $\widehat{\mu}_1(\m, \r)$, relevant for mean estimation of missing outcomes under MAR \citep{tsiatis2007semiparametric}. The details are analogous and omitted for brevity.

\begin{remark}[Key theoretical properties of DRDB]\label{remark_DRDB_theory}
Theorems~\ref{thm_BvM_mu1} and \ref{thm_BvM_ATE} establish the main theoretical guarantees of DRDB. Under Case~\textbf{C1}, the posterior $\pATE^\CF$ contracts around the true ATE at the parametric $1/\sqrt{n}$ rate. This condition holds, for example, if each nuisance contracts faster than $n^{-1/4}$, allowing wide flexibility in model choices. Also, the asymptotic variance of $\pATE^\CF$ remains unaffected by nuisance estimation error: it relies only on the limiting functions $m^*$ and $r^*$, no other features of the nuisance posteriors. This robustness arises from the Bayesian debiasing strategy combined with CF and the targeted modeling introduced in Section~\ref{sec_methodology}. Thus, DRDB allows high-dimensional or nonparametric nuisance models with rates slower than $n^{-1/2}$, while retaining $\sqrt{n}$-rate inference for the ATE. Moreover, under Cases~\textbf{C2} and \textbf{C3}, $\pATE^\CF$ still contracts around the true ATE at the rate of the correctly specified model, showing Bayesian double-robustness of DRDB. Analogous results for the one-arm case: $\mu^\d(1)$ also follow from Theorem~\ref{thm_BvM_mu1}.
\end{remark}

\subsection{Comparison with alternative Bayesian debiasing strategies}
A main challenge in Bayesian semiparametric inference is that regularization bias from flexible nuisance models can propagate into the posterior for a low-dimensional target parameter, such as the ATE, compromising inferential validity \citep{bickel2012semiparametric, castillo2015bernstein}. To mitigate this nuisance-induced bias, two prominent strategies have emerged: {\it prior modification}, which tailors prior specification to the semiparametric model structure \citep[e.g.,]{ray2020semiparametric, breunig2025double}; and {\it posterior correction}, which applies a post-hoc adjustment using the efficient influence function (EIF) \citep{yiu2025}. A complementary approach by \cite{luo2023semiparametric} constructs posteriors via exponentially tilted empirical likelihood and establishes BvM results for partially linear and parametric models. Our DRDB procedure introduces an alternative perspective by embedding debiasing directly into the modeling process through {\it targeted learning}. Below, we compare DRDB with these two state-of-the-art alternatives, focusing on the seminal works of \cite{ray2020semiparametric} and \cite{yiu2025}, both mainly interested in the mean outcome under MAR, closely related to the one-arm case for our ATE setting, discussed in Section~\ref{sec_DRDB_for_mu1}.

The prior modification approach proposed by \cite{ray2020semiparametric} models the {\it full data distribution} with nonparametric priors, innovatively augmenting the prior for the outcome regression with an estimator of the propensity score (PS) obtained from an independent auxiliary data. This augmentation perturbs the prior in the model’s least favorable direction of the semiparametric model, and thereby mitigates nuisance-induced bias. While theoretically elegant, this approach requires customized prior design and strong conditions, including Donsker assumptions and smoothness constraints, limiting the use of machine learning or high dimensional nuisance models. Moreover, its theoretical guarantees require both nuisance models to be correctly specified.

The posterior correction approach in \cite{yiu2025} takes a different route by adding a stochastic correction to posterior draws, inspired by the frequentist {\it one-step estimator} \cite[Chapter 5.7]{van2000asymptotic}. Using the EIF and the {\it Bayesian bootstrap} \citep{rubin1981bayesian}, it projects draws towards the truth along the most informative direction, achieving bias reduction. While this approach can target multiple functionals, it relies on posterior draws for the full data model, and, similar to DRDB, its theoretical guarantees rely on product-type rate conditions on the nuisance posteriors. In addition, it requires Donsker-type conditions \citep[Chapter 19]{van2000asymptotic}, which are often restrictive for flexible or high-dimensional nuisance models, whereas DRDB avoids the explicit need for such conditions (see Assumption~\ref{assumption_nuisance}) via its distinct use of cross-fitting.

Cross-fitting (CF) is a well-established tool in the frequentist literature, commonly used to relax Donsker-type conditions on nuisance parameters \cite{chernozhukov2018double}. In DRDB, however, CF is {\it not} merely a technical device; it is a crucial component of our debiasing mechanism. Within the DRDB framework, CF is methodologically essential (see Remark~\ref{remark_data_splitting}), and also underpins a novel {\it Bayesian analogue}, an aggregation strategy that combines {\it posteriors} across splits (folds), enabling principled Bayesian semiparametric inference while leveraging the full data efficiently.

DRDB departs from both methods in philosophy and implementation. Rather than modifying the prior or correcting the posterior, DRDB embeds debiasing within the modeling process through the {\it debiased representation}. It explicitly identifies and {\it learns the bias as a separate target}, using summary statistics that are directly informative about the ATE and its bias. Furthermore, the role of the propensity score underlines these differences: it enters externally to guide prior design, while in DRDB it arises {\it naturally} through density ratio weighting. Likewise, the use of independent data differs: \cite{ray2020semiparametric} primarily leverages it for technical convenience, whereas DRDB uses it to {\it validate} the debiased representation and strengthen bias correction.

Theoretically, DRDB requires only {\it high-level posterior contraction} for the nuisance models, and mild moment conditions, making it compatible with flexible nuisance models. Most importantly, DRDB achieves {\it Bayesian double robustness}: ATE posterior remains consistent and contracts at the rate of the well-specified nuisance, even if the other is misspecified. This guarantee is stronger than those of both prior augmentation and one-step posterior correction methods, offering greater stability under model misspecification. Computationally, DRDB is also far simpler and more scalable. Whereas \cite{yiu2025} require a full set of posterior draws for all nuisance parameters, DRDB needs only a {\it single} posterior draw per nuisance per cross-fitting fold. Compared with \cite{ray2020semiparametric}, which demands intricate prior customization, DRDB’s modular structure allows seamless use of standard Bayesian regression tools without model-specific tuning.

A further distinction of DRDB is the generality of its methodological framework, which extends seamlessly beyond the ATE to a broad class of causal estimands. By representing each estimand as a weighted functional and adapting the debiasing and retargeting steps accordingly, DRDB maintains inferential validity and computational scalability under these broader settings. The detailed formulation of this extension is presented in Section~\ref{sec_DRDB_conditional} of the \hyperref[sec_supplementary]{Supplementary Material}.

DRDB unifies the theoretical strengths of prior modification and posterior correction while also introducing new modeling perspectives and methodological advances that extend its applicability. By embedding debiasing directly into the modeling process, it enables flexible nuisance estimation, ensures valid inference under mild conditions, and remains computationally efficient. These features establish DRDB as a robust, theoretically grounded, and practically scalable framework for Bayesian causal inference and, more broadly, Bayesian semiparametric inference in general.

\section{Numerical studies}\label{sec_numerics}

We evaluate the finite-sample performance of the proposed DRDB procedure for both estimation and inference of the ATE through extensive simulation studies across various data-generating mechanisms and Bayesian methods for nuisance estimation, including both well-specified and misspecified settings. The mean of the DRDB posterior $\pATE^\CF$ serves as our point estimator. We report the empirical bias ({\bf Bias}) and mean squared error ({\bf MSE}) for estimation accuracy. For inference evaluation, we report the empirical coverage probabilities ({\bf Cov}) and average lengths of the 95\% credible intervals ({\bf CI-Len}) based on 1000 posterior samples from $\pATE^\CF$. For the number of folds, we set $\bbK = 5$ for computational efficiency. All reported results are based on 500 replications. We study two scenarios: one where {\it both} nuisance models are well-specified (Section~\ref{sec_sim_correctly_specified}), and another where {\it only one} of them is well-specified (Section~\ref{supp_sec_sim_misspecified} of the \hyperref[sec_supplementary]{Supplementary Material}).

The following notations are used throughout this section. For any integer $p \geq 1$ and $v \in \bbR$, let $v_p$ be the vector $v_p := (v, \dots, v)' \in \bbR^{p \times 1}$. Let $I_p$ denote the $p \times p$ identity matrix, $\calN_p(\mu_p, \Sigma_p)$ denotes the $p$-variate Gaussian distribution with mean vector $\mu_p \in \bbR^{p}$ and covariance matrix $\Sigma_p \in \bbR^{p \times p}$.

\subsection{Simulation results}\label{sec_sim_correctly_specified}
Throughout, we set $n = 1000$ and consider $p = 10$, $50$, and $200$, representing low, moderate, and high-dimensional settings, respectively. For each $i = 1, \dots, n$, the covariate vector is generated as: $\bX_i \iid \calN_p(0_p, I_p)$. Conditional on $\bX_i$, the treatment assignment follows: $T_i| X_i \sim \mathrm{Ber}\{e^\d(\bX_i)\}$, where $e^\d(\bX_i) = 1/\{1 + \exp^{-(\bX_i'\beta_3 - 0.08)}\}$ and $\beta_3 = (0.35_2, 0_{p-2})$, ensuring the positivity condition in Assumption~\ref{assumptions_standard_causal_assump}. Given $\bX_i$, the potential outcomes are generated as: $Y_i(t) \sim \calN(m_t^\d(\bX_i), \sigma_t^2)$, with $m_1^\d(\bX) = 5 + 2 \bX'\beta_1$ and $m_0^\d(\bX) = 3 + \bX'\beta_0$, and variances $\sigma_t^2 = \Var\{m_t^\d(\bX)\}/5$ for $t = 0, 1$. The observed outcome is therefore $Y_i = T_i Y_i(1) + (1-T_i) Y_i(0)$, and the observed data is $\calD = \{(Y_i, \bX_i, T_i)\}_{i=1}^n$. The regression coefficients $\beta_1 = \beta_0$ are set as $ (1_{s/2}, 0.5_{s/2}, 0_{p-s})$, where $s$ denotes {\it sparsity}. For $p = 10$, we use $s = 3$ and $s = 10$; for $p = 50$ and $p = 200$, we take $s \approx \sqrt{p}$ and $s \approx p/4$ to represent sparse and moderately dense regimes.

For illustrative purposes, we employ the sample mean $\phat_1: = n^{-1}\sum_{i = 1}^n T_i$ as a point estimator for $p_1 := \bbP(T = 1)$. For the posterior $\Pi_e$, we only use sparse Bayesian ({\tt BS}) logistic regression with nonlocal priors (NLP) \citep{johnson2012bayesian} for simplicity. For the posteriors $\Pi_{m_1}$ and $\Pi_{m_0}$, in addition to {\tt BS} linear regression with NLP, we consider Bayesian ridge regression ({\tt BR}) and {\tt BART} \citep{bart2010}, implemented using the R package \texttt{BART}. For parametric methods, we consider the Gaussian linear and logistic regression working models: for $i = 1, \dots, n$, $Y_i | \bX_i, T_i = 0, \gamma_0, \theta_0,\sigma \iid \calN(\gamma_0 + \bX_i'\theta_0, \sigma^2)$,  $Y_i| \bX_i, T_i = 1, \gamma_1, \theta_1,\tau \iid \calN(\gamma_1 + \bX_i'\theta_1, \tau^2)$ and $T_i | \bX_i \sim \mathrm{Ber}\{e(\bX_i)\}$, where $e(\bX_i) = 1/\{1+ \exp^{-(\gamma_3 + \bX_i'\theta_3)}\}$.
For {\tt BR}, we employ a Gaussian prior on the regression coefficients and an improper prior on the variance parameter. The ridge parameter $\lambda$ is estimated using an empirical Bayes approach, with the point estimate $\widehat{\lambda}$ obtained via the R package \texttt{glmnet}. For {\tt BS}, posterior samples for $(\gamma_0, \theta_0)$, $(\gamma_1, \theta_1)$ and $(\gamma_3, \theta_3)$ are obtained using the R package \texttt{mombf}. Finally, as a performance benchmark, we report the results for the frequentist {\it oracle} estimator, constructed using the empirical mean of the EIF of the ATE \citep{hahn98} over $\calD_n$ using the {\it true} nuisance parameters, denoted as {\tt Oracle}.

\begin{table}%[!htbp]
\caption{Estimation and inference results for the ATE based on DRDB with the setting given in Section~\ref{sec_sim_correctly_specified}. To distinguish between nuisance estimation methods, we denote each approach as DRDB-M, where ``M'' indicates the specific nuisance method: S = {\tt BS}, R = {\tt BR}, and B = {\tt BART}.
\label{table_correctly_specified}}
\begin{tabular*}{\columnwidth}{@{\extracolsep\fill}lcccccccc@{\extracolsep\fill}}
\toprule
{\bf Method} & {\bf Bias} & {\bf MSE}
& {\bf Cov} & {\bf CI-Len} & {\bf Bias} & {\bf MSE}
& {\bf Cov} & {\bf CI-Len} \\
\midrule
\multicolumn{1}{@{}l@{}}{$p = 10$} & \multicolumn{4}{@{}c@{}}{$s = 3$} & \multicolumn{4}{@{}c@{}}{$s = 10$} \\
\cline{1-1}\cline{2-5}\cline{6-9}
 {\tt Oracle}  & -0.006 & 0.008 & 0.936 & 0.333  &0.010 &0.020 & 0.944 &0.552 \\
 DRDB-S  & -0.006 &0.008 &0.936& 0.335 & 0.010 &0.020 & 0.944 &0.552 \\
 DRDB-R   & -0.008 & 0.008 & 0.934 & 0.340 & 0.010 &0.021 & 0.952 & 0.566 \\
 DRDB-B   & -0.010 & 0.009 & 0.958 & 0.396 & 0.014 &0.029 & 0.962 & 0.722 \\

\multicolumn{1}{@{}l@{}}{$p = 50$} & \multicolumn{4}{@{}c@{}}{$s = 7$} & \multicolumn{4}{@{}c@{}}{$s = 13$} \\
\cline{1-1}\cline{2-5}\cline{6-9}%
{\tt Oracle}   & 0.008 & 0.015 & 0.952 & 0.482  & 0.000 & 0.025 & 0.954 & 0.644 \\
DRDB-S   & 0.009 & 0.016 & 0.954 & 0.494  & 0.000 & 0.027 & 0.952 & 0.671 \\
DRDB-R   & 0.015 & 0.018 & 0.956 & 0.546  & -0.001 & 0.030 & 0.960 & 0.729 \\
DRDB-B   & 0.013 & 0.020 & 0.968 & 0.630  & -0.006 & 0.035 & 0.992 & 0.908 \\

\multicolumn{1}{@{}l@{}}{$p = 200$} & \multicolumn{4}{@{}c@{}}{$s = 14$} & \multicolumn{4}{@{}c@{}}{$s = 50$} \\
\cline{1-1}\cline{2-5}\cline{6-9}%
{\tt Oracle}   & 0.007  & 0.026 & 0.962 & 0.656   & -0.010 & 0.105 & 0.946 & 1.235 \\
DRDB-S   & 0.006  & 0.027 & 0.972 & 0.684   & -0.014 & 0.134 & 0.980 & 1.600 \\
DRDB-R   & 0.028  & 0.048 & 0.982 & 1.026   & 0.010  & 0.167 & 0.978 & 1.930 \\
DRDB-B   & 0.016  & 0.036 & 0.986 & 0.973   & -0.019 & 0.243 & 0.994 & 2.468 \\
\bottomrule
\end{tabular*}
\end{table}
\begin{figure}[ht!]
  \centering
  % First row of plots
  \begin{minipage}[b]{0.45\textwidth}
    \centering
    \includegraphics[width=\linewidth]{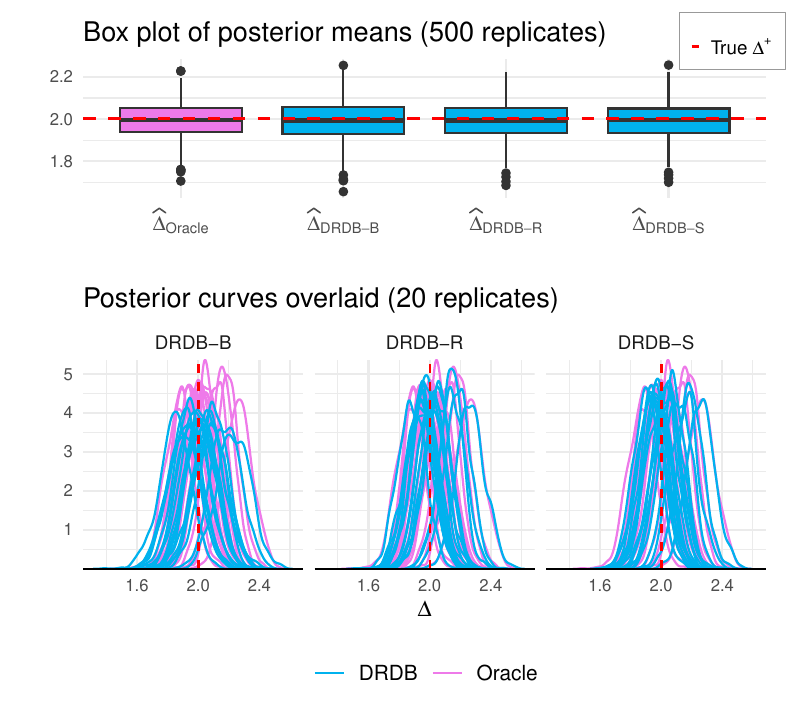}
    \par\vspace{2pt}
    (a) Setting: \(p = 10\) with \(\mathbf{s = 3}\).
  \end{minipage}
  \hfill
  \begin{minipage}[b]{0.45\textwidth}
    \centering
    \includegraphics[width=\linewidth]{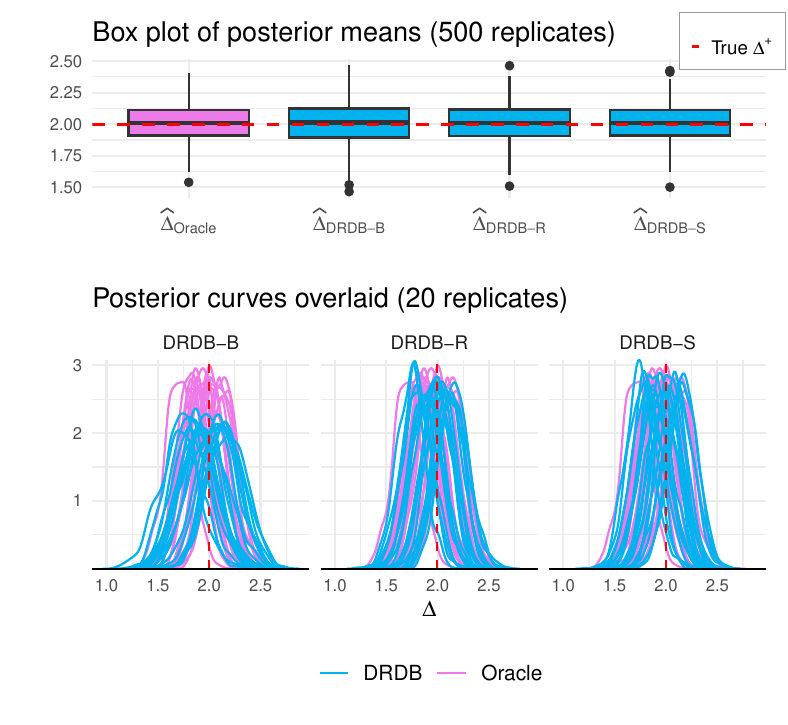}
    \par \vspace{2pt}
    (b) Setting: \(p = 10\) with \(\mathbf{s = 10}\).
  \end{minipage}

  \vspace{0.3cm} % Add some vertical space between rows

  % Second row of plots
  \begin{minipage}[b]{0.45\textwidth}
    \centering
    \includegraphics[width=\linewidth]{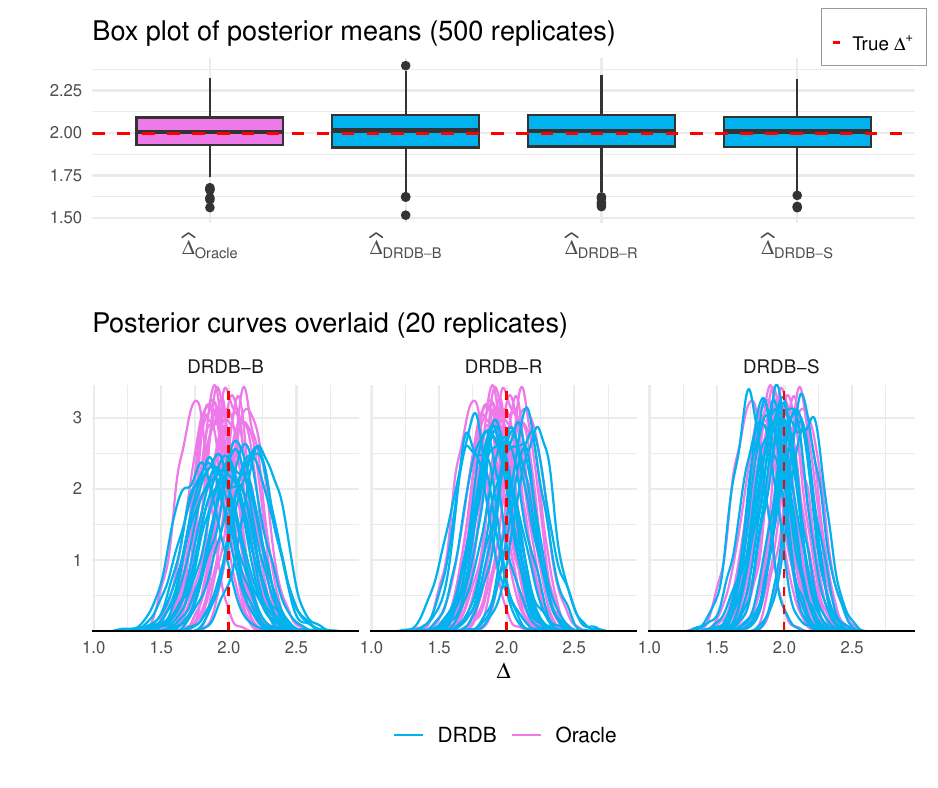}
    \par\vspace{2pt}
    (c) Setting: \(p = 50\) with \(\mathbf{s = 7}\).
  \end{minipage}
  \hfill
  \begin{minipage}[b]{0.45\textwidth}
    \centering
    \includegraphics[width=\linewidth]{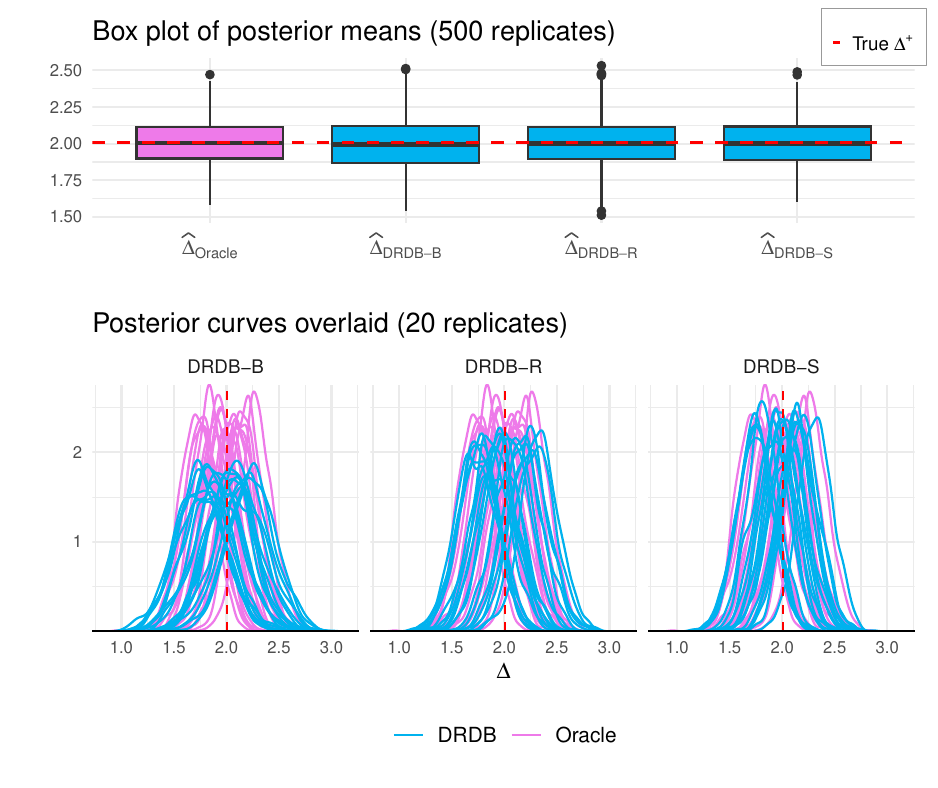}
    \par\vspace{2pt}
    (d) Setting: \(p = 50\) with \(\mathbf{s = 13}\).
  \end{minipage}
  \caption{
    Box plots of posterior means (based on 500 replications) and overlaid density curves (based on 20 iterations) for the posteriors \(\Pi_{\texttt{Oracle}}\) (pink) and \(\pATE\) (blue) of \(\rATE\). The plots show results from using three methods (\texttt{BART}, \texttt{BR}, \texttt{BS}) to obtain the nuisance posteriors. The subfigures correspond to different values of \(p\) and \(s\). Each density curve is generated using 1000 posterior samples of \(\rATE\). The red dashed vertical line indicates the true \(\tATE\) ($=2$ for all settings).
  }
  \label{fig_combined_p1050}
\end{figure}

\begin{figure}[ht!]
  \centering
  % Third row of plots
  \begin{minipage}[b]{0.42\textwidth}
    \centering
    \includegraphics[width=\linewidth]{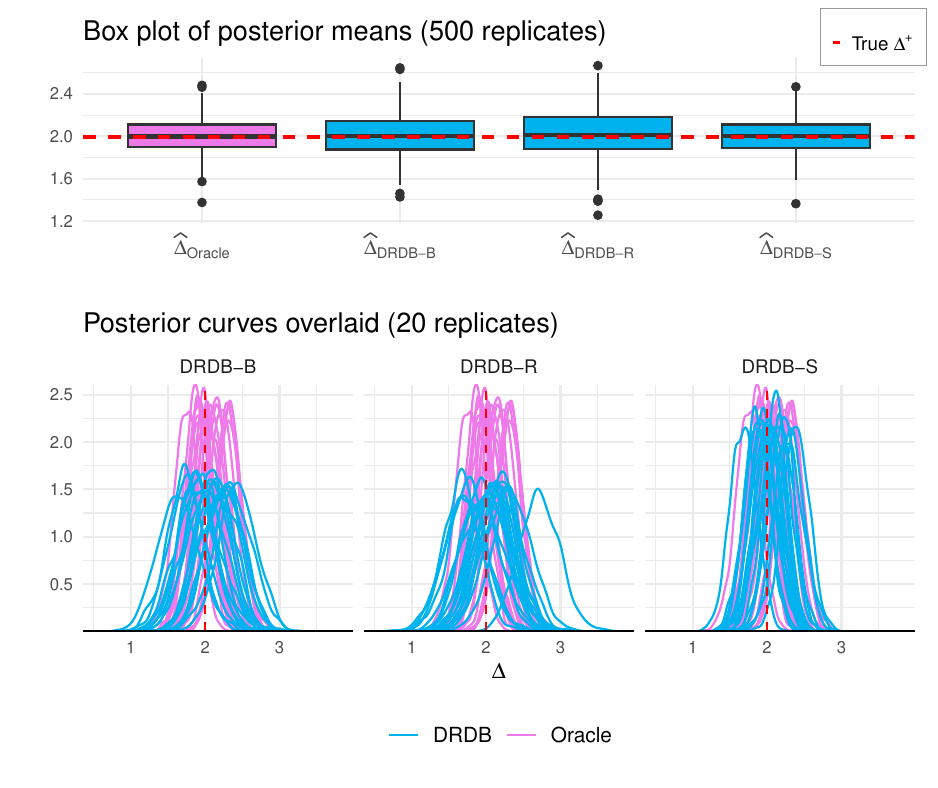}
    \par\vspace{2pt}
    (e) Setting: \(p = 200\) with \(\mathbf{s = 14}\).
  \end{minipage}
  \hfill
  \begin{minipage}[b]{0.42\textwidth}
    \centering
    \includegraphics[width=\linewidth]{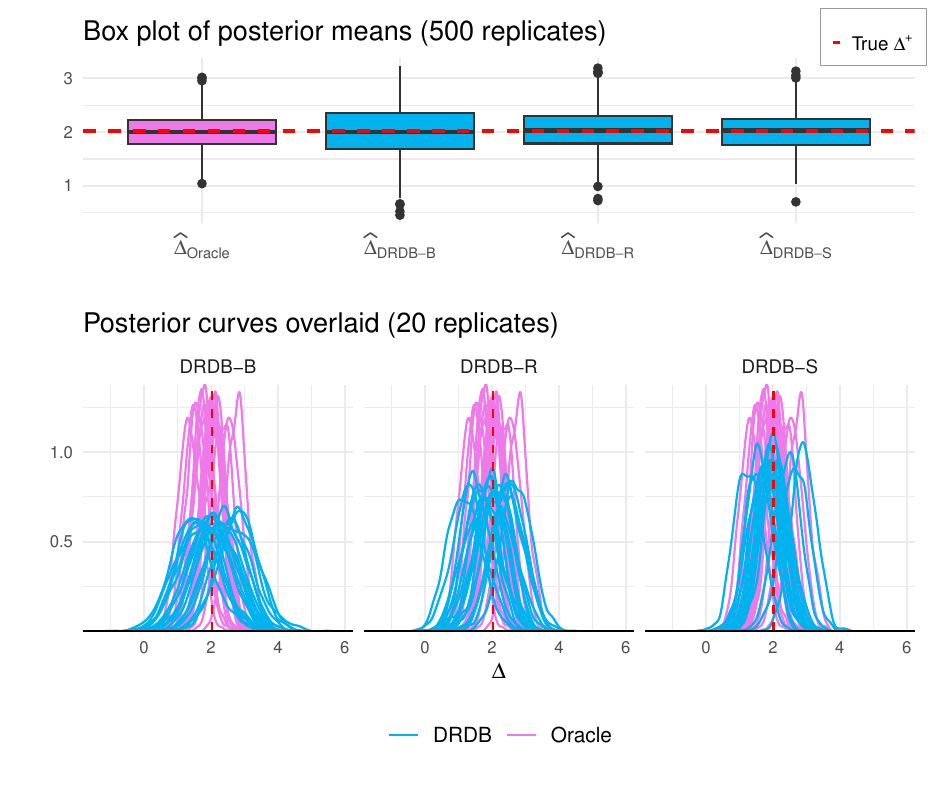}
    \par\vspace{2pt}
    (f) Setting: \(p = 200\) with \(\mathbf{s = 50}\).
  \end{minipage}
  \caption{
    Box plots of posterior means and overlaid density curves for the posteriors \(\Pi_{\texttt{Oracle}}\) and \(\pATE\) of \(\rATE\) for $p =200$ and $s= 14$ or $50$. The rest of the caption details are the same as in Figure~\ref{fig_combined_p1050}.}
  \label{fig_p200}
\end{figure}

Table~\ref{table_correctly_specified} reports the estimation and inference results for the ATE. Across all scenarios and nuisance estimation methods, DRDB performs nearly identically to the {\tt Oracle} estimator. In both low- and moderate-dimensional settings ($p = 10, 50$), and across sparse ($s = \sqrt{p}$) and moderately dense ($s = p/4$) regimes, all DRDB variants yield negligible differences from {\tt Oracle}. Even DRDB-B performs comparably well, though with slightly higher finite-sample bias and wider credible intervals, reflecting the slower convergence of nonparametric nuisance estimation. As the dimensionality increases to $p = 200$, the role of sparsity becomes more pronounced. The sparsity-adaptive DRDB-S {\it continues} to perform close to {\tt Oracle}, effectively leveraging the underlying sparse structure, whereas DRDB-R and DRDB-S exhibit increased finite-sample bias and wider CIs, with DRDB-B most affected due to slower convergence inherent to nonparametric procedures.

Table~\ref{table_correctly_specified} shows that DRDB {\it consistently} achieves coverage probabilities near the {\it nominal} 95\% level across nearly all settings with various dimensions, sparsity levels, and the nuisance estimation methods. In high-dimensional scenarios, DRDB-B tends to produce slightly conservative coverage and wider CIs, reflecting the effect of slower convergence rates and increased finite-sample bias associated with higher parameter dimensionality. Nevertheless, CI lengths {\it remain comparable} to those of the {\tt Oracle}, and DRDB-S consistently {\it maintains} coverage near the nominal level. These patterns highlight DRDB’s {\it robustness} and {\it adaptability}: when nuisance models capture the true
underlying structure, its performance remains near-optimal {\it even} in high-dimensional regimes.

Figures~\ref{fig_combined_p1050} and \ref{fig_p200} visually support these findings. Across all scenarios, DRDB posteriors exhibit approximately {\it Gaussian shapes}, always concentrated around the truth $\tATE = 2$, with posterior means (from box plots) {\it tightly} centered near the truth. In low or moderate-dimensional regimes ($p = 10$ or $p = 50$), all DRDB variants produce nearly {\it identical} posteriors that closely {\it align with} the Oracle, confirming the {\it stability} and {\it efficiency} of DRDB when nuisance parameters are accurately estimated. In the high-dimensional settings ($p = 200$), however, the differences between nuisance models become more pronounced: DRDB-S maintains the most concentrated posterior, nearly matching the Oracle, while DRDB-R and DRDB-B show wider and heavier-tailed posteriors, reflecting slower nuisance convergence and the inherent difficulty of ridge and nonparametric estimators in such settings. Yet, all density curves {\it remain} well centered around $\tATE = 2$, and the credible interval lengths are largely comparable to the Oracle’s. Overall, the plots confirm DRDB’s robust finite-sample performance, while clearly revealing how high dimensionality amplifies the effect of nuisance model choice on second-order properties of posterior concentration.

Finally, while our analysis herein focuses on correctly specified models, high-dimensional settings inherently introduce a {\it soft} form of misspecification due to finite-sample nuisance estimation bias. Even in such cases, DRDB {\it consistently} estimates the ATE across nuisance methods, demonstrating its double robustness: as long as the product of nuisance convergence rates exceeds the parametric rate, DRDB yields near-Oracle estimates with {\it accurate coverage}. Under explicit functional-form misspecification (see Section~\ref{supp_sec_sim_misspecified} in the \hyperref[sec_supplementary]{Supplementary Material}), DRDB continues to produce stable estimates across all nuisance models and maintains {\it valid coverage} with only mildly wider credible intervals. These findings support its {\it double robustness}, and highlight a distinct {\it advantage} of the {\it Bayesian} framework as well. Overall, the results confirm that DRDB offers robust, efficient, and {\it reliable inference} across diverse settings, highlighting its {\it insensitivity} to nuisance estimation.

\subsection{Real data application}
We apply the proposed DRDB approach to evaluate {\it the effect of smoking cessation on weight gain} using data from the National Health and Nutrition Examination Survey Epidemiologic Follow-up Study (NHEFS). The NHEFS is a longitudinal study initiated by the National Center for Health Statistics and the National Institute on Aging, in collaboration with other agencies of the U.S. Public Health Service \citep{Hernan2020}, and is well-studied in the causal inference literature \citep{Ertefaie2022,
zhang2023double}.
Our analysis focuses on a subset of $n = 1566$ individuals who were cigarette smokers aged 25--74 years at baseline in 1971, and had a follow-up visit in 1982. The treatment $T \in \{0,1\}$ ({\tt qsmk}) equals 1 if the individual quit smoking before the follow-up and 0 otherwise. The outcome $Y \in \bbR$ ({\tt wt82\_71}) is the weight gain (in kg), calculated as the difference between body weight at the follow-up and baseline weight. We use the same nine covariates as in \cite{Ertefaie2022}: {\tt sex}, {\tt race}, {\tt age}, {\tt education}, {\tt smokeintensity}, {\tt smokeyrs}, {\tt exercise}, {\tt active}, and {\tt wt71}. A detailed data description is available at \url{https://miguelhernan.org/whatifbook}.

Our goal is to estimate the ATE of the treatment on body weight gain. As a baseline for comparison, we include the naive estimator ({\tt Naive}), defined as the difference in empirical mean outcomes between the treatment groups \citep{tsiatis2007semiparametric}. Using the Bayesian nuisance estimation methods given in Section~\ref{sec_sim_correctly_specified}, we compute the DRDB posterior of $\rATE$ via Algorithm~\ref{algo} with $\bbK = 5$. For each posterior, we draw 1000 samples to compute Monte Carlo approximations of the posterior mean (as the point estimate) and the 2.5\% and 97.5\% quantiles to construct 95\% credible intervals (CIs).

Table~\ref{table_real_data} reports the ATE estimates, the 95\% CIs, and the respective CI lengths (CI-Leng). All methods yield a positive ATE, indicating that quitting smoking is associated with weight gain. A notable disparity exists between the \texttt{naive} estimate of 2.541 and the DRDB-based estimates, which are consistently higher and clustered in the 3.31 to 3.54 range. This gap strongly suggests the presence of confounding (given the observational setting), which the \texttt{naive} estimator fails to address. Further, our estimates {\it align} closely with those of \cite{Ertefaie2022}, who reported ATE estimates between 3.20 and 3.42 using IPW estimators on the same data, providing strong external validation for our method. Moreover, DRDB yields substantially narrower CIs, with lengths of approximately 1.55, compared with the \texttt{naive} estimator (length $= 1.91$) and the IPW-based CIs in \citet{Ertefaie2022} (length $\approx 2.34$ -- $2.41$), demonstrating {\it improved efficiency}. This consistency in results across various nuisance estimation methods highlights the {\it robustness} of DRDB to the choice of nuisance model. Overall, these findings indicate the presence of notable confounding via $\bX$, and support a causal effect of smoking cessation on weight gain after adjusting for confounders.
\begin{table}%[!htbp]
\caption{Estimation and inference results for ATE in NHEFS data based on the naive and DRDB approaches. The rest of the caption details remain the same as in Table~\ref{table_correctly_specified}.}
\label{table_real_data}
\begin{tabular*}{\columnwidth}{@{\extracolsep\fill}lccc@{\extracolsep\fill}}
\toprule
{\bf Method} & {\bf Estimator} & {\bf 95\% Cred. Interval} & {\bf CI-Leng} \\
\midrule
{\tt Naive}   & 2.541 & (1.585, 3.496) & 1.911 \\
DRDB-S  & 3.368 & (2.560, 4.116) & 1.556 \\
DRDB-R  & 3.312 & (2.502, 4.124) & 1.622 \\
DRDB-B  & 3.544 & (2.699, 4.458) & 1.759 \\
\bottomrule
\end{tabular*}
\end{table}

\section{Concluding discussion}\label{sec_conclusion}

We proposed a DRDB procedure for estimating {\it causal functionals} within a Bayesian framework, focusing on the ATE and the single-arm mean $\mu^\d(1)$. DRDB builds on {\it two key ideas}: (i) explicit separation of the nuisance estimation from inference on the target,
through a {\it Bayesian debiasing} mechanism (coupled with {\it data splitting}); and (ii) a {\it targeted modeling} strategy, along with {\it hierarchical learning} and use of {\it posterior aggregation}. The debiasing step corrects nuisance-induced bias via a density-ratio–based retargeting, while targeted modeling ensures efficient use of relevant summaries. All these aspects put together make DRDB both theoretically robust/efficient and computationally scalable. Theorems~\ref{thm_BvM_mu1}--\ref{thm_BvM_ATE} establish BvM theorems for DRDB with matching frequentist guarantees when both nuisances are correctly specified, as well as a Bayesian analogue of the frequentist double robustness, while operating within the posterior framework. Computationally, DRDB is flexible, allowing for {\it off-the-shelf} nuisance model choices, admits a simple final posterior, and requires only a {\it single} posterior draw from each nuisance per CF fold, ensuring its scalability in high dimensions. Finally, our framework also readily extends to a wide class of policy-relevant causal estimands, as shown in Section \ref{sec_DRDB_conditional} of the \hyperref[sec_supplementary]{Supplementary Material}, with similar theoretical guarantees expected to hold therein as well. Overall, DRDB provides a {\it novel} Bayesian perspective on debiasing and double robustness, offering both theoretical guarantees and practical feasibility, and opens the door to broader applications of Bayesian targeted semiparametric learning.

\phantomsection
\addcontentsline{toc}{section}{Acknowledgments and funding}
\section*{Acknowledgments and funding}
 The authors gratefully acknowledge funding support from the NSF grants: NSF-DMS 2113768 (Abhishek Chakrabortty) and  NSF-DMS 2210689 (Anirban Bhattacharya) towards partial support of this research.

\phantomsection
\addcontentsline{toc}{section}{Data availability}
\section*{Data availability}
The dataset used in this paper is publicly available. The code used to generate the simulation results is available at \url{https://github.com/gozdesert/DRDB}.

\phantomsection
\addcontentsline{toc}{section}{Supplementary material}
\section*{Supplementary material}\label{sec_supplementary}
The Supplementary Material contains: (i) an extension of our methodology to other causal estimands, (ii) additional numerical results, and (iii) all technical details, including the complete proofs of the main results, that could not be included in the main paper.

\phantomsection
\addcontentsline{toc}{section}{References}
\bibliography{references}

\begin{thebibliography}{39}
\providecommand{\natexlab}[1]{#1}
\providecommand{\url}[1]{\texttt{#1}}
\expandafter\ifx\csname urlstyle\endcsname\relax
  \providecommand{\doi}[1]{doi: #1}\else
  \providecommand{\doi}{doi: \begingroup \urlstyle{rm}\Url}\fi

\bibitem[Antonelli et~al.(2022)Antonelli, Papadogeorgou, and
  Dominici]{antonelli2022causal}
Joseph Antonelli, Georgia Papadogeorgou, and Francesca Dominici.
\newblock Causal inference in high dimensions: a marriage between bayesian
  modeling and good frequentist properties.
\newblock \emph{Biometrics}, 78:\penalty0 100--114, 2022.

\bibitem[Bang and Robins(2005)]{bang2005doubly}
Heejung Bang and James~M Robins.
\newblock Doubly robust estimation in missing data and causal inference models.
\newblock \emph{Biometrics}, 61\penalty0 (4):\penalty0 962--973, 2005.

\bibitem[Bickel and Kleijn(2012)]{bickel2012semiparametric}
Peter~J. Bickel and Bart J.~K. Kleijn.
\newblock The semiparametric {Bernstein--von Mises} theorem.
\newblock \emph{The Annals of Statistics}, 40\penalty0 (1):\penalty0 206--237,
  2012.

\bibitem[Breunig et~al.(2025)Breunig, Liu, and Yu]{breunig2025double}
Christoph Breunig, Ruixuan Liu, and Zhengfei Yu.
\newblock Double robust {B}ayesian inference on average treatment effects.
\newblock \emph{Econometrica}, 93\penalty0 (2):\penalty0 539--568, 2025.

\bibitem[Castillo and Rousseau(2015)]{castillo2015bernstein}
Isma{\"e}l Castillo and Judith Rousseau.
\newblock A {Bernstein--von Mises} theorem for smooth functionals in
  semiparametric models.
\newblock \emph{The Annals of Statistics}, 43\penalty0 (5):\penalty0
  2353--2383, 2015.

\bibitem[Chernozhukov et~al.(2018)Chernozhukov, Chetverikov, Demirer, Duflo,
  Hansen, Newey, and Robins]{chernozhukov2018double}
Victor Chernozhukov, Denis Chetverikov, Mert Demirer, Esther Duflo, Christian
  Hansen, Whitney Newey, and James Robins.
\newblock Double/debiased machine learning for treatment and structural
  parameters.
\newblock \emph{The Econometrics Journal}, 21\penalty0 (1):\penalty0 C1--C68,
  2018.

\bibitem[Chipman et~al.(2010)Chipman, George, and McCulloch]{bart2010}
Hugh~A. Chipman, Edward~I. George, and Robert~E. McCulloch.
\newblock {BART: Bayesian additive regression trees}.
\newblock \emph{The Annals of Applied Statistics}, 4\penalty0 (1):\penalty0
  266--298, 2010.

\bibitem[Drovandi et~al.(2015)Drovandi, Pettitt, and Lee]{drovandi2015}
Christopher~C. Drovandi, Anthony~N. Pettitt, and Anthony Lee.
\newblock {Bayesian Indirect Inference Using a Parametric Auxiliary Model}.
\newblock \emph{Statistical Science}, 30\penalty0 (1):\penalty0 72--95, 2015.

\bibitem[Ertefaie et~al.(2022)Ertefaie, Hejazi, and van~der Laan]{Ertefaie2022}
Ashkan Ertefaie, Nima~S. Hejazi, and Mark~J. van~der Laan.
\newblock Nonparametric inverse probability weighted estimators based on the
  highly adaptive lasso.
\newblock \emph{Biometrics}, 79\penalty0 (2):\penalty0 1029--1043, 2022.

\bibitem[Gelman et~al.(2014)Gelman, Carlin, Stern, Dunson, Vehtari, and
  Rubin]{gelman2014}
Andrew Gelman, John~B. Carlin, Hal~S. Stern, David~B. Dunson, Aki Vehtari, and
  Donald~B. Rubin.
\newblock \emph{Bayesian Data Analysis}.
\newblock Texts in Statistical Science Series. CRC Press, Boca Raton, FL, 3rd
  edition, 2014.
\newblock ISBN 978-1-4398-4095-5.

\bibitem[George and McCulloch(1993)]{george1993variable}
Edward~I. George and Robert~E. McCulloch.
\newblock Variable selection via {Gibbs} sampling.
\newblock \emph{Journal of the American Statistical Association}, 88\penalty0
  (423):\penalty0 881--889, 1993.

\bibitem[Ghosal and Van~der Vaart(2017)]{ghosal2017fundamentals}
Subhashis Ghosal and Aad Van~der Vaart.
\newblock \emph{Fundamentals of Nonparametric {B}ayesian Inference}, volume~44.
\newblock Cambridge University Press, 2017.

\bibitem[Hahn(1998)]{hahn98}
Jinyong Hahn.
\newblock On the role of the propensity score in efficient semiparametric
  estimation of average treatment effects.
\newblock \emph{Econometrica}, 66\penalty0 (2):\penalty0 315--331, 1998.

\bibitem[Hahn et~al.(2020)Hahn, Murray, and Carvalho]{hahn2020bayesian}
P~Richard Hahn, Jared~S Murray, and Carlos~M Carvalho.
\newblock Bayesian regression tree models for causal inference: Regularization,
  confounding, and heterogeneous effects (with discussion).
\newblock \emph{Bayesian Analysis}, 15\penalty0 (3):\penalty0 965--1056, 2020.

\bibitem[Hernán and Robins(2020)]{Hernan2020}
Miguel~A. Hernán and James~M. Robins.
\newblock \emph{Causal Inference: What If}.
\newblock Chapman \& Hall/CRC, 2020.

\bibitem[Hirano et~al.(2003)Hirano, Imbens, and Ridder]{hirano2003}
Keisuke Hirano, Guido Imbens, and Geert Ridder.
\newblock Efficient estimation of average treatment effects using the estimated
  propensity score.
\newblock \emph{Econometrica}, 71:\penalty0 1161--1189, 02 2003.

\bibitem[Imbens and Rubin(2015)]{imbens2015causal}
Guido~W Imbens and Donald~B Rubin.
\newblock \emph{Causal Inference in Statistics, Social, and Biomedical
  Sciences}.
\newblock Cambridge University Press, 2015.

\bibitem[Johnson and Rossell(2012)]{johnson2012bayesian}
Valen~E. Johnson and David Rossell.
\newblock Bayesian model selection in high-dimensional settings.
\newblock \emph{Journal of the American Statistical Association}, 107\penalty0
  (498):\penalty0 649--660, 2012.

\bibitem[Li et~al.(2023)Li, Ding, and Mealli]{li2023bayesian}
Fan Li, Peng Ding, and Fabrizia Mealli.
\newblock Bayesian causal inference: {A} critical review.
\newblock \emph{Philosophical Transactions of the Royal Society A},
  381\penalty0 (2247):\penalty0 20220153, 2023.

\bibitem[Linero and Antonelli(2023)]{linero2022}
Antonio Linero and Joseph Antonelli.
\newblock The how and why of {B}ayesian nonparametric causal inference.
\newblock \emph{WIREs Computational Statistics}, 15\penalty0 (1):\penalty0
  e1583, 2023.

\bibitem[Luo et~al.(2023)Luo, Graham, and McCoy]{luo2023semiparametric}
Yu~Luo, Daniel~J Graham, and Emma~J McCoy.
\newblock Semiparametric {B}ayesian doubly robust causal estimation.
\newblock \emph{Journal of Statistical Planning and Inference}, 225:\penalty0
  171--187, 2023.

\bibitem[Ray and Szab{\'o}(2019)]{ray2019debiased}
Kolyan Ray and Botond Szab{\'o}.
\newblock Debiased {B}ayesian inference for average treatment effects.
\newblock \emph{Advances in Neural Information Processing Systems}, 32, 2019.

\bibitem[Ray and van~der Vaart(2020)]{ray2020semiparametric}
Kolyan Ray and Aad van~der Vaart.
\newblock Semiparametric {Bayesian} causal inference.
\newblock \emph{The Annals of Statistics}, 48\penalty0 (5):\penalty0
  2999--3020, 2020.

\bibitem[Rivoirard and Rousseau(2012)]{rivoirard2012bernstein}
Vincent Rivoirard and Judith Rousseau.
\newblock {Bernstein--von Mises} theorem for linear functionals of the density.
\newblock \emph{The Annals of Statistics}, 40\penalty0 (3):\penalty0
  1489--1523, 2012.

\bibitem[Robins and Rotnitzky(1995)]{robins1995semiparametric}
James~M Robins and Andrea Rotnitzky.
\newblock Semiparametric efficiency in multivariate regression models with
  missing data.
\newblock \emph{Journal of the American Statistical Association}, 90\penalty0
  (429):\penalty0 122--129, 1995.

\bibitem[Robins et~al.(1994)Robins, Rotnitzky, and Zhao]{robins1994estimation}
James~M Robins, Andrea Rotnitzky, and Lue~P Zhao.
\newblock Estimation of regression coefficients when some regressors are not
  always observed.
\newblock \emph{Journal of the American Statistical Association}, 89\penalty0
  (427):\penalty0 846--866, 1994.

\bibitem[Rosenbaum and Rubin(1984)]{rosenbaum1984reducing}
Paul Rosenbaum and Donald Rubin.
\newblock Reducing bias in observational studies using subclassification on the
  propensity score.
\newblock \emph{Journal of the American Statistical Association}, 79\penalty0
  (387):\penalty0 516--524, 1984.

\bibitem[Rosenbaum and Rubin(1983)]{rosenbaum1983central}
Paul~R Rosenbaum and Donald~B Rubin.
\newblock The central role of the propensity score in observational studies for
  causal effects.
\newblock \emph{Biometrika}, 70\penalty0 (1):\penalty0 41--55, 1983.

\bibitem[Rubin(1974)]{rubin1974estimating}
Donald~B Rubin.
\newblock Estimating causal effects of treatments in randomized and
  nonrandomized studies.
\newblock \emph{Journal of Educational Psychology}, 66\penalty0 (5):\penalty0
  688, 1974.

\bibitem[Rubin(1981)]{rubin1981bayesian}
Donald~B Rubin.
\newblock The {B}ayesian bootstrap.
\newblock \emph{The Annals of Statistics}, 9\penalty0 (1):\penalty0 130--134,
  1981.

\bibitem[Scott et~al.(2022)Scott, Blocker, Bonassi, Chipman, George, and
  McCulloch]{scott2022bayes}
Steven~L. Scott, Alexander~W. Blocker, Fernando~V. Bonassi, Hugh~A. Chipman,
  Edward~I. George, and Robert~E. McCulloch.
\newblock {Bayes and big data: The consensus Monte Carlo algorithm}.
\newblock In \emph{Big Data and Information Theory}, pages 8--18. Routledge,
  2022.

\bibitem[Sert et~al.(2025)Sert, Chakrabortty, and Bhattacharya]{sert2025}
Gözde Sert, Abhishek Chakrabortty, and Anirban Bhattacharya.
\newblock Bayesian semi-supervised inference via a debiased modeling approach.
\newblock \emph{Econometrics and Statistics}, {\it (to appear)}, 2025.
\newblock ISSN 2452-3062.

\bibitem[Tao and Fu(2019)]{tao2019}
Yebin Tao and Haoda Fu.
\newblock Doubly robust estimation of the weighted average treatment effect for
  a target population.
\newblock \emph{Statistics in Medicine}, 38\penalty0 (3):\penalty0 315--325,
  2019.

\bibitem[Tsiatis(2007)]{tsiatis2007semiparametric}
Anastasios Tsiatis.
\newblock \emph{Semiparametric Theory and Missing Data}.
\newblock Springer Science \& Business Media, 2007.

\bibitem[van~der Vaart(2000)]{van2000asymptotic}
Aad~W. van~der Vaart.
\newblock \emph{Asymptotic Statistics}, volume~3.
\newblock Cambridge University Press, 2000.

\bibitem[Williams(1998)]{williams1998prediction}
Christopher K.~I. Williams.
\newblock {Prediction with Gaussian processes: From linear regression to linear
  prediction and beyond}.
\newblock In \emph{Learning in Graphical Models}, pages 599--621. Springer,
  1998.

\bibitem[Yin and Zhou(2018)]{yin2018semi}
Mingzhang Yin and Mingyuan Zhou.
\newblock Semi-implicit variational inference.
\newblock In \emph{International Conference on Machine Learning}, pages
  5660--5669. PMLR, 2018.

\bibitem[Yiu et~al.(2025)Yiu, Fong, Holmes, and Rousseau]{yiu2025}
Andrew Yiu, Edwin Fong, Chris Holmes, and Judith Rousseau.
\newblock Semiparametric posterior corrections.
\newblock \emph{Journal of the Royal Statistical Society Series B: Statistical
  Methodology}, pages 1--30, 2025.

\bibitem[Zhang et~al.(2023)Zhang, Chakrabortty, and Bradic]{zhang2023double}
Yuqian Zhang, Abhishek Chakrabortty, and Jelena Bradic.
\newblock {Double robust semi-supervised inference for the mean: selection bias
  under {MAR} labeling with decaying overlap}.
\newblock \emph{Information and Inference: A Journal of the IMA}, 12\penalty0
  (3):\penalty0 2066--2159, 2023.

\end{thebibliography}

\clearpage

\section*{\centering \large Supplement to ``Bayesian semiparametric causal inference: Targeted doubly robust estimation of treatment effects''}\label{supplementary_material}

\renewcommand{\thefootnote}{\arabic{footnote}}

{\centering
\large Gözde Sert, Abhishek Chakrabortty, and Anirban Bhattacharya \par }
\vspace{-0.01cm}
{\centering {\large \it Department of Statistics, Texas A\&M University} \footnote{{\it Email addresses:} \href{ mailto:gozdesert@stat.tamu.edu}{gozdesert@stat.tamu.edu} (Gözde Sert), \href{mailto:abhishek@stat.tamu.edu}{abhishek@stat.tamu.edu} (Abhishek Chakrabortty),
\href{mailto:anirbanb@stat.tamu.edu}{anirbanb@stat.tamu.edu} (Anirban Bhattacharya; corresponding author).} \par}

\renewcommand*{\theHsection}{\thesection}
\renewcommand*{\theHsubsection}{\thesubsection}
\setcounter{section}{0}
\setcounter{equation}{0}
\setcounter{subsection}{0}

\renewcommand \thesection{S\arabic{section}}
\renewcommand \thesubsection{\thesection.\arabic{subsection}}
\renewcommand \thesubsubsection{\thesubsection.\arabic{subsubsection}}
\renewcommand\thetable{S.\arabic{table}}
\renewcommand \thefigure{S.\arabic{figure}}
\renewcommand{\theequation}{S.\arabic{equation}}

\setcounter{table}{0}
\setcounter{figure}{0}
\setcounter{footnote}{0}

% Reset counters to 0 for the supplement
\setcounter{theorem}{0}
\setcounter{lemma}{0}
\setcounter{corollary}{0}

% Redefine the counter display format to add the 'S' prefix
\renewcommand{\thetheorem}{S\arabic{theorem}}
\renewcommand{\thelemma}{S\arabic{lemma}}
\renewcommand{\thecorollary}{S\arabic{corollary}}

\par\medskip

This supplement (Sections \ref{sec_DRDB_conditional}--\ref{supp_sec_proof_of_preliminary}) provides methodological extensions, additional numerical analyses, and technical materials, including all proofs, that could not be accommodated in the main paper. {\bf (i)} Section~\ref{sec_DRDB_conditional} presents a detailed analysis of the {\it extension} of the DRDB procedure to a broader class of policy-relevant causal estimands. {\bf (ii)} Section~\ref{supp_sec_sim_misspecified} presents a supplementary table and a detailed discussion of simulation results for the {\it misspecified setting} introduced in Section~\ref{sec_numerics}. {\bf (iii)} Section~\ref{supp_sec_preliminary} contains preliminary results and intermediate lemmas used in proving the main theorems (Theorems~\ref{thm_BvM_mu1} and \ref{thm_BvM_ATE}). {\bf (iv)} Section~\ref{sup_sec_proofs_of_main_results} provides the full proofs of the main results of the paper. {\bf (v)} Finally, Section~\ref{supp_sec_proof_of_preliminary} includes proofs of the preliminary results and lemmas introduced in Section~\ref{supp_sec_preliminary}.

\section{Extension of DRDB to other causal estimands}\label{sec_DRDB_conditional}

DRDB naturally extends from the ATE to a broad class of causal estimands of the form:
\begin{align*}
    \tATE_h ~:=~ \mathbb{E}[Y(1) - Y(0)| h(\bX, T) = c],
\end{align*}
where $h(\bX, T)$ is a {\it known} measurable function of the covariates and the treatment, and $c$ specifies the condition of interest. Notably, for $\tATE_h$ to be well-defined, a standard overlap condition is required, ensuring the target set $\{(\bX, T): h(\bX, T) = c\}$ has positive probability.

Many classical causal estimands arise as special cases of this general formulation. For example, this includes: the {\it average treatment effect on the treated} (ATT), $\bbE[Y(1) - Y(0)| T = 1]$, and that on the control (ATC), $\bbE[Y(1) - Y(0)| T = 0]$, based on $h(\bX, T) = T$ and $c = t \in \{1,0\}$; as well as {\it subgroup-specific effects}, such as $\bbE[Y(1) - Y(0) | \bX_{[j]} > \delta]$ (subgroups under a specific
covariate threshold, e.g., `Age $>$ 60'), using $h(\bX, T) = 1(\bX_{[j]} > \delta)$ for a given $\delta$ and $j$, and $c = 1$. In general, $\tATE_h$ unifies many policy-relevant subgroup effects \citep{hirano2003, tao2019}.

To explain the DRDB extension, we fix an arm $t \in \{0, 1\}$ and consider $\mu_h^\d(t) := \mathbb{E}[Y(t) | h(\bX, T) = c]$. Under the NUC assumption, we have:
\begin{align*}
    \mu_h^\d(t) ~=~ \mathbb{E}_{\bX}[m_t^\d(\bX) w_h^\d(\mathbf{X})], ~~\mbox{where}~~ w_h^\d(\mathbf{X}) := \bbP\{h(\bX, T) = c \,|\, \bX\} /\bbP\{h(\bX, T) = c\}
\end{align*}
is the {\it adjustment} needed to {\it target} the subpopulation defined by $h, c$. DRDB for $\mu_h^\d(t)$ then parallels that in Section~\ref{sec_DRDB_for_mu1}, with two key adaptations: (i) the bias is defined by the weighted functional, and (ii) the retargeting step is adapted to the subgroup defined by $h, c$.

\paragraph{Bias decomposition and retargeting.} Following the notation of Section~\ref{sec_DRDB_for_mu1}, for a given nuisance draw $\underbar{\it m}_t \sim \Pi_{m_t}(\cdot; S^\-)$, we decompose $\mu_h^\d(t)$ as:
\begin{align*}
    \mu_h^\d(t) ~=~ b_t^\d(\underbar{\it m}_t, w_h^\d) ~+~ \bbE_{\bX \sim S}\{\underbar{\it m}_t(\bX) w_h^\d(\bX)\mid \underbar{\it m}_t\}, \quad \mbox{where:}
\end{align*}
$b_t^\d \equiv b_t^\d(\underbar{\it m}_t, w_h^\d) := \bbE_{\bX \sim S}[\{m_t^\d(\bX) - \underbar{\it m}_t(\bX)\} w_h^\d(\bX)| \underbar{\it m}_t]$
is the nuisance estimation bias. This equality follows from the independence of $S$ and $S^\-$. Under the NUC condition and the definition of $w_h^\d(\bX)$ as above, we further obtain that:
\begin{align*}
    b_t^\d = \bbE_{\bX \sim S}[w_h^\d(\bX)\bbE\{Y(t) -\underbar{\it m}_t(\bX) | \bX, T \!=\! t\}] = \bbE_{\bX \sim f_A}[\bbE\{Y(t) -\underbar{\it m}_t(\bX)| \bX, T \!=\! t\}],
\end{align*}
where $A := \{ h(\bX, T) = c\}$ is the {\it subgroup of interest} with covariate distribution $f_A(\bX) \equiv f(\bX|A)$. This shows that {\it correct} estimation of $b_t^\d$ requires access to $(Y(t),\bX)$ for all units in $S$, but $Y(t)$ is observed {\it only} for the subgroup $S_t \subset S$.

To address this distributional mismatch, DRDB introduces a {\it density ratio} adjustment as follows. Let $r_h^\d(\bX):= f_A(\bX)/f(\bX|T= t)$ be the unknown density ratio, which reweights the observed treated (or control) sample to match the covariate distribution $f_A$. Then, the bias admits the {\it estimable representation:} $b_t^\d = \bbE_{(Y, \bX)|T=t}[r_h^\d(\bX)\{Y - \underbar{\it m}_t(\bX)\}|T=t]$.

To estimate $r_h^\d \equiv r_h^\d(\cdot)$, notice that Bayes' theorem yields:
\begin{align*}
    r_h^\d(\bX) ~=~ \{\bbP(A\mid \bX)\bbP(T = t)\}/\{\bbP(A)\bbP(T = t\mid\bX)\}.
\end{align*}
This reveals that estimating $r_h^\d$ requires {\it two} flexible binary regression models -- for the subgroup probability $\bbP(A|\bX)$ and the treatment assignment probability $\bbP(T = t|\bX)$, avoiding explicit density estimation, along with simple point estimators for $\bbP(T = t)$ and $\bbP(A)$, similar to the discussion in Remark \ref{remark_post_for_r_and_PS}. It is also worth noting that the original ATE corresponds to the special/trivial case: $\bbP(A|\bX) = 1$ and $\bbP(A) = 1$, which simplifies the density ratio to the form \eqref{eqn_r1_as_PS} in Remark \ref{remark_post_for_r_and_PS}.

Given {\it one} draw $\underbar{\it r}_h$ from $\Pi_{r_h}(\cdot; S^\-)$, the bias $b_t^\d$ can be modeled using the observables $\underbar{\it r}_h(\bX)\{Y - \underbar{\it m}_t(\bX)\} \in S_t$ via the targeted modeling strategy. Once the bias is estimated, DRDB applies hierarchical learning to model $\mu_h^\d(t) - b_t^\d$ using the weighted observables: $\{\underbar{\it w}_h(\bX)\underbar{\it m}_t(\bX)\}_{\bX \in S}$, where $\underbar{\it w}_h(\cdot)$ denotes {\it a sample} from the posterior $\Pi_{w_h}$. Analogous to the estimation of $r_h^\d$, posterior samples of $w_h \equiv w_h(\cdot)$ are obtained by fitting a flexible binary regression for $\bbP(A|\bX)$ on $S^\-$, along with a simple point estimator for $\bbP(A)$, similar to the procedure described in Remark~\ref{remark_post_for_r_and_PS}. The details for both steps are parallel to those in Section~\ref{sec_DRDB_for_mu1} and are omitted for brevity. Finally, the present discussion addresses the one-arm case. The two-arm version requires a similar appropriate modification, in the above spirit, to the method of Section~\ref{sec_DRDB_extended_ATE} via analogous {\it joint} debiasing steps.

Hence, DRDB offers a unified and principled framework for valid Bayesian inference across a broad class of treatment effects. The essence of this seamless extension is in expressing the target as a weighted functional and appropriately extending the original debiasing and retargeting steps.

\section{\large Additional simulation results: Misspecified nuisance models}\label{supp_sec_sim_misspecified}
To illustrate DRDB performance under model misspecification, we adapt the simulation setup from Section~\ref{sec_sim_correctly_specified}. The true regression function is now {\it quadratic:} \(m_1^\d(\bX) = 5 + 2\bX'\beta_1 + (\bX^2)'\beta_{12}\), where \(\bX^2\) is the coordinate--wise square of \(\bX\). We choose \(\beta_{12}\) to satisfy \(\Var(2\bX'\beta_1)/\Var\{(\bX^2)'\beta_{12}\} = 3\), which balances the contributions from the linear and quadratic terms. All other aspects of the data-generating mechanism and the nuisance functions remain the same as in Section~\ref{sec_sim_correctly_specified}.

For parametric {\it nuisance model estimators}, in addition to the Gaussian linear regression employed in Section~\ref{sec_sim_correctly_specified}, we also consider the following Gaussian {\it quadratic regression} model: for $i = 1, \dots, n$, $Y_i| \bX_i, T_i = 1, \gamma_1, \theta_1, \theta_2,\tau \iid \calN(\gamma_1 + \bX_i'\theta_1 + (\bX_i^2)'\theta_2, \tau^2)$ where the inclusion of $(\bX_i^2)'\theta_2$ allows the model to capture the nonlinear effects. This quadratic model represents the {\it correctly specified} nuisance setting, since the true regression function $m_1^\d(\cdot)$ is quadratic.

As in Section~\ref{sec_sim_correctly_specified}, we fit this model using two Bayesian approaches. The Bayesian ridge quadratic regression ({\tt BR-q}) extends the Bayesian ridge regression ({\tt BR}) method to include quadratic terms, with Gaussian priors on $(\gamma_1, \theta_1, \theta_2)$ and an improper prior on $\tau^2$, where the ridge parameter $\lambda$ is estimated using an empirical Bayes approach, with the point estimate $\widehat{\lambda}$ obtained via the \texttt{glmnet} R package. Similarly, the Bayesian sparse quadratic regression ({\tt BS-q}) extends the sparse Bayesian linear regression ({\tt BS}) method by including both linear and quadratic covariates $(\bX, \bX^2)$, while employing the same nonlocal prior structure implemented via the {\tt mombf} R package.

For comparison, the original linear models {\tt BS} and {\tt BR}, which omit the quadratic terms, represent {\it misspecified} nuisance settings, as the posterior $\Pi_{m_1}$ now contracts to a limiting function $m^*_1$, which may differ from the true nonlinear regression function $m_1^\d$. In contrast, the quadratic ({\tt BS-q}, {\tt BR-q}) and nonparametric ({\tt BART}) models correctly specify $m_1^\d$, and therefore provide consistent
nuisance estimates. To distinguish between the resulting DRDB approaches, we use the following {\it notation:} DRDB-S and DRDB-R correspond to DRDB with misspecified linear nuisance models ({\tt BS} and {\tt BR}); DRDB-Sq and DRDB-Rq correspond to DRDB with correctly specified quadratic nuisance models ({\tt BS-q} and {\tt BR-q}); and DRDB-B corresponds to DRDB with the nonparametric {\tt BART} model. This setup allows a clear comparison of DRDB's performance under misspecified (DRDB-S, DRDB-R) versus correctly specified (DRDB-Sq, DRDB-Rq, DRDB-B) nuisance models.

Table~\ref{table_misspecified} reports the results for the data-generating mechanisms presented in Section~\ref{sec_sim_correctly_specified}. For DRDB-Sq, DRDB-Rq, and DRDB-B, where the nuisance model for $m_1$ is correctly specified (quadratic or nonparametric), the results closely match those in Section~\ref{sec_sim_correctly_specified}, with low bias, low MSE, coverage probabilities near the nominal 95\% level, and near-optimal credible interval (CI) lengths across all settings, supporting the theoretical properties of DRDB in Theorem~\ref{thm_BvM_ATE}.

\begin{table}[!ht]
\caption{Estimation and inference results for the ATE based on DRDB under the misspecified nuisance model setting given in Section~\ref{supp_sec_sim_misspecified}, where
$n = 1000$, $p \in \{10, 50, 200\}$ with $s \approx \sqrt{p}$ or $s \approx p/4$ and $\bbK = 5$. For $p = 10$, we set $s \approx \sqrt{p}$ or $s = p$. Each variant of DRDB are denoted by DRDB-M, where ``M'' indicates the specific nuisance method and fitted model: S = {\tt BS}, Sq = {\tt BS} with a quadratic model, R = {\tt BR}, Rq = {\tt BR} with a quadratic model, and B = {\tt BART}. \label{table_misspecified}}
\begin{tabular*}{\columnwidth}{@{\extracolsep\fill}lcccccccc@{\extracolsep\fill}}
\toprule
{\bf Method} & {\bf Bias} & {\bf MSE}
& {\bf Cov} & {\bf CI-Len} & {\bf Bias} & {\bf MSE}
& {\bf Cov} & {\bf CI-Len} \\
\midrule
\multicolumn{1}{@{}l@{}}{$p = 10$} & \multicolumn{4}{@{}c@{}}{$s = 3$} & \multicolumn{4}{@{}c@{}}{$s = 10$} \\
\cline{1-1}\cline{2-5}\cline{6-9}
Oracle   & 0.000 & 0.009 & 0.948 & 0.371   & 0.005 & 0.023 & 0.950 & 0.617 \\
DRDB-S   & 0.006 & 0.011 & 0.972 & 0.429   & 0.017 & 0.034 & 0.972 & 0.723 \\
DRDB-Sq  & 0.000 & 0.009 & 0.944 & 0.380   & 0.004 & 0.025 & 0.962 & 0.642 \\
DRDB-R   & 0.006 & 0.011 & 0.964 & 0.426   & 0.019 & 0.033 & 0.960 & 0.713 \\
DRDB-Rq  & -0.001 & 0.009 & 0.952 & 0.385  & 0.003 & 0.026 & 0.942 & 0.643 \\
DRDB-B   & -0.002 & 0.010 & 0.964 & 0.434  & -0.001 & 0.033 & 0.968 & 0.789 \\
\multicolumn{1}{@{}l@{}}{$p = 50$} & \multicolumn{4}{@{}c@{}}{$s = 7$} & \multicolumn{4}{@{}c@{}}{$s = 13$} \\
\cline{1-1}\cline{2-5}\cline{6-9}%
Oracle   & 0.003 & 0.019 & 0.940 & 0.539   & -0.004 & 0.033 & 0.952 & 0.723 \\
DRDB-S   & 0.017 & 0.026 & 0.940 & 0.628   & 0.018 & 0.046 & 0.952 & 0.860 \\
DRDB-Sq  & 0.003 & 0.021 & 0.938 & 0.559   & -0.003 & 0.035 & 0.960 & 0.760 \\
DRDB-R   & 0.015 & 0.030 & 0.964 & 0.686   & 0.015 & 0.048 & 0.954 & 0.917 \\
DRDB-Rq  & 0.001 & 0.024 & 0.964 & 0.657   & -0.001 & 0.043 & 0.960 & 0.877 \\
DRDB-B   & 0.001 & 0.027 & 0.962 & 0.693   & -0.005 & 0.051 & 0.962 & 0.997 \\
\multicolumn{1}{@{}l@{}}{$p = 200$} & \multicolumn{4}{@{}c@{}}{$s = 14$} & \multicolumn{4}{@{}c@{}}{$s = 50$} \\
\cline{1-1}\cline{2-5}\cline{6-9}%
Oracle   & 0.005 & 0.034 & 0.958 & 0.731   & 0.018 & 0.131 & 0.954 & 1.383 \\
DRDB-S   & 0.010 & 0.046 & 0.970 & 0.873   & 0.049 & 0.227 & 0.976 & 2.187 \\
DRDB-Sq  & 0.004 & 0.037 & 0.968 & 0.771   & 0.013 & 0.168 & 0.974 & 1.848 \\
DRDB-R   & 0.050 & 0.064 & 0.982 & 1.215   & 0.089 & 0.248 & 0.986 & 2.308 \\
DRDB-Rq  & 0.121 & 0.109 & 0.970 & 1.350   & 0.180 & 0.348 & 0.972 & 2.544 \\
DRDB-B   & 0.014 & 0.051 & 0.982 & 1.057   & 0.062 & 0.309 & 0.978 & 2.634 \\
\bottomrule
\end{tabular*}
\end{table}

Examining the misspecified cases, DRDB-S and DRDB-R, which use linear models for a truly quadratic $m_1$, estimation accuracy remains reasonable in low and moderate dimensions ($p = 10$ and $p = 50$). As dimension and sparsity increase ($p = 200$), bias and MSE increase modestly compared to the Oracle and correctly specified variants (DRDB-Sq, DRDB-Rq, and DRDB-B), reflecting the effects of model misspecification and finite-sample error, consistent with theoretical expectations. {\it Despite misspecification}, DRDB-S and DRDB-R maintain {\it valid inference}, with coverage probabilities consistently at or above the nominal 95\% level across all settings. However, the CIs are consistently wider for DRDB-S and DRDB-R, compared to their correctly specified counterparts, particularly in high-dimensional settings, reflecting increased finite-sample nuisance estimation error and the effects of model misspecification. Overall, these results provide empirical support for the double robustness property of DRDB: when only one model is correctly specified, DRDB {\it still} yields reliable {\it consistent estimation}, and also {\it valid inference}, albeit with wider CIs.

\section{Preliminary results}\label{supp_sec_preliminary}
{\bf Notational conventions.} Throughout this supplement, we use the following notation. For a sequence $b_n > 0$ and a sequence of random variables $W_n$, we write $W_n = \op(b_n)$ if and only if (iff) $|W_n|/b_n \cvP 0$ as $n \to \infty$. In particular, if $W_n \cvP 0$, we denote this as $W_n = \op(1)$. Similarly, $W_n = O_{\bbP}(b_n)$ iff for every $\varepsilon > 0$, there exist constants $B_\varepsilon > 0$ and $n_\varepsilon$ such that $\bbP(|W_n| \leq B_{\varepsilon} \hspace{0.5mm} b_n) > 1 - \varepsilon$ for all $n \geq n_{\varepsilon}$. Moreover, $W_n = \op(1)$ iff there exists a sequence $b_n \to 0$, $W_n = \Op(b_n)$. Finally, we follow the standard empirical process notation \citep{van2000asymptotic}: For any function $f(\cdot) \in \bbL_2(\bW)$ of the random variable $\bW$, we define the empirical mean operator on the dataset $\calD$ with index set $\calI$ of size $N = |\calI|$ as $\bbP_N\{f(\bW)\} := N^{-1} \sum_{i \in \calI} f(\bW_i)$.

\begin{lemma}\label{lemma_TV_distance_btw_two_posteriors} For $i = 1,2$,
    let $f_i(\cdot)$ denoted the pdfs of the corresponding distribution $\calP_i$ satisfying
    \begin{align*}
        f_i(\rATE) ~:=~ \int q_i(\rATE \mid \theta) \hspace{0.5mm}g_i(\theta) \hspace{0.5mm}\dd \theta,
    \end{align*}
where $g_i(\cdot)$ is the pdf of the distribution $\calP_\theta^{(i)}$ of $\theta$ and $q_i(\cdot \mid \theta)$ is the pdf of the conditional distribution $\calP_{\rATE \mid \theta}^{(i)}$ of $\rATE$ given $\theta$ which belongs to a location family generated by a density $\psi_i(\cdot)$ with $\theta + C$ being the location parameter for some $C \in \bbR$. Then,
    \begin{align*}
       d_{\TV}(\calP_1, \calP_2) ~\leq ~ d_{\TV}(\calP_{\rATE \mid \theta}^{(1)}, \calP_{\rATE \mid \theta}^{(2)}) \, + \, d_{\TV}(\calP_\theta^{(1)}, \calP_\theta^{(2)})
    \end{align*}
\end{lemma}

\begin{lemma}\label{lemma_TV_convolution}
Let $P$, $Q$ be probability distributions (pds) with densities $p(\cdot)$, $q(\cdot)$ where $
p(x)=(p_1 * p_2)(x)$ and $q(x)= (q_1 * q_2)(x)$, for some pdfs $p_i(\cdot), q_i(\cdot)$ of pds $P_i, Q_i$ for $i = 1, 2$, and $*$ denotes the convolution operator. Then, $
d_{\TV}( P, Q) ~\leq~ d_{\TV}( P_1, Q_1) + d_{\TV}( P_2, Q_2)$.
\end{lemma}

\begin{lemma}\label{lemma_TV_bound_theta_as_eta_minus_beta}
Let $\theta = \eta - \beta$ with $\eta \perp\!\!\!\perp \beta$. For $i=1,2$, let $\eta^{(i)} \sim \mathcal{P}_\eta^{(i)}$ and $\beta^{(i)} \sim \mathcal{P}_\beta^{(i)}$ be independent, define $\theta^{(i)} := \eta^{(i)} - \beta^{(i)}$ with law $\mathcal{P}_\theta^{(i)} = \mathcal{P}_\eta^{(i)} * \mathcal{P}_{-\beta}^{(i)}$. Suppose the conditional laws $\mathcal{P}_{\rATE\mid\theta}^{(i)}$ admits densities $q_i(\rATE\mid\theta)$ belonging to a location family generated by $\psi_i$, i.e.,
\(
q_i(\rATE\mid\theta) \;=\; \psi_i(\rATE - \theta - C),
\)
for a fixed $C \in \mathbb{R}$. Let $\mathcal{P}_i$ denote the distribution of $\rATE$ with density
$f_i(\rATE) = \int q_i(\rATE\mid\theta)\, g_i(\theta)\, d\theta$,
where $g_i$ is the density of $\mathcal{P}_\theta^{(i)}$ for $i = 1, 2$. Then,
\[
d_{\mathrm{TV}}(\mathcal{P}_1,\mathcal{P}_2)
~\le~
d_{\mathrm{TV}}(\mathcal{P}_\eta^{(1)},\mathcal{P}_\eta^{(2)})
~+~
d_{\mathrm{TV}}(\mathcal{P}_\beta^{(1)},\mathcal{P}_\beta^{(2)})
~+~
d_{\mathrm{TV}}(\mathcal{P}_{\rATE\mid\theta}^{(1)},\mathcal{P}_{\rATE\mid\theta}^{(2)}).
\]
\end{lemma}
\begin{lemma}[Nuisance contraction rates]\label{lemma_nuisance_convergence_rates}
Under Assumption~\ref{assumption_nuisance} and the setups of Theorems \ref{thm_BvM_mu1} and \ref{thm_BvM_ATE},
\begin{align*}
   & \| \underbar{\it m}_t(\bX) - m_t^*(\bX) \|_{\bbL_2(\bbP_{\bX})} ~=~ O_{\bbP_{m_t}}(\varepsilon_{m,n}) \quad \text{and }\\
   & \|\underbar{\it r}_t(\bX) - r_t^*(\bX)\|_{\bbL_2(\bbP_{\bX})} ~=~  O_{\bbP_{r_t}}(\varepsilon_{r,n}).
\end{align*}
where for $t = 0,1$, $m^*_t(\cdot)$ and $r^*_t(\cdot)$ are non-random limiting functions around which the nuisance posteriors $\Pi_{m_t}$ and $\Pi_{r_t}$ contract, respectively.
\end{lemma}

\begin{lemma}[Convergence of posterior variance of $\Pi_{\mu_1}$.]\label{lemma_pvar_convergence} Under the assumptions and setup of Theorem~\ref{thm_BvM_mu1}, and assuming that Case~\textbf{C1} holds, the posterior variance $\hat{c}_n^2(\m, \r)$ of $\Pi_{\mu_1}$ converges in probability to the variance $c^2(m_1^\d, r_1^\d)$ of the limiting distribution at rate $1/n$:
\begin{align}
    n\big\{\hat{c}_n^2(\m, \r) - c^2(m_1^\d, r_1^\d)\big\} ~=~ \op(1).\label{eqn_pvar_convergence_for_mu1}
\end{align}
\end{lemma}

\begin{corollary}[Convergence of posterior variance of $\pATE$.]\label{cor_pvar_convergence_for_ATE} Under the assumptions and setup of Theorems~\ref{thm_BvM_mu1}--\ref{thm_BvM_ATE}, and Case~\textbf{C1}, the posterior variance $\hat{c}_n^2(\bm, \br)$ of $\pATE$ converges in probability to the variance $c^2(m^\d, r^\d)$ of the limiting distribution at rate $1/n$:
\begin{align}
    n\big\{\hat{c}_n^2(\bm, \br) - c^2(m^\d, r^\d)\big\} ~=~ \op(1).\label{eqn_pvar_convergence_for_ATE}
\end{align}
\end{corollary}

\section{Proofs of the main results} \label{sup_sec_proofs_of_main_results}

\phantomsection
\addcontentsline{toc}{subsection}{Proof of Propositions~\ref{prop_posterior_b1} and \ref{prop_posterior_mu1_given_b1}}

\begin{proof}[Proof of Propositions~\ref{prop_posterior_b1} and \ref{prop_posterior_mu1_given_b1}]
Both results follow directly from Bayes' theorem, applied to the Normal likelihood with an improper prior on the mean and variance parameters, which yields a $t$-distribution for the posterior of the mean. Since both propositions rely on Normal likelihoods with the same form of prior, we present a unified proof.

For notational simplicity, let $\underline{\eta}$ be a sample from the nuisance posterior $\Pi_{\eta}$ and define $W(\bZ,\underline{\eta})$ as the observable from the data $\mathcal{L}$ with $N = |\mathcal{L}|$. For Proposition~\ref{prop_posterior_b1}, we have $\underline{\eta} = (\mt, \rt)$, $W(\bZ,\underline{\eta}) = W(\bZ,\mt, \rt) = \rt(\bX)\{Y - \mt(\bX)\}$ and $\mathcal{L} = S_t$ for $t = 0, 1$ where $\bZ = (Y, \bX)$. For Proposition~\ref{prop_posterior_mu1_given_b1}, we have $\underline{\eta} = \underline{\varphi}$, $W(\bZ,\underline{\eta}) = W(\bZ,\underline{\varphi}) = \underline{\varphi}(\bX)$ and $\mathcal{L} = S$ where $\bZ = \bX$.

In both settings, the observables $W(\bZ_i, \underline{\eta})$ are modeled as $W(\bZ_i, \underline{\eta}) | \delta, \tau^2, \underline{\eta} \sim \calN(\delta, \tau^2)$ for $i = 1, \dots ,N$, with improper prior $\pi(\delta | \tau^2) \propto 1, \ \pi(\tau^2) \propto (\tau^2)^{-1}$.

By Bayes' theorem, the joint posterior for $(\delta, \tau^2)$ is given by:
\begin{align*}
\pi(\delta, \tau^2 | \mathcal{L}) & ~\propto~ \frac{1}{(\tau^2)^{N/2}}\exp\!\left[-\frac{1}{2\tau^2} \sum_{i = 1}^N \{ W(\bZ_i, \underline{\eta}) \}^2\right] (\tau^2)^{-1} \\
& ~\propto~  \frac{1}{(\tau^2)^{N/2 + 1}}\exp\!\Big(\!-\frac{\mathcal{S} + N(\delta-\overline{W})^2}{2\tau^2}\Big),
\end{align*}
where the parameters are defined as $\overline{W} := N^{-1}\sum_{i=1}^N W(\bZ_i, \underline{\eta})$ and $\mathcal{S} := \sum_{i=1}^N\{W(\bZ_i, \underline{\eta})-\overline{W}\}^2$.

By integrating out $\tau^2$, we obtain the marginal posterior of $\delta$: $\pi(\delta|\mathcal{L}) \propto (\mathcal{S} + N(\delta - \overline{W})^2)^{-N/2}$. Notably, this expression has exactly the kernel of a $t$-distribution. Therefore, $\delta | \mathcal{L} \sim t_{\nu}(\theta, c^2)$ where the degrees of freedom $\nu = N - 1$, the location parameter $\theta = \overline{W}$ and the scale parameter $c^2 = \mathcal{S}/\{N(N - 1)\}$, which establishes the desired result. For additional derivation details, we refer to Section~3.2 of \cite{gelman2014}.
\end{proof}

\phantomsection
\addcontentsline{toc}{subsection}{Proof of Theorem~\ref{thm_BvM_mu1}}
\begin{proof}[Proof of Theorem~\ref{thm_BvM_mu1}.]
We divide the proof of Theorem~\ref{thm_BvM_mu1} into two parts. First, we establish the Bernstein-von Mises (BvM) result under the correct specification of both nuisance models. Subsequently, we present a unified proof of the posterior concentration result when only one of the nuisance models is correctly specified.

{\bf Proof of the BvM result:} Let $\mu_1^{(1)}, \dots, \mu_1^{(\bbK)}$ be independent samples from the corresponding posteriors $\{\Pi_{\mu_1}^{(k)}\}_{k = 1}^\bbK$. Define the random variable $\tilde\mu_1 := \sum_{k = 1}^\bbK \mu_1^{(k)}$ and let $\widetilde \Pi_{\mu_1}$ denote its posterior. By the location-scale invariance of TV distance, proving the BvM result in Theorem~\ref{thm_BvM_mu1} is equivalent to showing that $d_\TV\!\left(\widetilde \Pi_{\mu_1}, \calN(\bbK \mu_1(m_1^\d, r_1^\d), \bbK^2 c^2(m_1^\d, r_1^\d))\right) \cvP 0$ as $n \to \infty$.

Since $\mu_1^{(1)}, \dots, \mu_1^{(\bbK)}$ are independent, by construction, $\widetilde \Pi_{\mu_1}$ is the convolution of $\Pi_{\mu_1}^{(1)} , \dots, \Pi_{\mu_1}^{(\bbK)}$. By Lemma~\ref{lemma_TV_convolution}, we have
$$
\mathbb{T}~:=~ d_\TV\left(\widetilde \Pi_{\mu_1}, \, \calN(\bbK \mu_1(m_1^\d, r_1^\d), \bbK^2 c^2(m_1^\d, r_1^\d))\right) ~\leq~ \mathbb{T}_1 + \cdots + \mathbb{T}_{\bbK},$$
where, for $k = 1, \dots, \bbK$, $\mathbb{T}_k := d_\TV\left(\Pi_{\mu_1}^{(k)}, \ \calN(\mu_1^{(k)}(m_1^\d, r_1^\d), c^{2, (k)}(m_1^\d, r_1^\d))\right)$ is the TV distance between the DRDB posterior $\Pi_{\mu_1}^{(k)}$ based on the $k$-th split $(\calD_k, \calD_k^\-)$ and the corresponding limiting distribution $\calN(\mu_1^{(k)}(m_1^\d, r_1^\d), c^{2, (k)}(m_1^\d, r_1^\d))$, with parameters analogously to those in Equation~\ref{eqn_pmean_pvar_for_mu1} in Section~\ref{sec_main_results}. Thus, the proof of Theorem~\ref{thm_BvM_mu1} (a) reduces to showing $\mathbb{T}_k \cvP 0$ for each $k \in \{1, \dots, \bbK\}$.

\underline{Roadmap of the BvM proof.} {\it The proof proceeds by replacing the conditional and marginal posteriors $\Pi_{\mu_1 \mid b_1}$ and $\Pi_{b_1}$, shown in Propositions~\ref{prop_posterior_b1} and \ref{prop_posterior_mu1_given_b1} to be $t$-distributions with Gaussian proxies $Q_{\mu_1|b_1}$ and $Q_{b_1}$ that share their location and scale parameters. This yields a Gaussian marginal posterior $Q_{\mu_1}$ for $\mu_1$. By adding and subtracting $Q_{\mu_1}$, we decompose $\mathbb{T}_1$ in \eqref{eqn_original_TV_for_mu1} into two terms $R$ and $T$ as given in \eqref{eqn_original_TV_for_mu1}. The term $T$, comparing two Gaussians, reduces to bounding differences in means and variances, which follow from Corollary~\ref{cor_pmean_convergence} and Lemma~\ref{lemma_pvar_convergence}. The main challenge lies in controlling $R$, which arises from the hierarchical structure of $\Pi_{\mu_1}$ as defined in \eqref{eqn_posterior_for_mu_for_S}, where the TV distance itself involves an embedded integral. Traditional approaches for controlling such distances usually necessitate strong, often impractical, constraints on the prior or posterior. However, this is precisely where the strengths of our debiasing mechanism and targeted modeling strategy become essential. More specifically, since the bias $b_1$ appears is a part of the location parameter in both $\Pi_{\mu_1|b_1}$ and $Q_{b_1}$, we can further bound $R$ as a sum of two TV distances between Normal and $t$-distributions with matched centers and scales, referring to \eqref{eqn_bound_for_R_for_mu1}. This allows us to exploit known TV distance bounds between Normal and  $t$-distribution pairs. The remainder of the proof follows by carefully applying Assumptions~\ref{assumptions_standard_causal_assump} and \ref{assumption_nuisance}, along with several key preliminary results from Section~\ref{supp_sec_preliminary}. Below, we present a detailed and rigorous argument for each step.}

Without loss of generality, we take $k = 1$. Set $S := \calD_1$ with $n_S := |S| = n/\bbK$ and $S^\- = \calD \setminus \calD_1$. We prove the case $k = 1$, as other cases follow by the same argument.

For notational simplicity, we drop the superscript in $ \calN(\mu_1^{(1)}(m_1^\d, r_1^\d), c^{2,{(1)}}(m_1^\d, r_1^\d))$ and $\Pi_{\mu_1}^{(1)}$ and accordingly, write $\mathbb{T}_1 = d_\TV\left(\Pi_{\mu_1}, \calN(\mu_1(m_1^\d, r_1^\d), c^{2}(m_1^\d, r_1^\d))\right).$
The first step of the proof is to replace the posteriors $\Pi_{\mu_1 \mid b_1}$ and $\Pi_{b_1}$ with Normal distributions whose location and scale parameters match those of their corresponding $t$-distribution forms. Under this substitution, the marginal posterior for $\mu_1$ becomes a Normal distribution constructed in direct analogy to Step~\eqref{eqn_integral_representation_of_marginal_post_for_mu1_with_b1m}. Specifically, define $Q_{\mu_1|b_1} := \calN(\eta_S, c_S^2)$ and $Q_{b_1} := \calN(\eta_1, c_1^2)$, where the parameters are given in \eqref{eqn_marginal_posterior_b1} and \eqref{eqn_conditional_posterior_mu1_given_b1}. Then, the marginal distribution for $\mu_1$ is $Q_{\mu_1} =  \calN(\eta_1 + \eta_{m_1}, c^2_1 + c^2_S) := \calN(\mu_1(\m, \r), c^{2}(\m, \r))$ with parameters defined analogously to those in Section~\ref{sec_likelihood_and_posterior_calc}.

By the triangle inequality,
\begin{align}
    \mathbb{T}_1 ~\leq~ & ~ d_\TV\left(\Pi_{\mu_1}, Q_{\mu_1}\right) \,+\,  d_\TV( Q_{\mu_1}, \calN(\mu_1(m_1^\d, r_1^\d), c^{2}(m_1^\d, r_1^\d))) ~:=~ R + T. \label{eqn_original_TV_for_mu1}
\end{align}

\noindent{\bf Analysis of $T= d_\TV(Q_{\mu_1}, \calN(\mu_1(m_1^\d, r_1^\d), c^{2}(m_1^\d, r_1^\d)))$.}\\
Since both distributions are Normal, we obtain the following upper bound for $T$: For some universal constant $C < \infty$,
\begin{align}
    T & ~\leq C \left| \frac{c^{2}(\m, \r) - c^{2}(m_1^\d, r_1^\d)}{c^{2}(m_1^\d, r_1^\d)} \right|
        + \left| \frac{\mu_1(\m, \r) - \mu_1(m_1^\d, r_1^\d)}{\sqrt{c^{2}(m_1^\d, r_1^\d)}} \right| \nonumber \\
& ~ := C V(\m,\r) + \Omega(\m,\r).
    \label{eqn_upper_bound_for_T}
\end{align}
Since $\sqrt{n_S\, c^{2}(m_1^\d, r_1^\d)} \geq \delta$ for some $\delta > 0$ by definition, the analysis of $\Omega(\m,\r)$ reduces to showing that
$\big| \sqrt{n_S}\,\{\mu_1(\m, \r) - \mu_1(m_1^\d, r_1^\d)\} \big| \to 0$
in $\bbP_\calD$-probability. Further, since $\mu_1(\m, \r)$ is the mean of $\Pi_{\mu_1}$, Corollary~\ref{cor_pmean_convergence} and its proof applied to the test-training pair $(S, S^\-)$ yields $| \sqrt{n_S}\{\mu_1(\m, \r) - \mu_1(m_1^\d, r_1^\d)\}| = o_{\PD}(1)$, establishing the claim.

A similar argument applies to $V(\m,\r)$. Since $n_S\, c^{2}(m_1^\d, r_1^\d) \geq \delta > 0$ by definition, establishing $V(\m,\r) = o_{\PD}(1)$ is equivalent to showing
\(\big| n_S\,\{c^2(\m, \r) - c^{2}(m_1^\d, r_1^\d)\} \big| \to 0
\)
in $\PD$-probability. This follows directly from Lemma~\ref{lemma_pvar_convergence} applied to $(S, S^\-)$, in complete analogy with the posterior variance case. Thus,
\(\big| n\,\{c^2(\m, \r) - c^{2}(m_1^\d, r_1^\d)\} \big| = o_{\bbP}(1)
\). Combining these results with inequality in \eqref{eqn_upper_bound_for_T}, we conclude that $T \to 0$ in probability under $\PD$.

\noindent{\bf Analysis of $R := d_\TV\left(\Pi_{\mu_1}, Q_{\mu_1}\right)$.}\\
Using the construction of the DRDB posterior $\Pi_{\mu_1}$ and the distribution $Q_{\mu_1}$, we note that the location parameter $\ell(b_1)$ of both $\Pi_{\mu_1|b_1}$ and $Q_{\mu_1|b_1}$ are equivalent. In both cases, it has the form $b_1 + \omega$ for some $\omega \in \mathbb{R}$ given the data $(S, S^\-)$, specifically, $\ell(b_1) := b_1 + \eta_{m_1}$. Hence, by directly applying Lemma~\ref{lemma_TV_distance_btw_two_posteriors}, we obtain the decomposition:
\begin{equation}
    R ~\leq~ d_\TV(Q_{\mu_1 \mid b_1}, \Pi_{\mu_1\mid b_1}) \, + \, d_\TV(Q_{b_1}, \Pi_{b_1}) ~:=~ R_2 + R_1. \label{eqn_bound_for_R_for_mu1}
\end{equation}
Moreover, by construction, $Q_{\mu_1 \mid b_1}$ and $Q_{b_1}$ are Normal distributions whose location and scale parameters {\it match} those of $\Pi_{\mu_1 \mid b_1}$ and $\Pi_{b_1}$, as given in Section~\ref{sec_likelihood_and_posterior_calc}. By the {\it location-scale invariance} of the TV distance (refer to Lemma~S4.1 of \cite{sert2025}) applied to both $R_2$ and $R_1$, we have
\begin{align*}
R_2 & ~=~  d_\TV(t_{\nu_S}(\eta_S, c^2_S), \,\calN(\eta_S, c^2_S) )
    ~=~ d_\TV(
        \calN(0,1), \,
        t_{\nu_S}(0,1)
        );
~~~~\mbox{and} \\
R_1 & ~=~ d_\TV\left(t_{\nu_1}(\eta_1,c^2_1) , \, \calN(\eta_1, c^2_1) \right)
    ~=~ d_\TV\left(
    \calN(0,1), \, t_{\nu_1}(0,1)
       \right).
\end{align*}
By directly applying Lemma S4.4 from \cite{sert2025}, which gives an explicit bound for the TV distance between a Normal and a $t$ distribution with matching location and scaling parameters, we obtain that for universal some constants $C_1, C_S < \infty$, $R_1 \leq C_1/ \sqrt{\nu_1}$ and $R_2 = C_S/\nu_S$.

Note that as $n\to \infty$, for fixed $\bbK$, $\nu_S \to \infty$ by definition, which implies $R_2 \to 0$. Moreover, by definition, $\nu_1  = n_1 - 1 = \sum_{i = 1}^n T_i -1$. Defining $\ptilde_n:=  n_S^{-1}\sum_{i \in \calI} T_i$, we can rewrite $\nu_1 = n_S \ptilde_n -1$. As $n \to \infty$ with $\bbK$ fixed, we have $\ptilde_n \to p_1 = \bbP(T = 1)$ in $\PD$-probability with $p_1 > 0$ by Assumption~\ref{assumptions_standard_causal_assump}. Thus, as $n \to \infty$, $\nu_1 \to \infty$ which in turn implies $R_1 \to 0$ in probability.

Combining this with the convergence in probability result for $T$, and recalling that $\mathbb{T}_1$ is bounded by the sum of $T$ and $R$, we conclude that
\(
    \mathbb{T}_1 \to 0
\)
in probability under $\bbP_\calD$.

Following the same steps and applying the same proof strategy, we obtain for each $k = 1, \dots, \bbK$, $\mathbb{T}_k \cvP 0$, which finally establishes the BvM result stated in Theorem~\ref{thm_BvM_mu1}.

\smallskip
\noindent\textbf{Proof of posterior contraction results.}
For brevity, we provide a unified proof covering both one-correctly-specified nuisance cases: either the regression function is well-specified (Case \textbf{C2}: $m_1^* = m_1^\d$) or the density ratio is well-specified (Case \textbf{C3}: $r_1^* = r_1^\d$). The argument is identical in both cases, with the contraction rate set to $\epsilon_n = \varepsilon_{m,n}$ in Case \textbf{C2} and $\epsilon_n = \varepsilon_{r,n}$ in Case \textbf{C3}.

First, by the construction of $\mu_1^\CF$ in \eqref{eqn_CF_version_mu1}, for any $\epsilon > 0$ and fixed $\bbK < \infty$, we have:
\[
\Pi_{\mu_1}^\CF\big(|\mu_1 - \mu^\d(1)| > \bbK \epsilon | \calD\big) ~\leq~ \Pi_{\mu_1}^{(1)}\big(|\mu_1 - \mu^\d(1)| > \epsilon | \calD\big)  + \cdots + \Pi_{\mu_1}^{(\bbK)}\big(|\mu_1 - \mu^\d(1)| > \epsilon | \calD\big).
\]
This decomposition suggests that it suffices to prove the desired result for
\(
    \Pi_{\mu_1}^{(k)}\big(|\mu_1 - \mu^\d(1)| > \epsilon \mid \calD\big)
\)
for an arbitrary $k \in \{1,\dots,\bbK\}$.

By suppressing the superscript for simplicity, to show posterior contraction results in Theorem~\ref{thm_BvM_mu1}, we aim to prove that, for any $M_n \to \infty$, $\Pi_{\mu_1}\big(|\mu_1 - \mu^\d(1)| > M_n \epsilon_n \mid \calD\big) \;\to\; 0$, in $\bbP_\calD$-probability, where $\epsilon_n >0$ denotes the contraction rate of the correctly specified nuisance parameter.
In particular, $\epsilon_n = \varepsilon_{m,n}$ in Case \textbf{C2}, and $\epsilon_n = \varepsilon_{r,n}$ in Case \textbf{C3}.

Let $\muhat_1 \equiv \mu_1(\m, \r)$ denote the posterior mean of $\Pi_{\mu_1}$. By the triangle inequality, we obtain that:
\begin{align*}
 \Pi_{\mu_1}\big(|\mu_1 - \mu^\d(1)| > 2M_n \epsilon_n |  \calD\big) & ~\leq~ \Pi_{\mu_1}\big(|\mu_1 - \muhat_1| > M_n \epsilon_n | \calD\big) \\
 & ~~~~ + \Pi_{\mu_1}\big(|\muhat_1 - \mu^\d(1)| > M_n \epsilon_n | \calD\big)  \nonumber \\
 & ~:=~ V(\m, \r, S) + \mathbb{D}(\m, \r, S),
\end{align*}
where $S = \calD_k$ for some $k \in \{1,\dots,\bbK\}$. Thus, it suffices to establish that both terms $V(\m,\r, S)$ and $\mathbb{D}(\m,\r, S)$ converge to zero in probability.

First, note that the second probability $\mathbb{D}(\m,\r,S)$ is the indicator of the event $\{|\muhat_1 - \mu^\d(1)| > M_n \epsilon_n\}$. In Corollary~\ref{cor_pmean_convergence} (b), applied to the split $(S, S^\-) = (\calD_k, \calD_k^\-)$ with properly adjusting the parameters, we have already established that the posterior mean $\muhat_1$ converges to $\mu^\d(1)$ at the rate $\epsilon_n$, where $\epsilon_n = \varepsilon_{m,n}$ in Case \textbf{C2}, and $\epsilon_n = \varepsilon_{r,n}$ in Case \textbf{C3}. Hence, Corollary~\ref{cor_pmean_convergence} (b) directly gives:
\begin{align}
    \mathbb{D}(\m,\r,S) \;\to\; 0
    \; \text{ in probability under } \bbP_\calD.
    \label{eqn_post_contraction_pmean_part}
\end{align}
Next, consider the first probability $V(\m,\r,S)$. By Chebyshev’s inequality,
\begin{align}
    V(\m,\r,S)
    ~\equiv~ \Pi_{\mu_1}\big(|\mu_1 - \muhat_1| > M_n \epsilon_n \mid \calD\big) ~\leq~ (M_n \epsilon_n)^{-2} \Var(\mu_1 \mid \calD).
    \label{eqn_V_bound_by_variance}
\end{align}
Thus, since $M_n \to \infty$, it suffices to show that  $\Var(\mu_1 \mid \calD) = \Op(\delta_n^2)$ under both Cases \textbf{C2} and \textbf{C3} where $\delta_n \to 0$ faster than or equal to $\epsilon_n$.

To this end, recalling the construction of the posterior $\Pi_{\mu_1}$ in Section~\ref{sec_DRDB_for_mu1} with parameters given in Section~\ref{sec_likelihood_and_posterior_calc}, we explicitly calculate the posterior variance $\hat{c}_n^2(\m,\r)$. Specifically, given $\m \sim \Pi_{m_1}$ and $\r \sim \Pi_{r_1}$, and $ W(\bZ,\m,\r) := \r(\bX)\{Y - \m(\bX)\}$, we obtain:
\begin{align*}
    \hat{c}_n^2(\m,\r)
    & = \frac{n_S}{n_S-2} c_S^2 + \frac{n_1}{n_1-2} c_1^2 ~:=~ \lambda_n c_S^2 + \lambda_{n_1} c_1^2 \\
    & = \frac{\lambda_n}{n_S(n_S-1)} \sum_{i \in \calI} \{\m(\bX_i) - \eta_{m_1}\}^2
      +  \frac{\lambda_{n_1}}{n_1(n_1-1)}\! \sum_{i \in \calI_1} \{ W(\bZ_i, \m,\r) - \eta_1 \}^2,
\end{align*}
where $n_S = |\calI|$, $n_1 = |\calI_1|$, $\lambda_n = n_S/(n_S-2)$, $\lambda_{n_1} = n_1/(n_1-2)$, and $n_1 = \sum_{i \in \calI_1}T_i$.

Since $n_S \to \infty$ (as $\bbK$ is fixed), we have $\lambda_n \to 1$. Moreover, since $n_1/n \to p_1 \in (0,1)$, we obtain $\lambda_{n_1} \to 1$ in $\PD$-probability; in particular, $\lambda_{n_1} = \Op(1)$.

Next, for the first component $c_S^2$, we observe that for any $t > 0$, Markov's inequality
\[
\bbP_S(c_S^2 > t | \m) ~\leq~ t^{-1}\,\bbE_S(c_S^2 | \m)
      ~=~ n_S^{-1}\Var_\bX\{\m(\bX)| \m\} ~\leq~ n_S^{-1}\,\bbE_\bX\{\m^2(\bX)| \m\},
\]
where the equality in the first step follows from the definition of $c_S^2$.

To bound $\bbE_\bX\{\m^2(\bX)\mid \m\}$, we write $\m(\bX) = \{\m(\bX) - m_1^*(\bX)\} + m_1^*(\bX)$. Then, the inequality $(a+b)^2 \le 2a^2 + 2b^2$ gives $\bbE_\bX\{\m^2(\bX) | \m\} \le 2\|\m(\bX) - m_1^*(\bX)\|_{\bbL_2(\bbP_\bX)}^2
        + 2\|m_1^*(\bX)\|_{\bbL_2(\bbP_\bX)}^2$. Since $\|\m(\bX) - m_1^*(\bX)\|_{\bbL_2(\bbP_\bX)} = \Op(\varepsilon_{m,n})$ by Lemma~\ref{lemma_nuisance_convergence_rates} and $\|m_1^*(\bX)\|_{\bbL_2(\bbP_\bX)} < \infty$ by Assumption~\ref{assumption_nuisance}, we obtain $\bbE_\bX\{\m^2(\bX)\mid \m\} = \Op(\varepsilon_{m,n}) + O(1)$. Thus, applying Lemma~S4.6 of \cite{sert2025} to upgrade from the conditional bound to an unconditional one, we obtain
\(
    c_S^2 = \Op(n_S^{-1})
\), as $\varepsilon_{m,n} \to 0$.
Since $\lambda_n = n_S/(n_S-2) \to 1$ and $\bbK < \infty$ is fixed, we conclude that:
\begin{equation}
       \frac{n_S}{n_S-2}\,c_S^2 ~=~ \Op(n^{-1}).
    \label{c_S2_rate_of_variance}
\end{equation}
For the treated component, we write $c_1^2 = \widehat{\sigma}_1^2 / n_1$, where $\widehat{\sigma}_1^2 = \sum_{i \in \calI_1} \{ W(\bZ_i, \m,\r) - \eta_1 \}^2/ (n_1-1)$ and $n_1 = n_S \ptilde_n$ with $\ptilde_n \cvP p_1 \in (0,1)$ under Assumption~\ref{assumptions_standard_causal_assump}. Hence $1/n_1 = \Op(n^{-1})$ as $n \to \infty$ (since $\bbK < \infty$ is fixed). Then, for any $t>0$, Markov’s inequality yields:
\begin{align}
    \bbP_S(\widehat{\sigma}_1^2 > t \mid \m, \r)
    & ~\leq~ t^{-1}\,\bbE_S(\widehat{\sigma}_1^2 \mid \m, \r)
      ~=~ t^{-1}\,\Var_{\bZ\mid T=1}\{ W(\bZ,\m,\r) \mid \m,\r\} \nonumber \\
    & ~\leq~ t^{-1}\,\bbE_{(Y,X)\mid T=1}\!\left[\r^2(\bX)\{Y-\m(\bX)\}^2\right],\label{eqn_analysis_of_c12_for_mu1}
\end{align}
where we omit the conditioning argument in the last step for notational simplicity.

By the triangle inequality, we decompose $\bbE_{(Y,X)\mid T=1}\!\left[\r^2(\bX)\{Y-\m(\bX)\}^2\right]  = E_r + E_m$,  where:
\begin{align*}
E_r & ~:=~ \bbE_{S_1}\!\left[ \{\r^2(\bX) - r_1^{*2}(\bX)\}\{Y-m_1^*(\bX)\}^2 \right], ~~\mbox{and}\\
E_m & ~:=~ \bbE_{S_1}\!\left( \r^2(\bX)\left[\{Y-m_1^*(\bX)\}^2 - \{Y-\m(\bX)\}^2\right]\right).
\end{align*}
For $E_r$, by Assumption~\ref{assumption_nuisance}, $\Gamma^2_m = \sup_{\bx}\bbE[\{Y -m_1^*(\bX)\}^2|\bX = \bx] < \infty$, and using Assumption~\ref{assumptions_standard_causal_assump}, we obtain that: by the Cauchy–Schwarz (CS) inequality,
\[E_r ~\leq~ \Gamma^2_m \, \| \r(\bX) - r_1^{*}(\bX) \|_{\bbL_2(\bbP_\bX)} \, \{  \| \r(\bX) - r_1^{*}(\bX) \|_{\bbL_2(\bbP_\bX)} + 2  \|r_1^{*}(\bX) \|_{\bbL_2(\bbP_\bX)}\}.
\]
Since $\|r_1^{*}(\bX) \|_{\bbL_2(\bbP_\bX)} < \infty$, and by Lemma~\ref{lemma_nuisance_convergence_rates}, we have:
\begin{equation}
    E_r ~=~ \Op(\varepsilon_{r, n}). \label{eqn_Er_rate_for_variance_control}
\end{equation}
For the $E_m$ term, we first obtain the following identity
\begin{align*}
    E_m & ~=~ 2\bbE_{S_1}\left(\r^2(\bX)\{\m(\bX)-m_1^*(\bX)\}\{Y-m_1^*(\bX)\}\right)
          - \bbE_{S_1}\left(\r^2(\bX)\{\m(\bX)-m_1^*(\bX)\}^2\right).
\end{align*}
By the CS inequality and Assumptions~\ref{assumptions_standard_causal_assump} and~\ref{assumption_nuisance}, we have:
\begin{align*}
        & \left| \bbE_{S_1}\left(\r^2(\bX)\{\m(\bX)-m_1^*(\bX)\}\{Y-m_1^*(\bX)\}\right)\right|\\
    & ~~~\leq~ \Gamma_m \|\r(\bX)\|_{\bbL_\infty(\bbP_\bX)}^2\|\m(\bX)-m_1^*(\bX)\|_{\bbL_2(\bbP_\bX)},
\end{align*}
and similarly, $\bbE_\bX\left(\r^2(\bX)\{\m(\bX)-m_1^*(\bX)\}^2\right)
    \leq \|\r(\bX)\|_{\bbL_\infty(\bbP_\bX)}^2 \|\m(\bX)-m_1^*(\bX)\|_{\bbL_2(\bbP_\bX)}^2$.
By Assumption~\ref{assumption_nuisance}, together with Lemma~\ref{lemma_nuisance_convergence_rates}, we therefore conclude that:
\begin{equation}
E_m ~=~ \Op(\varepsilon_{m, n}) + \Op(\varepsilon_{m, n}^2) ~=~ \Op(\varepsilon_{m, n}). \label{eqn_Em_rate_for_variance_control}
\end{equation}
Combining \eqref{eqn_Er_rate_for_variance_control} and \eqref{eqn_Em_rate_for_variance_control}, and applying Lemma~S4.6 of \cite{sert2025} to upgrade from conditional to unconditional statements, we obtain $\widehat{\sigma}_1^2 = \Op(\varepsilon_{m,n}) + \Op(\varepsilon_{r,n})
    = \Op \left(\max\{\varepsilon_{m,n},\varepsilon_{r,n}\}\right)$. Since $1/n_1 = \Op(n^{-1})$ by construction, it follows that:
\begin{equation}
    c_1^2 \;=\; \frac{\widehat{\sigma}_1^2}{n_1}
    \;=\; \Op\left(n^{-1}\varepsilon_{m,n}\right) \,+\, \Op\left(n^{-1}\varepsilon_{r,n}\right)
    \;=\; \Op\left(n^{-1}\max\{\varepsilon_{m,n},\,\varepsilon_{r,n}\}\right).\label{eqn_final_analysis_of_c12_for_mu1}
\end{equation}
Together with \eqref{c_S2_rate_of_variance}, which gives $\frac{n_S}{n_S-2}\,c_S^2 = \Op(n^{-1})$, and recalling that $\lambda_n = n_S/(n_S-2) = \Op(1)$ and $\lambda_{n_1} = n_1/(n_1-2) = \Op(1)$, we conclude that:
\begin{equation}
    \hat{c}_n^2(\m,\r) ~=~ O_{\bbP_\calD}(n^{-1}) \;+\; O_{\bbP_\calD}\left(n^{-1}\max\{\varepsilon_{m,n},\,\varepsilon_{r,n}\}\right)
     \;=\; O_{\bbP_\calD}(n^{-1}), \label{eqn_final_bound_for_variance}
\end{equation}
where the last equality uses $\varepsilon_{m,n} \to 0$ and $\varepsilon_{r,n} \to 0$ by Assumption~\ref{assumption_nuisance}.

Finally, recalling the inequality given in \eqref{eqn_V_bound_by_variance}, and setting $\epsilon_n = \varepsilon_{m,n}$ in Case \textbf{C2} and $\epsilon_n = \varepsilon_{r,n}$ in Case \textbf{C3}, and combining with \eqref{eqn_final_bound_for_variance}, we obtain $V(\m,\r,S) \to 0 \text{ in probability under } \bbP_{\calD}$. This completes the proof of the posterior contraction at rate $\epsilon_n$, with $\epsilon_n=\varepsilon_{m,n}$ in Case \textbf{C2} and $\epsilon_n=\varepsilon_{r,n}$ in Case \textbf{C3}; that is, whenever only one of the two nuisance models is correctly specified.
\end{proof}

\phantomsection
\addcontentsline{toc}{subsection}{Proof of Theorem~\ref{thm_BvM_ATE}}

\begin{proof}[Proof of Theorem~\ref{thm_BvM_ATE}.]
    The argument proceeds similarly to the proof of Theorem~\ref{thm_BvM_mu1} with the necessary modifications. We first establish the proof of the BvM result in Theorem~\ref{thm_BvM_ATE}, and then provide a unified proof for the posterior contraction statements therein.

\noindent{\bf Proof of the BvM result:} \\
\noindent\underline{Roadmap of the BvM proof:} {\it The overall proof strategy follows the same structure as that of Theorem~\ref{thm_BvM_mu1}, with one additional step where our methodology plays a crucial role. In the extended DRDB framework, the bias involves joint learning of the individual bias components. However, thanks to our targeted modeling approach and efficient use of the data, we can express this joint posterior as a product of the individual bias posteriors (refer to \eqref{eqn_posterior_for_bias}). This key simplification enables us to decompose the TV distance into terms corresponding to each component and to bound them using known results on the TV distance between Normal and $t$-distributions, as shown in \eqref{eqn_analysis_of_R_for_ATE}. The remainder of the proof then proceeds analogously to the argument in Theorem~\ref{thm_BvM_mu1}.}

By recalling the construction of $\rATE^\CF$ with its posterior $\pATE^\CF$ presented in Section~\ref{sec_DRDB_extended_ATE}, and applying Lemma~\ref{lemma_TV_convolution}, it suffices to prove the BvM result for a single test-training split $(\calD_k, \calD_k^\-)$ for any $k \in \{1, \dots, \bbK\}$ with corresponding Normal limiting distribution. Specifically, for a fixed $k$, we aim to show:
\[
d_\TV\Big(\pATE^{(k)}, \calN\big(\rATE^{(k)}(m^\d,r^\d), c^{2,(k)}(m^\d,r^\d)\big)\Big) \, \cvP \, 0, \, \  {\text as } \ n \to \infty, \,
\]
where $\pATE^{(k)}$ denotes the posterior of $\rATE$ obtained by applying the DRDB procedure to $(\calD_k,\calD_k^\-)$, and $\rATE^{(k)}(m^\d,r^\d)$ and $c^{2,(k)}(m^\d,r^\d)$ are the mean and variance, defined analogously to \eqref{eqn_pmean_pvar_for_ATE} in Section~\ref{sec_theory}.

For notational simplicity, we suppress the superscript $(k)$ and define $S :=\calD_k$ and $S^\- :=\calD_k^\-$, adopting the notational conventions of Sections~\ref{sec_methodology} and~\ref{sec_theory}. Under this convention, we want to show:
\[
\mathbb{T} ~:=~ d_\TV\!\Big(\pATE,\ \calN\!\big(\rATE(m^\d,r^\d),\ c^2(m^\d,r^\d)\big)\Big) ~\cvP~ 0.
\]
The key idea of the proof is to replace the $t$-distributions used in constructing the posterior $\pATE$ in \eqref{eqn_posterior_for_mu_for_S} by Normal distributions with matching center and scale. Specifically, set $Q_{\rATE|b}:= \calN(\eta_S, c^2_S)$ and $Q_{b}:= \calN(b_1; \eta_1, c^2_1) *  \calN(b_0;-\eta_0, c^2_0)$ where the parameters are define analogous to Section~\ref{sec_likelihood_and_posterior_calc} based on $S$ with $\bm \sim \Pi_{m}(\cdot; S^\-)$ and $\br:= (\ubr, \r)$ where $\underbar{\it r}_t \sim \Pi_{r_t}(\cdot; S^\-)$ for $t = 0, 1$.

By the {\it conditional independence}
structure (given the nuisance posterior draws $(\bm, \r)$ where $S \ind \bm$ and $S \ind \underbar{\it r}_t$ for $t = 0, 1$) detailed in Section~\ref{sec_DRDB_extended_ATE}, the posteriors for $b_1$ and $b_0$ are independent, and since $b=b_1-b_0$, the Normal distribution $Q_b$ is well-defined as the convolution above. Then, let $Q_{\rATE}$
denote the marginal distribution and is equal to:
\[
Q_\rATE ~=~ \calN(\eta_S+ \eta_1 - \eta_0, c^2_S + c^2_1 + c_0^2) ~:=~ \calN(\rATE(\bm,\br),\ c^2(\bm,\br)).
\]
By using the triangle inequality,
\begin{equation}
    \mathbb{T} ~\le~ d_\TV\big(\pATE,\ Q_\rATE \big) +\, d_\TV\Big(Q_\rATE,\ \calN\big(\rATE(m^\d,r^\d),\ c^2(m^\d,r^\d)\big)\Big) ~:=~ R + T. \label{eqn_first_bound_for_TV_for_mu}
\end{equation}
{\bf Analysis of $R$.} By the construction of $\pATE$ in Section~\ref{sec_DRDB_extended_ATE} and the independence of $\Pi_{b_1}$ and $\Pi_{b_0}$, a direct application of Lemma~\ref{lemma_TV_bound_theta_as_eta_minus_beta} yields:
\begin{equation}
    R
\;\le\;
d_\TV\big(\Pi_{\rATE\mid b},\, Q_{\rATE\mid b}\big)
\,+\,
d_\TV\big(\Pi_{b_1},\, Q_{b_1}\big)
\,+\,
d_\TV\big(\Pi_{b_0},\, Q_{b_0}\big)
\;=:\;
R_1 + R_2 + R_3. \label{eqn_analysis_of_R_for_ATE}
\end{equation}
Note that $R_1, R_2, R_3$ represent TV distances between $t$ and Normal distributions with matching centers and scales. Applying the known bound for the TV distance between $t$ and Normal distributions with matching centers and scales given in Lemma S4.4 of \cite{sert2025}, we obtain $R_1 \leq C_1/\nu_S, \, R_2 \leq C_2/\nu_1$ and $C_3/\nu_0$ for universal constants $C_1,C_2,C_3<\infty$.

Since $\bbK$ is fixed, $\nu_S = n_S - 1 \to \infty$ as $n \to \infty$, hence $R_1 = O(n^{-1})$. Also, since $\nu_1 = n_1 -1$ with $n_1 = |S_1| = \sum_{i \in \calI_1} T_i = n_S \ptilde_n$ and $\ptilde_n \cvP p_1$, where $p_1 = \bbP(T = 1) > 0$ by Assumption~\ref{assumptions_standard_causal_assump}, we have $R_2 = \Op(n^{-1})$ for fixed $\bbK$. An identical argument, with $p_0 := 1 - p_1 > 0$, gives $R_3 = \Op(n^{-1})$. This implies that $R$ converges to 0 in probability under $\bbP_\calD$.

\noindent{\bf Analysis of $T$.} Since $T$ is the TV distance between two Normal distributions, standard known bounds yield: for some universal constant $C < \infty$,
\begin{equation}
    T ~\le~ C\left| \frac{n_S\{c^{2}(\bm, \br) - c^{2}(m^\d, r^\d)\}}{n_Sc^{2}(m^\d, r^\d)} \right|
        ~\; + ~\; \left| \frac{\sqrt{n_S}\{\rATE(\bm, \br) - \rATE(m^\d, r^\d)\}}{\sqrt{n_Sc^{2}(m_1^\d, r_1^\d)}} \right|.
        \label{upper_bound_for_T_of_mu}
\end{equation}
Since $n_Sc^{2}(m^\d, r^\d) > \delta > 0$ by definition, it suffices to show both $\Omega(\bm,\br): = \sqrt{n_S}\{\rATE(\bm, \br) - \rATE(m^\d, r^\d)\}$ and $V(\bm,\br) := n_S\{c^{2}(\bm, \br) - c^{2}(m^\d, r^\d)\}$ tend to zero in probability under $\bbP_\calD$.

For $\Omega(\bm,\br)$, since $\rATE(\bm, \br)$ is the posterior mean of $\pATE(\cdot; S)$, applying Corollary~\ref{cor_pmean_convergence} to the pair $(S,S^\-)$ with fixed $\bbK < \infty$ and under correct specification of both nuisance models (Case \textbf{C1} in Theorem~\ref{thm_BvM_ATE}) yields
\(\Omega(\bm,\br) = \op(1)
\),
hence equivalently, $\Omega(\bm,\br)\to 0$ in $\bbP_{\calD}$-probability.

Similarly, by using Lemma~\ref{lemma_pvar_convergence} and Corollary~\ref{cor_pvar_convergence_for_ATE} for the single-split $(S,S^\-)$ with fixed $\bbK < \infty$ and under Case \textbf{C1}, we obtain $V(\bm,\br) = \op(1)$, which means $V(\bm,\br)\to 0$ in $\bbP_{\calD}$-probability. Combining these results with the bound obtained in \eqref{eqn_first_bound_for_TV_for_mu}, we conclude that $\mathbb{T} \to 0$ in probability under $\bbP_{\calD}$, completing the proof of the BvM result in Theorem~\ref{thm_BvM_ATE}.

{\bf Proof of posterior contraction results.} Instead of presenting two separate proofs, we provide unified arguments that address both Cases \textbf{C2} and \textbf{C3}, i.e., the settings in which only one nuisance model is correctly specified. The proof leverages the definition of $\rATE^{\CF}$ and the construction of the posterior $\pATE^{\CF}$. Specifically, for any $\epsilon > 0$ and fixed $\bbK < \infty$, we have:
\[
   \pATE^{\CF}\big(|\rATE - \tATE| > \bbK \epsilon \mid \calD\big) ~\leq~ \pATE^{(1)}\big(|\rATE - \tATE| > \epsilon \mid \calD\big)
          + \cdots
          + \pATE^{(\bbK)}\big(|\rATE - \tATE| > \epsilon \mid \calD\big).
\]
Using this decomposition, to establish posterior contraction results in Theorem~\ref{thm_BvM_ATE}, for any $k \in \{1, \dots, \bbK\}$, it suffices to show that for every $M_n \to \infty$, $\pATE^{(k)}(|\rATE - \tATE| > M_n\epsilon_n \mid \calD)$ converges to 0 in $\bbP_\calD$-probability, where $\epsilon_n >0$ is the contraction rate of the correctly specified nuisance model. In particular, $\epsilon_n = \varepsilon_{m,n}$ in Case \textbf{C2}, and $\epsilon_n = \varepsilon_{r,n}$ in Case \textbf{C3}.

For simplicity, we omit the superscript and set the test-training pair as $(S, S^\-):= (\calD_k, \calD_k^{-})$. Let $\widehat{\rATE} = \rATE(\bm, \br)$ be the posterior mean of $\pATE$. The triangle inequality gives:
\begin{align}
 \pATE\big(|\rATE - \tATE| > 2M_n \epsilon_n |  \calD\big) & \le \pATE\big(|\rATE - \widehat{\rATE}| > M_n \epsilon_n | \calD\big) + \pATE\big(|\widehat{\rATE} - \tATE| > M_n \epsilon_n | \calD\big)  \nonumber \\
 &:= V(\bm, \br, S) + \mathbb{D}(\bm, \br, S). \nonumber
\end{align}
Firstly, note that the probability $\mathbb{D}(\bm, \br, S)$ is the indicator of the event $\{|\widehat{\rATE} - \tATE| > M_n \epsilon_n\}$. By applying Corollary~\ref{cor_pmean_convergence} to the split $(S, S^\-)$ with proper parameter adjustments, we conclude that:
\begin{equation}
    \mathbb{D}(\bm, \br, S) \;\to\; 0
    \; \text{ in probability under } \bbP_\calD.    \label{eqn_post_contraction_pmean_part_for_ATE}
\end{equation}
Secondly, since $\widehat{\rATE}$ is the mean of the posterior $\pATE$, applying Chebyshev's inequality yields that $V(\bm,\br,S)
    \equiv \pATE(|\rATE - \widehat{\rATE}| > M_n \epsilon_n | \calD) \leq (M_n \epsilon_n)^{-2} \,\Var(\rATE | \calD)$.
    Hence, since $M_n \to \infty$, it suffices to show $\Var(\rATE | \calD) = \Op(\delta_n^2)$ under both Cases \textbf{C2} and \textbf{C3}, where $\delta_n \to 0$ at a rate faster than or equal to $\epsilon_n$.

By the analysis in the proof of Corollary~\ref{cor_pvar_convergence_for_ATE} and adopting the notation therein, we write the posterior variance $\hat{c}^2(\bm, \br):= \Var_\rATE(\rATE \mid \calD)$ as:
\begin{align*}
  \hat{c}^2(\bm, \br) & = \frac{n_S}{(n_S-2)}\frac{\sum_{i \in \calI}\{\bm(\bX) - \eta_m\}^2}{n_S(n_S-1)} + \frac{n_1}{(n_1-2)}\frac{\sum_{i \in \calI_1}\{W(\bZ_i,\m,\r) - \eta_1\}^2}{n_1(n_1-1)} \\
  & ~~~ + \frac{n_0}{(n_0 - 2)}\frac{1}{n_0(n_0-1)} \!\sum_{i \in \calI_0}\{U(\bZ_i, \ubm,\ubr) - \eta_0\}^2 := \lambda_n c_S^2 + \lambda_{n_1} c_1^2 + \lambda_{n_0} c_0^2,
\end{align*}
where $\bm(\cdot) = \m(\cdot) - \ubm(\cdot)$, $W(\bZ, \m, \r) := \r(\bX)\{Y - \m(\bX)\}$ and $U(\bZ, \ubm, \ubr) := \ubr(\bX)\{Y - \ubm(\bX)\}$ and $\lambda_n = n_S/(n_S-2)$ and $\lambda_t = n_t/(n_t-2)$ for $t = 0,1$ and $\bZ = (Y, \bX)$.

By Assumption~\ref{assumptions_standard_causal_assump} and its definition, as $\bbK<\infty$ fixed, it is clear that $\lambda_n = O(1)$ and $\lambda_t = \Op(1)$ for $t \in \{0,1\}$. For detailed arguments, we refer to the proof of Theorem~\ref{thm_BvM_mu1}. Thus, it is enough to show $c^2_t = \Op(\delta_n)$ for $t \in \{S, 0,1\}$.

For $c^2_S$\,, following the conditional probability argument, Markov's inequality yields:
\begin{align*}
    \bbP(c_S^2 > t \mid \bm)
    & ~\leq~ t^{-1}\,\bbE(c_S^2 \mid \bm)
      = n_S^{-1}\Var_\bX\{\bm(\bX)\mid \bm\}
      \leq n_S^{-1}\,\bbE_\bX\{\bm^2(\bX)\mid \bm\} \\
    & ~\leq~ n_S^{-1} 2\left\{\,\|\bm(\bX) - m^*(\bX)\|_{\bbL_2(\bbP_\bX)}^2 + \|m^*(\bX)\|_{\bbL_2(\bbP_\bX)}^2 \right\},
\end{align*}
where the equality in the first line follows from the construction of $c_S^2$ and $m^*$ is the limiting function at which the nuisance posterior $\Pi_{m}$ contracts.

Since $\|\bm(\bX) - m^*(\bX)\|_{\bbL_2(\bbP_\bX)} = \Op(\varepsilon_{m,n})$ by Lemma~\ref{lemma_nuisance_convergence_rates} and $\|m^*(\bX)\|_{\bbL_2(\bbP_\bX)} < \infty$ by Assumption~\ref{assumption_nuisance}, we apply Lemma~S4.6 of \cite{sert2025} allow us to go from the conditional statement to an unconditional one. As $\bbK < \infty$ is fixed and by definition of $n_S$, this yields:
\begin{equation}
       c_S^2 = \Op(n_S^{-1}) \ \implies \   \frac{n_S}{n_S-2}\,c_S^2 ~=~ \Op(n^{-1}).
    \label{c_S2_rate_of_variance_for_ATE}
\end{equation}
The analysis of $c_1^2$ follows from the same steps as in the proof of Theorem~\ref{thm_BvM_mu1}, specifically through Steps~\eqref{eqn_analysis_of_c12_for_mu1}--\eqref{eqn_final_analysis_of_c12_for_mu1}. To avoid repetition, we omit details here and refer to the corresponding parts of the proof of Theorem~\ref{thm_BvM_mu1}. Therefore, we conclude that:
\begin{equation}
    \frac{n_1}{n_1-2}\,c_1^2 ~=~ \Op\left(n^{-1}\max\{\varepsilon_{m,n},\,\varepsilon_{r,n}\}\right).
    \label{c_12_rate_of_variance_for_ATE}
\end{equation}
Using analogous arguments from Steps~\eqref{eqn_analysis_of_c12_for_mu1}--\eqref{eqn_final_analysis_of_c12_for_mu1}, now applied to the control group $S_0$ with $(\ubm, \ubr)$ (by symmetry between the constructions of $c_0^2$ and $c_1^2$), we also obtain:
\begin{equation}
    \frac{n_0}{n_0-2}\,c_0^2 ~=~ \Op \left( n^{-1}\max\{\varepsilon_{m,n},\,\varepsilon_{r,n}\} \right ) .
    \label{c_02_rate_of_variance_for_ATE}
\end{equation}
Combining \eqref{c_S2_rate_of_variance_for_ATE}--\eqref{c_02_rate_of_variance_for_ATE}, since both $\varepsilon_{m,n}$ and $\varepsilon_{r,n}$ tends to zero as $n\to \infty$, we have:
\begin{equation}
\hat{c}^2(\bm,\br) ~=~ O_{\bbP_\calD}(n^{-1}) \;+\; O_{\bbP_\calD}\!\big(n^{-1}\max\{\varepsilon_{m,n},\,\varepsilon_{r,n}\}\big)
     ~=~ O_{\bbP_\calD}(n^{-1}). \label{eqn_final_bound_for_variance_for_ATE}
\end{equation}
Combining \eqref{eqn_post_contraction_pmean_part_for_ATE} and \eqref{eqn_final_bound_for_variance_for_ATE}, we obtain that for either $\epsilon_n = \varepsilon_{m,n}$ or $\epsilon_n = \varepsilon_{r,n}$, the posterior $\pATE$ contracts around $\tATE$ at rate $\epsilon_n$, where $\epsilon_n$ is the contraction rate of the well-specified model.
\end{proof}

\phantomsection
\addcontentsline{toc}{subsection}{Proof of Corollary~\ref{cor_pmean_convergence}}

\begin{proof}[Proof of Corollary~\ref{cor_pmean_convergence}.]
By construction of the final DRDB posterior $\pATE^\CF$ in \eqref{eqn_CF_version_mu1} and linearity of expectation, the posterior mean, denoted by $\muhat_1(\m, \r)$, is given by:
\begin{equation*}
\muhat_1(\m, \r)
~=~ \frac{1}{\bbK} \sum_{k = 1}^\bbK \muhat_1^{(k)}(\m^{(k)}, \r^{(k)}),
\end{equation*}
where $\muhat_1^{(k)}(\m^{(k)}, \r^{(k)})$ is the posterior mean of $\pATE^{(k)}$, with $\m^{(k)}$ and $\r^{(k)}$ denoting nuisance posterior draws from $\Pi_{m_1}^{(k)}$ and $\Pi_{r_1}^{(k)}$, respectively, for $k = 1,\dots,\bbK$.

Since $\bbK < \infty$ is fixed, it suffices to establish the results of Corollary~\ref{cor_pmean_convergence} for each $\muhat_1^{(k)}$ based on a split $(\calD_k, \calD_k^\-)$, $k = 1,\dots,\bbK$. Since the arguments are identical across $k$ by the construction of the DRDB procedure, for brevity, we present the proof for one posterior mean $\muhat_1^{(k)}$.

Fix $k \in \{1,\dots,\bbK\}$ and set the test-training pair $(S, S^\-) = (\calD_k, \calD_k^\-)$ with index sets $(\calI, \calI^\-)$ and $n_S := |\calI|$ as in Section~\ref{sec_methodology}. For notational simplicity, we drop the superscript $(k)$ and denote the posterior mean $\muhat_1 := \muhat_1^{(k)}(\m^{(k)}, \r^{(k)})$.

Let $S_1$ denote the treated subdata of the test data $S$ with sample size $n_1$. From the construction of the DRDB posterior $\Pi_{\mu_1}$, the posterior mean $\muhat_1 \equiv \muhat_1(\m, \r)$ can then be computed explicitly as:
\begin{align*}
\muhat_1 ~=~ \eta_{m_1} + \eta_1 ~:=~ \frac{1}{n_S} \sum_{i = 1}^{n_S} \m(\bX_i) \, + \, \frac{1}{n_1} \sum_{i \in \calI_1}\r(\bX_i)\{Y_i - \m(\bX_i)\},
\end{align*}
where $\m \sim \Pi_{m_1}(\cdot; S^\-)$ and $\r \sim \Pi_{r_1}(\cdot; S^\-)$. Then, the mean $\mu_1(m_1^*, r_1^*)$ of the limiting Gaussian distribution in Theorem~\ref{thm_BvM_mu1} can be expressed as
\begin{align*}
    \mu_1^* ~\equiv~ \mu_1(m_1^*, r_1^*) & ~=~ \frac{1}{n_S} \sum_{i=1}^{n_S} m_1^*(\bX_i) + \frac{1}{n_S} \sum_{i=1}^{n_S} \frac{r_1^*(\bX_i)}{p_1} T_i \{Y_i - m_1^*(\bX_i)\} \\
    & ~=~ \bbP_{n_S}\big\{m_1^*(\bX)\big\} + \frac{1}{p_1}\, \bbP_{n_S}\big[r_1^*(\bX) T \{Y - m_1^*(\bX)\}\big],
\end{align*}
where $m_1^*$ and $r_1^*$ are the deterministic limiting functions to which the nuisance posteriors $\Pi_{m_1}$ and $\Pi_{r_1}$ contract. When both nuisance models are correctly specified (Case \textbf{C1}), $m_1^* = m_1$ and $r_1^* = r_1$, corresponding to their true nuisance functions. Similarly, $\muhat_1$ can be rewritten as:
\begin{align*}
      \muhat_1 & ~=~ \frac{1}{n_S} \sum_{i = 1}^{n_S} \m(\bX) + \frac{1}{\left(n_S^{-1}\sum_{i = 1}^{n_S} T_i\right)}\frac{1}{n_S} \sum_{i = 1}^{n_S} \r(\bX_i)T_i\{Y_i - \m(\bX_i)\} \\
      & ~=~ \bbP_{n_S}(\m(\bX)) + \frac{1}{\ptilde_n}\bbP_{n_S}(\r(\bX)T\{Y - \m(\bX)\}) ~~~~{\left[\mbox{with}~\;{\ptilde_n} ~:=~ \frac{1}{n_S}\sum_{i = 1}^{n_S
      } T_i\right].}
\end{align*}
Then, the posterior mean $\muhat_1$ admits the following generalized decomposition:
\begin{align}
   \muhat_1 -  \mu_1^* ~=~ & ~ \bbP_{n_S\!}\left(\!\{\m(\bX) - m_1^*(\bX)\}\left\{1-\frac{Tr_1^\d(\bX_i)}{p_1}\right\}\!\right) \nonumber \\
   & ~+~ \bbP_{n_S}\!\left(\!\{Y - m_1^\d(\bX)\}T\left\{ \frac{\r(\bX)}{\ptilde_n} - \frac{r_1^*(\bX)}{p_1}\right\}\!\right) \nonumber \\
    & ~+~ \bbP_{n_S}\left(\{\m(\bX) - m_1^*(\bX)\}T\left\{ \frac{\r(\bX)}{\ptilde_n} - \frac{r_1^*(\bX)}{p_1}\right\}\right) \nonumber \\
    & ~+~ \bbP_{n_S}\left(\{m_1^*(\bX) - m_1^\d(\bX)\}T\left\{ \frac{\r(\bX)}{\ptilde_n} - \frac{r_1^*(\bX)}{p_1}\right\}\right) \nonumber \\
    & ~+~ \bbP_{n_S}\left(\{\m(\bX) - m_1^*(\bX)\}T\left\{ \frac{r_1^\d(\bX)}{p_1} - \frac{r_1^*(\bX)}{p_1}\right\}\right)  \nonumber\\
    & ~:=~  R_1 + R_2 + R_3 + R_4 + R_5. \label{eqn_extended_decomposition_of_pmean}
\end{align}
Notice that under the correct specification of nuisance models (Case \textbf{C1}), we have $R_4 = 0$ and $R_5 = 0$. By definition, $R_4 = 0$ also holds in Case \textbf{C2} ($m_1^* = m_1^\d$), while $R_5 = 0$ holds in Case \textbf{C3} ($r_1^* = r_1^\d$). Moreover, in all three cases (\textbf{C1}--\textbf{C3}), $R^* := \mu_1^* - \mu^\d(1)$ can be viewed as the sample average of centered i.i.d. random variables, implying that $R^* = \Op(n^{-1/2})$ (as $\bbK < \infty$ is fixed).

\textbf{Case C1:} $m_1^* = m_1^\d$ and $r_1^* = r_1^\d$. Then, $R_4 = 0$ and $R_5 = 0$. To prove Corollary \ref{cor_pmean_convergence} (a)(i), it suffices--by the dominated convergence theorem (DCT), or equivalently Lemma~S4.6 of \citet{sert2025}, and using the independence condition $S \!\ind\! S^\-$ (in particular, $S \!\ind\! (\m, \r)$) -- to verify that the conditional probabilities of $R_1, R_2$, and $R_3$ converge to zero in probability \footnote{{\it The same proof techniques with conditioning arguments, along with the independence between nuisance posterior draws and the test data $S$, is applied repeatedly in subsequent proofs; for brevity, we will not restate them each time.}}.

\noindent{\bf Analysis of $R_1$}. We first note that $\bbE(R_1|\m) = 0$ by the definition of $r_1(\cdot)$. Then, for any $t > 0$, Chebyshev's inequality gives:
\begin{align*}
\bbP_S(|R_1| > t | \m) & ~\leq ~ t^{-2} n_S^{-1} \Var_S(\{\m(\bX) - m_1^\d(\bX)\}\{1-Tr_1^\d(\bX)/p_1\}|\m) \\
& ~=~  t^{-2} n_S^{-1} \bbE_S(\{\m(\bX) - m_1^\d(\bX)\}^2\{1-Tr_1^\d(\bX)/p_1\}^2|\m) \\
& ~\leq~ M_r \, t^{-2} n_S^{-1} \|\m(\bX) - m_1^\d(\bX)\|_{\bbL_2(\bbP_\bX)}^2,
\end{align*}
where $M_r < \infty$ follows from the definition of $r_1^\d(\bX)$ and $p_1$ and by Assumption~\ref{assumptions_standard_causal_assump}. Next, by Lemma~S4.6 of \cite{sert2025} and Lemma~\ref{lemma_nuisance_convergence_rates}, as $\bbK < \infty$ is fixed, we conclude that:
\begin{align}
    R_1 ~=~ \Op\left(n_S^{-1/2}\varepsilon_{m,n}\right). \label{R_1_bigO_rate}
\end{align}
Also, as $\bbK$ is fixed, by Assumption~\ref{assumption_nuisance} and the DCT, $\sqrt{n} R_1 = o_{\bbP}(1)$, yielding the result.

\noindent{\bf Analysis of $R_2$}. Note that $R_2$ is {\it not} the average of independent random variables, since each summand depends on the common factor
\( \ptilde_n
\). Yet, the summands are identically distributed, which will simplify subsequent calculations.

We first establish that $\bbE(R_2 | \m) = 0$ to use Chebyshev's inequality. Since $\r \ind S$, we, for notational simplicity, suppress explicit conditioning on $\r$ in the expectation. Then, we can write:
\begin{align}
\bbE_S(R_2|\r) \, \equiv \, \bbE_S(R_2)
& ~=~ \bbE_S\left[\bbP_{n_S}\left(\{Y - m_1^\d(\bX)\}\frac{T\r(\bX)}{\ptilde_n}\right)\right] \nonumber \\
& ~~~~~ -~ \bbE_S\left(\{Y - m_1^\d(\bX)\} \frac{Tr_1^*(\bX)}{p_1}\right) \nonumber \\
& ~=~ \bbE_S\left[\bbP_{n_S}\left(\{Y - m_1^\d(\bX)\} \frac{T\r(\bX)}{\ptilde_n}\right)\right], \label{eqn_R2_expansion}
\end{align}
where the last step follows from the NUC condition in Assumption~\ref{assumptions_standard_causal_assump}. Thus, it remains to show that the first term in \eqref{eqn_R2_expansion} is zero, which will establish $\bbE_S(R_2|\r) \equiv \bbE_S[R_2] = 0$.

Towards that, define $\calT_n:= (T_1, \dots, T_{n_S})$ and $\mathcal{X}_n:= (\bX_1, \dots, \bX_{n_S})$. Then, we have:
\begin{align*}
\bbE_S & \left[\bbP_{n_S}\left(\{Y - m_1^\d(\bX)\}T \frac{\r(\bX)}{\ptilde_n}\right)\right] \\
& \, = \, n_S^{-1} \bbE\left(\bbE\left[\sum_{i = 1}^{n_S} \{Y_i(1) - m_1^\d(\bX_i)\}T_i \frac{\r(\bX_i)}{\ptilde_n}\mid \calT_n, \mathcal{X}_n\right]\right)\\
%& = n^{-1} \bbE\left(\frac{1}{\ptilde_n}\sum_{i = 1}^n \bbE\left[\{Y_i(1) - m_1^\d(\bX_i)\}T_i\r(\bX_i)\mid X_1, \dots , X_n, T_1, \dots , T_n\right]\right) \\
& \, = \, n_S^{-1} \bbE\left(\frac{1}{\ptilde_n}\sum_{i = 1}^{n_S} \bbE\left[\{Y_i(1) - m_1^\d(\bX_i)\}T_i \r(\bX_i)\mid \bX_i, T_i\right]\right) \\
%& = \bbE\left(\frac{1}{\ptilde_n}\bbE\left[\{Y_1(1) - m_1^\d(\bX_1)\}T_1 \r(\bX_1)\mid X_1, T_1\right]\right) \\
%& = \bbE\left(\frac{T_1\r(\bX_1)}{\ptilde_n}\bbE\left[\{Y_1(1) - m_1^\d(\bX_1)\}\mid X_1\right]\right) \\
& \, = \, \bbE\left(\frac{\r(\bX_1)}{\ptilde_n}\bbE\left[\{Y_1(1) - m_1^\d(\bX_1)\}\mid \bX_1, T_1 = 1\right]\right) \, = \, 0,
\end{align*}
since the inner expectation in the last step is zero by the NUC condition in Assumption~\ref{assumptions_standard_causal_assump}. This proves that $\bbE(R_2 \, | \, \r) = 0$, as desired.

Next, define $W(\bZ, \r) := \r(\bX)T\{Y - m_1^\d(\bX)\}/\ptilde_n$. Since $\bbE_S(W(\bZ, \r) |\r) = 0$, by Chebyshev's inequality, we have: for any $t>0$,
\begin{align*}
    \bbP(|R_2| > t | \r) & ~\leq~ \frac{\Var_S(R_2|\r)}{t^2} \, \equiv \, \frac{\Var_S(R_2)}{t^2} \, = \, \frac{\bbE_S(R_2^2)}{t^2} \\
    & \, = \, \frac{1}{t^2 \, n_S^2}\bbE_S\left[\left\{\sum_{i = 1}^{n_S} W(\bZ_i, \r)\right\}^2\right]\\
    & ~=~ \frac{\bbE_S\{W^2(\bZ_1, \r)\}}{t^2\, n_S} \, + \, \frac{(n_S-1)}{t^2 \, n_S} \bbE_S\{W(\bZ_1, \r)W(\bZ_2, \r)\}.
\end{align*}

To complete the analysis of $R_2$,
it remains to compute the quaantity: $\bbE_S\{W(\bZ_1, \r) W(\bZ_2, \r)\} \equiv \bbE_S\{W(\bZ_1, \r) W(\bZ_2, \r)| \r\}$. For notational simplicity, we omit conditioning on $\r$, as $S \ind \r$. Define $R_{12} := \bbE_S\{W(\bZ_1, \r) W(\bZ_2, \r)\}$. Then,
\begin{align*}
   \hspace{-3ex} R_{12}& = \bbE_S\!\left[ \{Y_1 - m_1^\d(\bX_1)\}T_1\left\{ \frac{\r(\bX_1)}{\ptilde_n} - \frac{r^*(\bX_1)}{p_1}\right\} \{Y_2 - m_1^\d(\bX_2)\}T_2\left\{ \frac{\r(\bX_2)}{\ptilde_n} - \frac{r^*(\bX_2)}{p_1}\right\}\right]\\
   & ~=~  \bbE_S(A_1T_1B_1A_2T_2B_2) = \bbE\left[\bbE(A_1T_1B_1A_2T_2B_2 \mid \bX_1, \bX_2, T_1, T_2, \ptilde_n)\right] \\
   & ~=~ \bbE\left[B_1T_1T_2B_2~\bbE_S(A_1A_2 \mid \bX_1, \bX_2, T_1, T_2, \ptilde_n)\right],
\end{align*}
where for $i = 1,2$, $A_i = \{Y_i - m_1^\d(\bX_i)\}$ and $B_i =\{ \r(\bX_i)/\ptilde_n - r^*(\bX_i)/p_1\}$.
Notice that for $i = 1,2$, $B_i$ is random through $\bX_i$ and $\ptilde_n$; conditioning on these quantities makes them deterministic. Thus, the only remaining randomness in the conditional expectation arises from $A_1$ and $A_2$, which depend on $Y_1$ and $Y_2$, respectively. Since $(Y_1, \bX_1, T_1)$ and $(Y_2, \bX_2, T_2)$ are independent, we can write:
\begin{align*}
    \bbE\left(A_1 A_2 \mid \bX_1, \bX_2, T_1, T_2, \ptilde_n\right)
    ~=~ \bbE\left(A_1 \mid \bX_1, \bX_2, T_1, T_2, \ptilde_n\right) \,
      \bbE\left(A_2 \mid \bX_1, \bX_2, T_1, T_2, \ptilde_n\right).
\end{align*}

Since $A_i$ depends only on $Y_i$ once $\bX_i$ and $T_i$ are given, conditioning on extra variables does not alter its conditional expectation. In particular, for $i=1$, by using Assumption~\ref{assumptions_standard_causal_assump},
\begin{align*}
\bbE\left(A_1| \bX_1, \bX_2, T_1, T_2, \ptilde_n\right)
    = \bbE\left(A_1 | \bX_1, T_1\right)
    = \bbE\{Y_1 - m_1^\d(\bX_1) | \bX_1, T_1 = 1\} = 0.
\end{align*}
Hence, we obtain that $\bbE_S\{W(\bZ_1, \r) W(\bZ_2, \r) \mid \r\} = 0$.

Returning to the analysis of $R_2$, using the derivation above, we now have:
\begin{align*}
    \bbP(|R_2| > t \mid \r) & \, \leq~ \frac{\bbE_S\big[W^2(\bZ_1, \r) \mid \r\big]}{t^2 n_S} \\
    & ~=~ \frac{1}{t^2 n_S} \bbE_S\!\left[ \{Y - m_1^\d(\bX)\}^2T^2\left\{ \frac{\r(\bX)}{\ptilde_n} - \frac{r^*(\bX)}{p_1}\right\}^2 \right]\\
    %& ~ \leq \frac{M \Gamma_{m}^*}{t^2 n} \left\|p_1\r(\bX) - \ptilde_n r_1^\d(\bX)\right\|_2^2 =
    & \, = ~ \frac{M\, \Gamma_{m}^*}{t^2 n_S}\, O_{\bbP}\left(\max\left\{\|p_1 - \ptilde_n \|_2^2, \|\r(\bX) - r_1^\d(\bX)\|_{\bbL_2(\bbP_\bX)}^2\right\}\right),
\end{align*}
where $M< \infty$ due to the positivity assumption in Assumption~\ref{assumptions_standard_causal_assump} and $\Gamma_{m}^* = \sup_{\bx} \bbE[\{Y - m_1^\d(\bX)\} | \bX = \bx] < \infty$ by Assumption~\ref{assumption_nuisance}. Since $\bbE_S(\ptilde_n) = p_1$, we have $\|p_1 - \ptilde_n \|_2^2 = \Var_S(\ptilde_n) = p_1(1-p_1)/n_S$. Hence, by Lemma~S4.6 of \cite{sert2025} and Lemma~\ref{lemma_nuisance_convergence_rates}, we conclude that:
\begin{align}
    R_2 ~=~ \Op\left(n_S^{-1/2}\varepsilon_{r,n}\right). \label{R_2_bigO_rate}
\end{align}
Lastly, since $\bbK$ is fixed, by Assumption~\ref{assumption_nuisance} and the DCT, we obtain $\sqrt{n} R_2 = o_{\bbP}(1)$.

\textbf{Analysis of $R_3$.} Let $U(\bZ, \m, \r):= \{m_1^\d(\bX) - \m(\bX)\} T\{ \r(\bX)/\ptilde_n -r_1^\d(\bX)/p_1 \}$. Then, we rewrite $R_3 = \bbP_{n_S}\{U(\bZ, \m, \r)\}$. Although, given $\m, \r$, $U(\bZ_1, \m, \r),\ldots, U(\bZ_{n_S}, \m, \r)$ are not independent due to the shared random factor $\ptilde_n$, they are identically distributed. Notably, $U(\bZ, \m, \r)$ is a \emph{product-type} random variable with $\bbE\{U(\bZ, \m, \r) | \m, \r\} \neq 0$. Thus, the preceding zero-mean arguments do not directly apply here. Omitting conditioning on $(\m,\r) (\ind S)$ for brevity, for any $t>0$, we have:
\begin{align*}
    \bbP\left(|R_3| > t | \m, \r \right)
    & ~\leq~ t^{-1} \bbE_S\left[\left| \frac{1}{n} \sum_{i = 1}^n \{m_1^\d(\bX_i) - \m(\bX_i)\} T_i \left\{ \frac{\r(\bX_i)}{\ptilde_n} - \frac{r_1^\d(\bX_i)}{p_1} \right\} \right| \right] \\
    & ~\leq~ t^{-1} \bbE_S\left[ \left| \{m_1^\d(\bX) - \m(\bX)\} T \left\{ \frac{\r(\bX)}{\ptilde_n} - \frac{r_1^\d(\bX)}{p_1} \right\} \right| \right] \\
    & ~\leq~ M \| m_1^\d(\bX) - \m(\bX) \|_2 \| p_1 \r(\bX) - \ptilde_n r_1^\d(\bX) \|_2,
\end{align*}
where the last bound follows from the CS inequality and the positivity condition in Assumption~\ref{assumptions_standard_causal_assump}. Following similar arguments to those used in the analysis of $R_2$, and applying Lemma~S4.6 of \cite{sert2025} together with Lemma~\ref{lemma_nuisance_convergence_rates}, we obtain:
\begin{align}
    R_3 ~=~ O_{\bbP}\left(\varepsilon_{m,n}\varepsilon_{r,n}\right). \label{R_3_bigO_rate}
\end{align}

As $\bbK < \infty$ is fixed, by Assumption~\ref{assumption_nuisance} and the DCT, $\sqrt{n} R_3 = o_{\bbP}(1)$, establishing the desired result. Combining the bounds derived for $R^*$, $R_1$, $R_2$, and $R_3$ through \eqref{R_1_bigO_rate}--\eqref{R_3_bigO_rate}, and recalling that $R_4 = R_5 = 0$ under Case \textbf{C1}, we obtain the conclusion of Corollary~\ref{cor_pmean_convergence} (a)(i) for this setting.

\noindent{\bf Case C2:} $m_1^* = m_1^\d$ but $r_1^* \neq r_1^\d$.
By the generalized decomposition of $\muhat_1$ in \eqref{eqn_extended_decomposition_of_pmean},
\begin{align}
    \muhat_1 - \mu^\d(1) ~=~ R^* + R_1 + R_2 + R_3 + R_5, \label{eqn_decomposition_of_pmean_for_C2}
\end{align}
as $R_4 = 0$ when $m_1^* = m_1$. As shown in \eqref{R_1_bigO_rate}, \eqref{R_2_bigO_rate} and \eqref{R_3_bigO_rate}, we have $R_1 = \Op(n^{-1/2}\varepsilon_{m,n})$, $R_2 = \Op(n^{-1/2}\varepsilon_{r,n})$, $R_3 = \Op(\varepsilon_{m,n}\varepsilon_{r,n})$, and, by construction, $R^* = \Op(n^{-1/2})$. Hence, the analysis for Case \textbf{C2} reduces to controlling the behavior of the remaining remainder term $R_5$.

For any $t>0$, using the same conditioning argument as above, suppressing it here for brevity, Markov’s inequality gives:
\begin{align*}
    \bbP_S\big(|R_5| > t \,\big|\, \m \big)
    & ~\le~ t^{-1}\,\bbE_S\!\left( \big|\m(\bX) - m_1^\d(\bX)\big|\,T
          \left| \frac{r_1^\d(\bX)}{p_1} - \frac{r_1^*(\bX)}{p_1} \right| \right) \\
    & ~=~ t^{-1}\,\bbE_S\!\left( \big|\m(\bX) - m_1^\d(\bX)\big|\left| 1 - \frac{r_1^*(\bX)}{r_1^\d(\bX)} \right| \right) \\
    & ~\le~ t^{-1}\,\|\m(\bX) - m_1^\d(\bX)\|_{\mathbb{L}_2(\bbP_\bX)}
          \,\big\| 1 - r_1^*(\bX)/r_1^\d(\bX) \big\|_{\mathbb{L}_2(\bbP_\bX)},
\end{align*}
where the second step uses $r_1^\d(\bX) > 0$ and $p_1 > 0$ and the last step follows from the CS inequality.

By Assumption~\ref{assumption_nuisance}, $\|1 - r_1^*(\bX)/r_1^\d(\bX)\|_{\mathbb{L}_2(\bbP_\bX)} < \infty$.
Applying Lemma~S4.6 of \cite{sert2025} together with Lemma~\ref{lemma_nuisance_convergence_rates}, we obtain:
\begin{align}
    R_5 ~=~ \Op(\varepsilon_{m,n}).
    \label{R_5_bigO_rate}
\end{align}

Now, take any sequence $M_n \to \infty$ and set $\epsilon_n = \varepsilon_{m,n}$, the contraction rate of the well-specified regression model in Case~\textbf{C2}. From the decomposition in~\eqref{eqn_decomposition_of_pmean_for_C2}, we have:
\begin{align*}
    \bbP\!\left(|\muhat_1 - \mu^\d(1)| > M_n \epsilon_n\right)
    & ~\le~ \bbP(|R^*| > M_n \epsilon_n)
        \, + \, \bbP(|R_1| > M_n \epsilon_n)
        \, + \, \bbP(|R_2| > M_n \epsilon_n) \\
    &\quad + \, \bbP(|R_3| > M_n \epsilon_n)
        + \bbP(|R_5| > M_n \epsilon_n).
\end{align*}
Combining Steps~\eqref{R_1_bigO_rate}--\eqref{R_5_bigO_rate}, since $M_n \to \infty$, we obtain $\bbP\!\left(|\muhat_1 - \mu^\d(1)| > M_n \epsilon_n\right) \to 0$, which implies $|\muhat_1 - \mu^\d(1)| = \Op(\varepsilon_{m,n})$. Thus, under Case~\textbf{C2}, the posterior mean $\muhat_1$ is $\varepsilon_{m,n}^{-1}$-consistent estimator for $\mu^\d(1)$.

\noindent\textbf{Case C3:} $r_1^* = r_1^\d$ but $m_1^* \neq m_1^\d$.
Following analogous arguments to those in Case~\textbf{C2}, we first note that $R_5 = 0$ by construction and by the decomposition given in \eqref{eqn_extended_decomposition_of_pmean}, we have:
\begin{align}
    \muhat_1 - \mu^\d(1)
    ~=~ R^* + R_1 + R_2 + R_3 + R_4. \label{eqn_decomposition_of_pmean_for_C3}
\end{align}
As established in \eqref{R_1_bigO_rate}, \eqref{R_2_bigO_rate} and \eqref{R_3_bigO_rate}, we have $R_1 = \Op(n^{-1/2}\,\varepsilon_{m,n})$, $R_2 = \Op(n^{-1/2}\,\varepsilon_{r,n})$,
\( R_3 = \Op(\varepsilon_{m,n}\,\varepsilon_{r,n}),
\)
and, $R^* = \Op(n^{-1/2})$ by construction.
Hence, the analysis of Case \textbf{C3} reduces to controlling the behavior of the last term $R_4$.

We apply the same conditioning argument as before and follow the analysis of $R_2$ and $R_3$. For any $t>0$, Markov’s inequality gives:
\begin{align*}
    \bbP_S(|R_4| > t| \r) ~ & ~\leq~ t^{-1} \bbE_S\left(|m_1^*(\bX) - m_1^\d(\bX)|\, T \left| \frac{\r(\bX)}{\ptilde_n} - \frac{r_1^\d(\bX)}{p_1} \right| \right)\\
    ~ & ~\leq~ t^{-1} \|m_1^*(\bX) - m_1^\d(\bX)\|_{\bbL_2(\bbP_\bX)} \| p_1 \r(\bX) - \ptilde_n r_1^\d(\bX) \|_{\bbL_2(\bbP_S)} \\
   ~ & ~\leq~ t^{-1} \|m_1^*(\bX) - m_1^\d(\bX)\|_{\bbL_2(\bbP_\bX)} p_1\|\r(\bX) -  r_1^\d(\bX) \|_{\bbL_2(\bbP_\bX)} \\
   ~&~~~ + t^{-1} \|m_1^*(\bX) - m_1^\d(\bX)\|_{\bbL_2(\bbP_\bX)} M_r\Var(\ptilde_n),
\end{align*}
where the second inequality follows from the CS inequality, and the third from the definition of $r_1(\cdot)$ together with $p_1>0$ and Assumption~\ref{assumptions_standard_causal_assump}. The detailed analysis of $\| p_1 \r(\bX) - \ptilde_n r_1^\d(\bX) \|_{\bbL_2(\bbP_S)}$ parallels that in the analysis of $R_2$ and $R_3$.

By Assumption~\ref{assumption_nuisance}, $\|m_1^*(\bX) - m_1^\d(\bX)\|_{\bbL_2(\bbP_\bX)} < \infty$. Further, we note $\Var(\ptilde_n) = p_1(1-p_1)/n = \Op(n^{-1})$. Then, applying Lemma~S4.6 of \cite{sert2025} together with Lemma~\ref{lemma_nuisance_convergence_rates}, we conclude
\begin{align}
    R_4 ~=~ \Op(\varepsilon_{r,n}).
    \label{R_4_bigO_rate}
\end{align}
Finally, take any sequence $M_n \to \infty$ and set $\epsilon_n = \varepsilon_{r, n}$, nuisance contraction rate in Case \textbf{C3}. From the decomposition in \eqref{eqn_decomposition_of_pmean_for_C3}, $\bbP\!\left(|\muhat_1 - \mu^\d(1)| > M_n \epsilon_n\right) \le \bbP(|R^*| > M_n \epsilon_n)
        + \bbP(|R_1| > M_n \epsilon_n)
        + \bbP(|R_2| > M_n \epsilon_n) + \bbP(|R_3| > M_n \epsilon_n)
        + \bbP(|R_4| > M_n \epsilon_n)$. Combining \eqref{R_1_bigO_rate}, \eqref{R_2_bigO_rate}, \eqref{R_3_bigO_rate}, and \eqref{R_4_bigO_rate}, we have $\bbP(|\muhat_1 - \mu^\d(1)| > M_n \epsilon_n) \to 0$, which implies $|\muhat_1 - \mu^\d(1)| = \Op(\varepsilon_{m,n})$. Thus, under Case \textbf{C3}, the posterior mean $\muhat_1$ is an $\varepsilon_{r,n}^{-1}$‑consistent estimator of $\mu^\d(1)$. This gives a complete proof of Corollary~\ref{cor_pmean_convergence} for $\mu^\d(1)$.

The proof for the second part of Corollary~\ref{cor_pmean_convergence} is omitted for brevity, as it follows from a symmetric argument to the first part. The posterior mean of \(\pATE^\CF\) is $\widehat{\rATE}(\bm, \br) = \muhat_1(\m, \r) - \muhat_0(\ubm, \ubr)$, where \(\muhat_0(\ubm, \ubr)\) is the posterior mean for the control group, defined as $\muhat_0(\ubm, \ubr) := \bbP_{n_S}\{\ubm(\bX)\} + \bbP_{n_S}\{\ubr(\bX)(1-T)\{Y - \ubm(\bX)\}\}/(1-\ptilde_n)$. Here, $\ubm(\cdot)$ and $\ubr(\cdot)$ denote nuisance samples drawn from posteriors $\Pi_{m_0}(\cdot; S^\-)$ and $\Pi_{r_0}(\cdot; S^\-)$. For clarity, $\Pi_{m_0}(\cdot; S^\-)$ and $\Pi_{r_0}(\cdot; S^\-)$ are the nuisance posteriors obtained applying the DRDB procedure to estimate the mean of control arm, $\mu^\d(0) = \bbE[Y(0)]$ similar to described in Section~\ref{sec_DRDB_for_mu1}.

Due to the symmetry in the constructions of \(\muhat_1(\m,\r)\) and \(\muhat_0(\ubm,\ubr)\), all arguments and results established for \(\muhat_1(\m,\r)\) in Steps~\ref{eqn_extended_decomposition_of_pmean}–\ref{R_4_bigO_rate} across Cases \textbf{C1}–\textbf{C3} apply directly to \(\muhat_0(\ubm,\ubr)\). This symmetry completes the proof of the second part of Corollary~\ref{cor_pmean_convergence}, establishing the desired properties of the posterior mean $\widehat{\rATE}(\bm, \br)$ as a point estimator of the ATE $\tATE$.
\end{proof}

\section{Proofs of the preliminary lemmas}\label{supp_sec_proof_of_preliminary}
This section provides proofs of the intermediate results used in the proofs of the main results.

\phantomsection
\addcontentsline{toc}{subsection}{Proof of Lemma~\ref{lemma_TV_distance_btw_two_posteriors}}
\begin{proof}[Proof of Lemma~\ref{lemma_TV_distance_btw_two_posteriors}]
The proof follows from the integral representation of the TV distance, Fubini's theorem, and the triangle inequality. To see this clearly, we observe that:
\begin{align*}
  d_{\TV}(\calP_1, \calP_2) & \, = \, \frac{1}{2}\int |f_1(\rATE) - f_2(\rATE)|\hspace{0.5mm} \dd \rATE =  \frac{1}{2}\int \left|   \int q_1(\rATE | \theta)\hspace{0.5mm}  g_1(\theta) - q_2(\rATE | \theta)\hspace{0.5mm} g_2(\theta) \hspace{0.5mm} \dd \theta \right| \dd \rATE \\
 & \, = \, \frac{1}{2}\int |g_1(\theta) - g_2(\theta)| \dd \theta \\
 & ~~~~+~ \frac{1}{2} \int \int |\psi_1(\rATE - \theta - C) - \psi_2(\rATE - \theta - C)|\hspace{0.5mm} g_2(\theta) \dd\theta \hspace{0.5mm}\dd \rATE \\
& \, = \, d_{\TV}(\calP_\theta^{(1)}, \calP_\theta^{(2)}) + \frac{1}{2} \int |\psi_1(w) - \psi_2(w)| \dd w \\
& \, = \, d_{\TV}(\calP_\theta^{(1)}, \calP_\theta^{(2)}) +  d_{\TV}(\calP_{\rATE \mid \theta}^{(1)}, \calP_{\rATE \mid \theta}^{(2)}),
\end{align*}
where the last step uses the location-family assumption with a common shift, $d_{\mathrm{TV}}(\mathcal{P}_{\rATE\mid\theta}^{(1)}, \mathcal{P}_{\rATE\mid\theta}^{(2)})$ does not depend on $\theta$ and equals $d_{\mathrm{TV}}(\psi_1,\psi_2)$, yielding the stated bound.
\end{proof}

\phantomsection
\addcontentsline{toc}{subsection}{Proof of Lemma~\ref{lemma_TV_convolution}}
\begin{proof}[Proof of Lemma~\ref{lemma_TV_convolution}]
The proof follows from a direct application of the integral representation of the total variation distance and Fubini's theorem.
\end{proof}

\phantomsection
\addcontentsline{toc}{subsection}{Proof of Lemma~\ref{lemma_TV_bound_theta_as_eta_minus_beta}}
\begin{proof}[Proof of Lemma~\ref{lemma_TV_bound_theta_as_eta_minus_beta}.]
By the triangle inequality and Fubini’s theorem,
\begin{align}
2d_{\mathrm{TV}}(\mathcal{P}_1,\mathcal{P}_2)
& \,= \, \int \Big|\int q_1(\rATE|\theta)\,g_1(\theta) - q_2(\rATE|\theta)\,g_2(\theta)\, d\theta\Big|\, d\rATE \nonumber\\
& \, \le \, \int\!\!\int \big|q_1(\rATE\mid\theta)\big|\,\big|g_1(\theta)-g_2(\theta)\big|\, d\theta\, d\rATE \\
& ~~~~ + \int\!\!\int \big|q_1(\rATE\mid\theta)-q_2(\rATE\mid\theta)\big|\, g_2(\theta)\, d\theta\, d\rATE  \nonumber \\
& \, = \, d_{\mathrm{TV}}(\mathcal{P}_\theta^{(1)},\mathcal{P}_\theta^{(2)})
\;+\; \int d_{\mathrm{TV}}(\mathcal{P}_{\rATE\mid\theta}^{(1)},\mathcal{P}_{\rATE\mid\theta}^{(2)})\, g_2(\theta)\, d\theta \nonumber \\
& \, \le \, d_{\mathrm{TV}}(\mathcal{P}_\theta^{(1)},\mathcal{P}_\theta^{(2)})
\;+\; \sup_\theta d_{\mathrm{TV}}(\mathcal{P}_{\rATE\mid\theta}^{(1)},\mathcal{P}_{\rATE\mid\theta}^{(2)}). \label{eqn_TV_marginal_and_conditional_lemma3}
\end{align}

Since $\theta^{(i)} = \eta^{(i)} - \beta^{(i)}$ with $\eta^{(i)} \perp\!\!\!\perp \beta^{(i)}$, we have $\mathcal{P}_\theta^{(i)} = \mathcal{P}_\eta^{(i)} * \mathcal{P}_{-\beta}^{(i)}$. By applying subadditivity of the TV distance under convolution and invariance under location shifts,
\[
d_{\mathrm{TV}}(\mathcal{P}_\theta^{(1)},\mathcal{P}_\theta^{(2)})
\;\le\; %d_{\mathrm{TV}}(\mathcal{P}_\eta^{(1)},\mathcal{P}_\eta^{(2)})
%\;+\; d_{\mathrm{TV}}(\mathcal{P}_{-\beta}^{(1)},\mathcal{P}_{-\beta}^{(2)})
%\;=\;
d_{\mathrm{TV}}(\mathcal{P}_\eta^{(1)},\mathcal{P}_\eta^{(2)})
\;+\; d_{\mathrm{TV}}(\mathcal{P}_\beta^{(1)},\mathcal{P}_\beta^{(2)}).
\]
Under the location-family assumption with common shift (and scale), $d_{\mathrm{TV}}(\mathcal{P}_{\rATE\mid\theta}^{(1)},\mathcal{P}_{\rATE\mid\theta}^{(2)})$ in \eqref{eqn_TV_marginal_and_conditional_lemma3} does not depend on $\theta$. Thus, combining the two bounds above yields the desired result.
\end{proof}

\phantomsection
\addcontentsline{toc}{subsection}{Proof of Lemma~\ref{lemma_nuisance_convergence_rates}}
\begin{proof}[Proof of Lemma~\ref{lemma_nuisance_convergence_rates}.]
We present a unified proof for a generic nuisance function $\psi(\cdot)$ whose posterior $\Pi_{\psi}$ is constructed on the training data $S^\- := \calD\setminus\calD_k$ for some $k \in \{1, \dots, \bbK \}$ as described in Section~\ref{sec_methodology}. The same proof steps directly apply to the cases $\psi(\cdot)=\underbar{\it m}_t(\cdot)\sim \Pi_{m_t}$ and $\psi(\cdot)=\underbar{\it r}_t(\cdot) \sim \Pi_{r_t}$ for $t\in\{0,1\}$.

Recalling that the randomness of $\psi$ comes from both $S^\-$ and the posterior $\Pi_\psi$, we can express the distribution $\bbP_\psi = \bbP_{S^\-} \otimes \Pi_\psi$. Let $\psi^* \equiv \psi^*(\cdot) \in \bbL_2(\bbP_\bX)$ be a deterministic limiting function at which $\Pi_\psi$ contract. To show that the nuisance posterior $\Pi_\psi$ contracts around $\psi^*$ at rate $\epsilon_n$, it suffices to verify that for every $M_n \to \infty$, $\bbP_\psi\!\{\| \psi(\bX) - \psi^*(\bX)\|_{\bbL_2(\bbP_\bX)} \ge M_n \epsilon_n \} \to  0$.

By definition of $\bbP_\psi$ and the iterated expectations, $\mathbf{1}(\cdot)$ denoting the indicator function, we have:
\begin{align*}
    \bbP_\psi\!\left\{\| \psi(\bX) - \psi^*(\bX)\|_{\bbL_2(\bbP_\bX)} \ge M_n \epsilon_n \right\}  & \,=\, \bbE_\psi\left\{\mathbf{1}(\| \psi(\bX) - \psi^*(\bX)\|_{\bbL_2(\bbP_\bX)} \ge M_n \epsilon_n)\right\} \\
    & \,=\, \bbE_{S^\-}[\Pi_\psi\left\{\| \psi(\bX) - \psi^*(\bX)\|_{\bbL_2(\bbP_\bX)} \ge M_n \epsilon_n | S^\-\right\}].
\end{align*}
Define $A_n := \big\{\|\psi(\bX)-\psi^*(\bX)\|_{\bbL_2(\bbP_{\bX})} \ge M_n \epsilon_n\big\}$ and \(
Z_n(S^\-):= \Pi_{\psi}(A_n \mid S^\-) \in [0,1]
\), where $Z_n(S^\-)$ denotes the conditional probability of the set $A_n$ given $S^\-$. Under nuisance contraction condition given in Assumption~\ref{assumption_nuisance}\,(a), we have $Z_n(S^\-) \xrightarrow{\,\bbP_{S^{\-}}\,} 0$ for any sequence $M_n \to \infty$. Since $Z_n(S^\-)$ is itself random through $S^\-$ and by definition
$0 \leq Z_n(S^\-) \leq 1$, boundedness implies uniform integrability, hence $\bbE_{S^\-}\{Z_n(S^\-)\} \to 0$ (convergence in $\bbL_1(\bbP_{S^\-})$-sense), equivalently, $\bbP_{\psi}(A_n) \to 0$. Thus, $\bbP_{\psi}\!\{\|\psi(\bX)-\psi^*(\bX)\|_{\bbL_2(\bbP_{\bX})} \ge M_n \epsilon_n\} \to 0$, as desired. Finally, specializing the generic result with $(\psi,\epsilon_n)=(\underbar{\it m}_t,\varepsilon_{m,n})$ and $(\psi,\epsilon_n)=(\underbar{\it r}_t,\varepsilon_{r,n})$ for $t\in\{0,1\}$ yields the corresponding statements in Lemma~\ref{lemma_nuisance_convergence_rates}.
\end{proof}

\phantomsection
\addcontentsline{toc}{subsection}{Proof of Lemma~\ref{lemma_pvar_convergence}}
\begin{proof}[Proof of Lemma~\ref{lemma_pvar_convergence}.]\label{proof_lemma_pvar_convergence}
Following the construction of the final DRDB posterior $\Pi_{\mu_1}^\CF$ in \eqref{eqn_CF_version_mu}, the posterior variance of $\Pi_{\mu_1}^\CF$ can be computed explicitly. As in the proof of Corollary~\ref{cor_pmean_convergence}, we focus on the variance $\hat{c}_n^{2, (k)}(\m,\r)$ of the DRDB posterior $\Pi_{\mu_1}^{(k)}$ based on one test split $\calD_k$ since the arguments are identical across $k$ by construction of the DRDB procedure.

Fix $k \in \{1,\dots,\bbK\}$ and consider the test-training pair $(S, S^\-) = (\calD_k, \calD_k^\-)$ with index sets $(\calI, \calI^\-)$ and $n_S := |\calI|$. Let $S_1$ denotes the treated group in $S$ with index set $\calI_1$ and $n_1 = |\calI_1|$ as in Section~\ref{sec_methodology}. For notational simplicity, we drop the superscript $(k)$ and write $\hat{c}_n^2(\m, \r) \equiv \hat{c}_n^{2, (k)}(\m, \r)$.

Let $\m \sim \Pi_{m_1}(\cdot; S^\-)$ and $\r \sim \Pi_{r_1}(\cdot; S^\-)$ be nuisance posterior draws and define $W(\bZ,\m, \r) := \r(\bX)\{Y - \m(\bX)\}$ where $\bZ = (Y, \bX)$. By the construction of the DRDB posterior $\Pi_{\mu_1}(\cdot; S)$ and the law of total expectation, the posterior variance $\hat{c}_n^2(\m, \r)$ admits and explicit expression:
\begin{align*}
    \hat{c}_n^2(\m, \r)
    & ~:=~ \frac{n_S}{n_S-2} c_S^2 + \frac{n_1}{n_1-2} c_1^2 ~=~ \lambda_n c_S^2 \,+\, \lambda_{n_1} c_1^2 \\
    & ~=~ \lambda_n \frac{1}{n_S(n_S-1)} \sum_{i=1}^{n_S} \left\{\m(\bX_i) - \eta_{m_1}\right\}^2 \\
      & ~~~~~  \, + \, \lambda_{n_1} \frac{1}{n_1(n_1-1)} \sum_{i \in \calI_1} \left\{ W(\bZ_i,\m, \r) - \eta_1 \right\}^2,
\end{align*}
where $\lambda_n = n_S/(n_S-2)$ and $\lambda_{n_1} = n_1/(n_1-2)$  with $n_1 = \sum_{i=1}^{n_S} T_i$.

First, observe that $\lambda_n \to 1$ as $n \to \infty$. Further, we can write $\lambda_{n_1} = \ptilde_n / (\ptilde_n - 2/n_S)$ where $\ptilde_n := n_1 / n_S$. Since $\ptilde_n \to p_1 = \bbP(T=1) > 0$ as $n \to \infty$, it follows that $\lambda_{n_1} \to 1$ in $\PD$-probability. As a result, by the continuous mapping theorem (CMT), it suffices to establish the desired result for $\hat{c}^2: = c_S^2 + c_1^2$ in place of $\hat{c}_n^2(\m, \r)$.

Given $\mu_1^* \equiv \mu_1(m_1^\d, r_1^\d)$ denoting the posterior mean of the limiting Normal distribution in Theorem~\ref{thm_BvM_mu1}, the corresponding variance $c^2(m_1^\d, r_1^\d)$ is calculated as:
\begin{align*}
c^2(m_1^\d, r_1^\d) & \, = \, \Var_S(\mu_1^*) \, = \, \Var_S\left(\!\bbP_{n_S}(m_1^\d(\bX)) + \bbP_{n_S}\!\left[\frac{r_1^\d(\bX)}{p_1}T\{Y - m_1^\d(\bX)\}\right]\right)\\
& ~ \, := \, \frac{\sigma^2_m}{n_S} + \frac{\sigma^2_1}{n_S},
\end{align*}
where the covariance term vanishes under the NUC condition in Assumption~\ref{assumption_nuisance}.

Towards showing $n_S|\hat{c}^2 - c^2(m_1^\d, r_1^\d)| = \op(1)$, by the same conditioning argument employed in the proof of Corollary~\ref{cor_pmean_convergence} together with applying the DCT or Lemma S4.6 of \cite{sert2025}, it suffices to establish conditional convergence, that is, we show $\bbP_S(n_S|\hat{c}^2 - c^2(m_1^\d, r_1^\d)| \mid \r, \m) \to 0$ in probability under $\bbP_{(\m, \r)}$.

Firstly, the triangle inequality gives $n_S|\hat{c}^2 - c^2(m_1^\d, r_1^\d)| \leq | n_S c_S^2 - \sigma_m^2 | + | n_S c_1^2 - \sigma_1^2 | := T_S + T_1$. Therefore, the problem reduces to showing that both $T_S$ and $T_1$ converge to zero in probability.

\noindent{\bf Analysis of $T_S$.} Define $\widehat{\sigma}_1^2(\m) := \Var_\bX\{\m(\bX)|\m\}$. For any $t>0$, the triangle inequality yields $\bbP_S(T_S > t | \m) \leq \bbP_S( |n_S c_S^2 - \widehat{\sigma}_1^2(\m)| > t/2 | \m ) + \bbP_S( |\widehat{\sigma}_1^2(\m) - \Var_\bX\{m_1^\d(\bX)\}| > t/2 | \m )$.

By construction, $\bbE_S(n_S c_S^2 \mid \m) = \widehat{\sigma}_1^2(\m)$. Applying Chebyshev's inequality yields as $n\to \infty$, $\bbP_S( |n_S c_S^2 - \widehat{\sigma}_1^2(\m)| > t/2 | \m) \leq 4 \Var_S(\{\m(\bX) - \eta_{m_1}\}^2 | \m)/(n_S t^2)$, provided that $\| \m(\bX) \|_{\bbL_4(\bbP_\bX)} = O_{\bbP_{m_1}}(1)$ in Assumption~\ref{assumption_nuisance}.

Moreover, given $\m \sim \Pi_{m_1}$, the variable $\bbP_S\big(|\widehat{\sigma}_1^2(\m) - \Var_\bX\{m_1^\d(\bX)\}| > t/2 \mid \m \big)$ is the indicator of the event where $|\widehat{\sigma}_1^2(\m) - \Var_\bX\{m_1^\d(\bX)\}|>t/2$. Thus, to complete the proof, it suffices to show that the probability of this event converges to zero. We first notice that $|\widehat{\sigma}_1^2(\m) - \Var_\bX\{m_1^\d(\bX)\}| \leq 4\|\m(\bX) - m_1^\d(\bX) \|_{\bbL_2(\bbP_\bX)} \| m_1^\d(\bX) \|_{\bbL_2(\bbP_\bX)} + \|\m(\bX) - m_1^\d(\bX) \|^2_{\bbL_2(\bbP_\bX)}$. Since $\| m_1^\d(\bX) \|_{\bbL_2(\bbP_\bX)} < \infty$, by applying Lemma~\ref{lemma_nuisance_convergence_rates}, we have $\|\m(\bX) - m_1^\d(\bX) \|_{\bbL_2(\bbP_\bX)} = o_{\bbP_{m_1}}(1)$. By the CMT, this implies that $T_S$ converges to zero in probability.

\noindent{\bf Analysis of $T_1 = |n_Sc_1^2 - \sigma^2_1|$.} Let $W(\bZ, \m, \r) := \r(\bX)\{Y - \m(\bX)\}$ for $\bZ = (Y, \bX)$. Then, we write:
\begin{align*}
n_Sc_1^2 ~=~ \frac{1}{\ptilde_n} \frac{1}{(n_1-1)}\sum_{i \in \calI_1} \big[W(\bZ_i,\m, \r) - \eta_{1}\big]^2 ~:=~ \frac{1}{\ptilde_n} \,\mathbb{S}^2_{n_1}.
\end{align*}
By using the construction of $\ptilde_n$ and applying the CMT, we obtain that as $n \to \infty$, $1/\ptilde_n \to 1/p_1$. This result enables us to reduce the analysis of $T_1$ to focus solely on $\mathbb{S}^2_{n_1}$.

\noindent For notational simplicity, we omit explicit conditioning on $\m$ and $\r$. Define $\calT_n := (T_1, \dots, T_{n_S})$. We then observe that by the tower property of expectation,
\begin{align*}
\bbE_S(\mathbb{S}^2_{n_1} | \m, \r) & \, \equiv \, \bbE_S\big(\mathbb{S}^2_{n_1} \big) \, = \, \bbE\{\bbE(\mathbb{S}^2_{n_1} \mid \calT_n)\} \\
& \, = \, \bbE\!\left(\!\bbE\!\left[ \frac{1}{n_1-1} \!\sum_{i \in\calI_1} \big\{ W(\bZ_i,\m, \r) - \eta_1 \big\}^2 \mid  \calT_n \right] \right).
\end{align*}

Note that given $\m$ and $\r$, conditioning on $\calT_n$, the random variables $n_1$ and $\calI_1$ are fixed, and the randomness of $\left\{ W(\bZ_i,\m, \r) : i \in \calI_1 \right\}$ arises only from $\bZ = (Y, \bX)$; also, these random variables are i.i.d.. Following this observation, the inner expectation can be computed as $\bbE(\mathbb{S}^2_{n_1} | \calT_n) = \bbE[\sum_{i \in \calI_1} W(\bZ_i,\m, \r)^2 - n_1\eta_{1}^2 \mid \calT_n ]/(n_1-1)$.

Let $\theta_1: = \bbE[W(\bZ, \m, \r) | T=1]$ ~and~ $q_1:= \bbE[W(\bZ,\m, \r)^2 | T=1]$. Then, we observe that:
\begin{align*}
    \bbE\Big[\!\sum_{i \in \calI_1} W(\bZ_i,\m, \r)^2 | \calT_n\Big] & \, = \,  \sum_{i \in \calI_1} \!\bbE\left[ W(\bZ_i,\m, \r)^2 | \calT_n \right] \\
    & \, = \, \sum_{i \in \calI_1} \bbE\left[ W(\bZ_i,\m, \r)^2 | T_i=1 \right] \, = \, n_1 q_1.
\end{align*}
Also, by using the definition of $\eta_1$, we calculate that:
\begin{align*}
    \bbE\left[\eta_{1}^2| \calT_n \right] & ~=~ \frac{1}{n_1^2} \bbE\bigg[\sum_{i \in S_1} W(\bZ_i,\m, \r)^2 | \calT_n \bigg] + \frac{1}{n_1^2}  \bbE\bigg[\sum_{ i \neq j} W(\bZ_i,\m, \r)  W_j(\m, \r)  | \calT_n \bigg] \\
    & ~=~ \frac{1}{n_1} \bbE[W(\bZ,\m, \r)^2 | T = 1] + \frac{(n_1 -1)}{n_1} \bbE[W(\bZ_1, \m, \r)  W(\bZ_2,\m, \r) | \calT_n] \\
    & ~=~  \frac{q_1}{n_1} + \frac{(n_1 -1)}{n_1} \theta_1^2.
\end{align*}
Then, combining these derivations above, we obtain that:
\begin{align*}
    \bbE[\mathbb{S}_{n_1} \mid \calT_n] ~=~ \frac{1}{n_1-1}\left[n_1q_1 - n_1\left(\frac{q_1}{n_1} + \frac{(n_1 -1)}{n_1} \theta_1^2 \right) \right] ~=~ (q_1 -\theta_1^2) ~:=~ \underline{\sigma}_1^2,
\end{align*}
where $\underline{\sigma}_1^2:= \Var(W(\bZ, \m, \r) | T=1)$. Given $\m$ and $\r$, since $\underline{\sigma}_1^2$ is deterministic, by the iterated expectation, we obtain $\bbE[\mathbb{S}^2_{n_1}] = \underline{\sigma}_1^2$. By Assumption~\ref{assumption_nuisance}, specifically $\| W(\bZ, \m, \r) \|_{\bbL_4(\bbP_\bZ)} = O_{\bbP_{(\m, \r)}}(1)$, it follows that $\mathbb{S}^2_{n_1} \to \underline{\sigma}_1^2$ in probability. Finally, by the CMT, and conditioning on $\m$ and $\r$, we obtain $n_S c_1^2 \to \underline{\sigma}_1^2 / p_1$ in probability.

Next, by the triangle inequality, we can bound $T_1$ as $$
T_1 \leq  \left| n_Sc_1^2 - \frac{\underline{\sigma}_1^2}{p_1} \right| + \left| \frac{\underline{\sigma}_1^2}{p_1} -  \sigma^2_1\right| := T_{11} + T_{12}.
$$
Since we already established that $T_{11} \to 0$ in probability in our earlier analysis, it remains to show that $T_{12} \to 0$ in probability to conclude the desired result for $T_1$.

Towards this goal, we focus on $\sigma^2_1$. Define $W(\bZ,m_1^\d, r_1^\d): = r_1^\d(\bX)T\{Y - m_1^\d(\bX)\}/p_1$ for $\bZ = (Y, \bX)$. By Assumptions~\ref{assumptions_standard_causal_assump}--\ref{assumption_nuisance}, it follows that $\bbE_S\{W(m_1^\d, r_1^\d)\} = 0$. Hence,
\begin{align*}
    \sigma^2_1 & = \Var_S\{W(\bZ, m_1^\d, r_1^\d)\} =  \bbE_S\{W^2(\bZ, m_1^\d, r_1^\d)\} \, = \, \bbE_S\!\left[\frac{\{r_1^\d(\bX)\}^2}{(p_1)^2}T^2\{Y - m_1^\d(\bX)\}^2\right] \\
    & = \bbE_{\bZ|T = 1}\left[\{r_1^\d(\bX)\}^2\{Y - m_1^\d(\bX)\}^2 \mid T = 1\right]/p_1.
\end{align*}
where the final step follows from the law of iterated expectations and by noting that $\bbP(T \!=\! 1) \!=\! p_1$. Moreover, since $\bbE_S\{W(\bZ, m_1^\d, r_1^\d)\} = 0$ and $p_1>0$, the law of iterated expectations gives $\bbE[r_1^\d(\bX)\{Y - m_1^\d(\bX)\} | T = 1] = 0$. Thus, $M\bbE[r_1^\d(\bX)\{Y - m_1^\d(\bX)\}| T = 1] = 0$ for any scalar $M < \infty$. Therefore, we obtain that:
\begin{align*}
    \sigma^2_1 & =\frac{1}{p_1}\bbE\left[\{r_1^\d(\bX)\}^2\{Y - m_1^\d(\bX)\}^2 \mid T = 1\right] + \frac{1}{p_1} (\bbE[r_1^\d(\bX)\{Y - m_1^\d(\bX)\} \mid T = 1])^2 \\
    & = (p_1)^{-1} \Var[r_1^\d(\bX)\{Y - m_1^\d(\bX)\} \mid T = 1] := (p_1)^{-1} \tau^2_1.
\end{align*}
Hence, we obtain that $T_{12} = |\underline{\sigma}_1^2 - \tau^2_1|/p_1$. Since $p_1 > 0$, to establish the desired result, it suffices to show that $|\underline{\sigma}_1^2 - \tau^2_1| \to 0$ in probability.

Next, define $U(\bZ, m_1^\d, r_1^\d):= r_1^\d(\bX)\{Y - m_1^\d(\bX)\}$ and $U(\bZ,\m, \r):= \r(\bX)\{Y - \m(\bX)\}$ for notational clarity. With these definitions, we can bound $T_{12}$ as follows:
\begin{align*}
    |\underline{\sigma}_1^2 - \tau^2_1| & ~=~ \left|\Var\{ U(\bZ, \m, \r) |T= 1\} - \Var\{U(\bZ, m_1^\d, r_1^\d) |T= 1\} \right|  \\
    & ~\leq~ 2 \|U(\bZ, \m, \r) - U(\bZ, m_1^\d, r_1^\d)\|^2_{\bbL_2(\bbP_{\bbS_1})} \|U(\bZ, \m, \r) - U(\bZ, m_1^\d, r_1^\d)\|_{\bbL_2(\bbP_{\bbS_1})} \\
    & ~ \qquad +~ 4 \|U(\bZ, \m, \r) - U(\bZ, m_1^\d, r_1^\d)\|_{\bbL_2(\bbP_{\bbS_1})} \|U(\bZ, m_1^\d, r_1^\d)\|_{\bbL_2(\bbP_{\bbS_1})}.
\end{align*}
Since $\|U(\bZ, m_1^\d, r_1^\d)\|_{\bbL_2(\bbP_{\bbS_1})} < \infty$ by Assumption~\ref{assumption_nuisance}, it suffices to show that $\|U(\bZ, \m, \r) - U(\bZ, m_1^\d, r_1^\d)\|_{\bbL_2(\bbS_1)} = o_{\bbP_{(\m, \r)}}(1)$. Using the definitions and the triangle inequality,
\begin{align*}
   & \|U(\bZ,\m, \r) - U(\bZ,m_1^\d, r_1^\d)\|_{\bbL_2(\bbP_{\bbS_1})}\\
   & ~~~\equiv~ \left\| \r(\bX)\{Y - \m(\bX)\} - r_1^\d(\bX)\{Y - m_1^\d(\bX)\} \right\|_{\bbL_2(\bbP_{\bbS_1})} \\
   %&\leq \left\| \{\r(\bX) - r_1^\d(\bX)\} \{Y - \m(\bX)\} \right\|_{\bbL_2(\bbP_{\bbS_1})}+ \left\| r_1^\d(\bX)\{\m(\bX) - m_1^\d(\bX)\} \right\|_2 \\
   & ~~~\leq~ \underline{\Gamma}_1 \| \r(\bX) - r_1^\d(\bX) \|_2
         + M_r \|\m(\bX) - m_1^\d(\bX)\|_{\bbL_2(\bbP_{\bbS_1})},
\end{align*}
where the final step follows from the CS inequality and the boundedness of $r_1^\d(\bX)$, i.e., $r_1^\d(\bX) < M_r$ for some $M_r < \infty$ by Assumption~\ref{assumptions_standard_causal_assump}. Here, $\underline{\Gamma}_1^2 := \sup_{x} \bbE \left[ \{Y - \m(\bX)\}^2 \mid X = x \right] = O_{\bbP_{m_1}}(1)$ by Assumption~\ref{assumption_nuisance}.

Given that $\| \r(\bX) - r_1^\d(\bX) \|_2 = o_{\bbP_{r_1}}(1)$ and $\| \m(\bX) - m_1^\d(\bX) \|_2 = o_{\bbP_{m_1}}(1)$ by Lemma~\ref{lemma_nuisance_convergence_rates}, we establish the conditional convergence of $T_1$. Finally, by applying the DCT or Lemma~S4.6 of \cite{sert2025}, the desired result follows.
\end{proof}

\phantomsection
\addcontentsline{toc}{subsection}{Proof of Corollary~\ref{cor_pvar_convergence_for_ATE}}
\begin{proof}[Proof of Corollary~\ref{cor_pvar_convergence_for_ATE}.]
Following the construction of the final DRDB posterior $\pATE^\CF$ in \eqref{eqn_CF_version_mu}, the posterior variance of $\pATE^\CF$ can be computed explicitly. As in the proof of Lemma~\ref{proof_lemma_pvar_convergence}, we focus on the variance $\hat{c}_n^2(\m,\r)^{(k)}$ of the posterior $\pATE^{(k)}$ based on $\calD_k$ since the arguments are identical across $k$ by construction of the DRDB procedure.

Fix $k \in \{1,\dots,\bbK\}$ and consider the test-training pair $(S, S^\-) = (\calD_k, \calD_k^\-)$ with index sets $(\calI, \calI^\-)$ and $n_S := |\calI|$. Let $(S_1, S_0)$ denotes the treated and control subgroups in $S$ with index sets $(\calI_1, \calI_0)$ with $n_1 = |\calI_1|$ and $n_0 = |\calI_0|$ as in Section~\ref{sec_methodology}. For notational simplicity, we drop the superscript $(k)$ and write $\hat{c}_n^2(\bm, \br) \equiv \hat{c}_n^2(\bm, \br)^{(k)}$. Using the construction of $\pATE$ in Section \ref{sec_DRDB_extended_ATE}, we have:
\begin{align*}
  \hat{c}^2(\bm, \br) & ~=~ \frac{n_S}{(n_S-2)}\frac{\sum_{i \in \calI}\{\bm(\bX) - \eta_m\}^2}{n_S(n_S-1)} \, + \, \frac{n_1}{(n_1-2)}\frac{\sum_{i \in \calI_1}\{W(\bZ_i, \m,\r) - \eta_1\}^2}{n_1(n_1-1)}  \\
  & \qquad + \, \frac{n_0}{n_0-2}\frac{\sum_{i \in \calI_0}\{U(\bZ_i, \ubm,\ubr) - \eta_0\}^2}{n_0(n_0-1)} \\
  & ~:=~ \lambda_n c_S^2 + \lambda_{n_1} c_1^2 + \lambda_{n_0} c_0^2,
\end{align*}
where $\bm(\cdot) = \m(\cdot) - \ubm(\cdot)$, $W(\bZ, \m, \r) := \r(\bX)\{Y - \m(\bX)\}$ and $U(\bZ, \ubm, \ubr) := \ubr(\bX)\{Y - \ubm(\bX)\}$ and $\lambda_n = n_S/(n_S-2)$ and $\lambda_t = n_t/(n_t-2)$ for $t = 0,1$ and $\bZ = (Y, \bX)$.

By definition, $\lambda_n \to 1$, $\lambda_t \to 1$ in probability for $t = 0,1$. Hence, the problem reduces to showing the desired result for $(c_S^2 + c_1^2 + c_0^2)$ in place of $\hat{c}^2(\bm, \br)$.

We next compute the variance $c^2(m^\d, r^\d)$ of the limiting Normal distribution in Theorem~\ref{thm_BvM_ATE}. Note that $c^2(m^\d, r^\d)$ is equal to the variance of the mean $\widehat\rATE(m^\d,r^\d)$ of the limiting Normal distribution. Recall that for $\bD = (Y, T, \bX)$ and $m^\d(\cdot):= m_1^\d(\cdot) - m_0^\d(\cdot)$, the posterior mean can be written as
\begin{align*}
    \widehat \rATE(m^\d,r^\d) & = \bbP_{n_S}\{m^\d(\bX)\} \, + \, \bbP_n\!\left[\frac{r_1^\d(\bX)}{p_1}T\{Y - m_1^\d(\bX)\} - \frac{r_0^\d(\bX)}{1- p_1}(1-T)\{Y - m_0^\d(\bX)\} \right] \\
    & ~=:~ \bbP_{n_S}\{m^\d(\bX)\} \, + \, \bbP_n\{\phi_1(\bD) - \phi_0(\bD)\} =: \bbP_{n_S}\{m(\bX)\} \, + \, \bbP_{n_S}\{\phi(\bD)\}.
\end{align*}
Thus, we can explicitly calculate $c^2(m^\d, r^\d)$ as\\ $c^2(m^\d, r^\d) = [\Var_S\{m^\d(\bX)\} + \Var_S\{\phi(\bD)\} + 2\Cov_S(m^\d(\bX), \phi(\bD))]/n_S = (V_1 + V_2 + 2V_3)/n_S$.

Note that
\(
    V_2 = \Var_S\{\phi_1(\bD)\} + \Var_S\{\phi_0(\bD)\} + 2\Cov_S\{\phi_1(\bD), \phi_0(\bD)\} =: \sigma^2_1 + \sigma^2_0
\),
where $\Cov_S\{\phi_1(\bD), \phi_0(\bD)\} = 0$ by using the tower property and $\bbE\{\phi_t(\bD)| \bX\} = 0$ for $t = 0,1$ from Assumption \ref{assumptions_standard_causal_assump}. Also, by definition, note that $\phi_1(\bD)\phi_0(\bD)=0$ since $T(1-T)=0$.

By definition,
\(
    V_3 = \bbE\{m^\d(\bX)\phi(\bD)\} - \bbE\{m^\d(\bX)\}\bbE\{\phi(\bD)\} = 0
\),
as $\bbE\{\phi_t(\bD)| \bX\} = 0$ by Assumption \ref{assumptions_standard_causal_assump} for $t = 0,1$. Thus, defining $\sigma^2_m:= \Var_S\{m^\d(\bX)\}$, we have $c^2(m^\d, r^\d) = (\sigma^2_m + \sigma_1^2 + \sigma_0^2)/n_S$. Combining these decompositions yields the following bound:
\[
n_S|c_S^2 + c_1^2 + c_0^2 - c^2(m^\d, r^\d)| \, \le \, |n_Sc_S^2 - \sigma^2_m| + |n_Sc_1^2 - \sigma_1^2| + |n_Sc_0^2 - \sigma_0^2| \, := \, T_S + T_1 + T_0.
\]

Since, by using Assumption~\ref{assumption_nuisance} and the definition of $\bm (\cdot) = \m(\cdot) - \ubm(\cdot)$, the triangle inequality yields $\|\bm(\bX)\|_{\bbL_4(\bbP_\bX)} \leq \|\m(\bX)\|_{\bbL_4(\bbP_\bX)} + \|\ubm(\bX)\|_{\bbL_4(\bbP_\bX)} = \Op(1)$, the analysis of $T_S$ analogously follows from the proof of Lemma~\ref{lemma_pvar_convergence}, replacing the pair $\{\m(\cdot), m_1^\d(\cdot)\}$ with the pair $\{\bm(\cdot), m^\d(\cdot)\}$; for brevity, we refer to that part of the proof. Thus, applying Lemma~S4.6 of \cite{sert2025} and Lemma~\ref{lemma_nuisance_convergence_rates}, we obtain $T_S = \op(1)$.

Similarly, the analysis of $T_1$ is identical to that in the proof of Lemma~\ref{lemma_pvar_convergence}; to avoid repetition, we refer to that proof for details. By using Assumption~\ref{assumptions_standard_causal_assump} (in particular, $p_1>0$) and Assumption~\ref{assumption_nuisance}, and applying Lemma~\ref{lemma_nuisance_convergence_rates} together with Lemma~S4.6 of \cite{sert2025}, we obtain $T_1 = \op(1)$. By the same arguments applied symmetrically to the control group, we skip the details for brevity and correspondingly conclude $T_0 = \op(1)$, which gives the desired result.
\end{proof}

\end{document}